\pdfoutput=1
\documentclass{article}
\usepackage[utf8]{inputenc}
\usepackage[margin=1in,vmarginratio=1:1]{geometry}
\usepackage{setspace}
\usepackage{natbib}
\setcitestyle{round}
\usepackage{enumitem}
\usepackage{float, soul, nicefrac}
\usepackage{multicol}
\usepackage{graphicx}
\usepackage{amsmath, dsfont}
\usepackage{amsthm, xcolor, mdframed}
\usepackage{amssymb}
\usepackage{bigints}
\usepackage{bm}
\usepackage[justification=centering]{caption}
\usepackage{subcaption, comment}
\usepackage{bbold}
\usepackage{hyperref}
\usepackage{mathtools}
\usepackage{ tikz}
\usetikzlibrary{automata,positioning}
\usetikzlibrary{decorations, decorations.text,backgrounds}
\usetikzlibrary{decorations.pathreplacing}
\usetikzlibrary{shapes, arrows, arrows.meta}
\usetikzlibrary{shapes.multipart}
\usetikzlibrary{calc}
\usetikzlibrary{fit}
\newtheorem{definition}{Definition}[section]

\newtheorem{assumption}{Assumption}
\newtheorem{proposition}{Proposition}

\newtheorem{lemma}{Lemma}
\newtheorem{corollary}{Corollary}

\newtheorem{namedtheorem}{Theorem}
\theoremstyle{remark}
\newtheorem{remark}{Remark}
\theoremstyle{plain}
\newtheorem*{theorem*}{Theorem}

\makeatletter
\newtheorem*{repeatednamedtheorem}{\repeatedtheoremname}
\newcommand{\repeatedtheoremname}{}
\newenvironment{repeatnamedtheorem}[2][]
  {%
    \renewcommand{\repeatedtheoremname}{%
      Theorem~\ref{#2}%
      \if\relax\detokenize{#1}\relax
      \else\ (#1)%
      \fi
    }%
    \begin{repeatednamedtheorem}%
  }
  {\end{repeatednamedtheorem}}
\makeatother

\title{Interacting Principal-Agent Games%
\thanks{Earlier versions of this paper circulated under the title
``Existence of Equilibrium Mechanisms in Generalized Principal-Agent Problems with Interacting Teams.''
I am grateful to Pierre Boyer, John Duggan, Christian Ewerhart,
Navin Kartik, Wojciech Olszewski, Philip J. Reny, and Ron Siegel
for helpful comments and discussions on earlier drafts.
I also thank participants in seminars at the Rochester Department of Political Science,
the Purdue Mechanism Design Workshop, the SAET 2025 Annual Meeting,
and the University of South Carolina Game Theory Conference for helpful comments.
Any remaining errors or omissions are my own.}}

\author{Brian Roberson%
\thanks{\href{mailto:brobers@purdue.edu}{brobers@purdue.edu}
$\bullet$ Purdue University, Department of Economics}}
\date{November 2026}

\begin{document}

\begin{titlepage}
    \maketitle
    \begin{abstract}
        \noindent
In many economic environments, multiple principals simultaneously design
incentive systems while interacting strategically with one another. The
resulting incentive constraints are interdependent: whether a mechanism is
incentive compatible for one principal may depend on the mechanisms chosen
by the other principals. This interdependence makes exact equilibrium
fragile: small changes in one principal's mechanism can alter the
incentive-compatible mechanisms available to others.

We establish two complementary equilibrium-existence results. First, for
environments with multidimensional type, action, and outcome spaces, we
show that, for every \(\varepsilon>0\), there exists a Bayesian--Nash
principals' equilibrium under \(\varepsilon\)-relaxed incentive
compatibility: principals best respond exactly over mechanisms whose
agents' ex ante expected deviation regret is at most \(\varepsilon\). Thus
equilibrium existence can be recovered with arbitrarily small expected
incentive regret.
In contrast to Myerson's quasi-equilibrium, which keeps incentive
compatibility exact but relaxes the principals' optimality requirement,
this concept keeps principals' optimization exact and relaxes the
agents' incentive constraints. Second, motivated by Myerson's nonexistence example, in which each principal has a single privately informed agent, we examine a one-member team-contest specialization of the general model with voluntary participation. This specialization retains the basic structure of Myerson's setting---competing principals, privately informed agents, and interdependent incentive constraints---while imposing primitive conditions under which exact feasibility varies continuously. In this environment, we establish existence of an exact Bayesian--Nash principals' equilibrium.    \end{abstract}

    \vspace{1em}

    \vspace{0.5em}

    \thispagestyle{empty}
\end{titlepage}
\section{Introduction}
Economic organizations rarely choose incentive systems in isolation. Firms compete in markets while addressing internal agency concerns,\footnote{See for example \citet{aspremont88}, \citet*{kamien92}, \citet{raith03}, and \citet*{anton2023} among others.} organizations coordinate production along supply chains while managing incentive relationships within and across firms, and teams compete for prizes while allocating rewards among team members.\footnote{See for example \citet{nitzan1991collective}, \citet{esteban2001collective}, \citet{nitzan2011prize}, \citet{nitzan2014intra}, \citet{balart2016strategic}, and \citet*{konishi2024allocation} among others.} Although these environments differ in their economic details, they share a common feature: multiple principals simultaneously design incentive systems while interacting strategically with one another. We model such environments as \emph{interacting principal-agent games}, in which each principal designs a mechanism governing incentive relationships within its own team while accounting for strategic interactions across teams.\footnote{Within-team incentive design has been studied in settings involving both adverse selection and moral hazard; see \citet{mcafee1991optimal}. Also related are the literatures on teamwork (e.g., \citealp{admati1991joint}; \citealp{yildirim2006getting}; \citealp{bonatti2011collaborating}; \citealp{georgiadis2015projects}; \citealp*{bowen2019collective}; \citealp*{cetemen2020uncertainty}; \citealp{ozerturk2021credit}; and \citealp{yildirim2023fares}) and contracting with teams (e.g., \citealp{holmstrom1982moral}; \citealp{itoh1991incentives}; \citealp{che2001optimal}; \citealp{winter2004incentives}; and \citealp*{halac2021rank}). Team production can generate substantial gains through complementarities among agents' skills and efforts. For more on this issue, see \citet{candougan2024implications}, who examines innovation contests involving a finite number of individuals and/or teams with complete information and exogenous prize-sharing rules and identifies conditions under which teams outperform individuals.}

A central question is whether such environments admit an equilibrium in
incentive systems. Unlike an isolated mechanism design problem, incentive
compatibility is interdependent: the set of incentive-compatible mechanisms
available to each designer depends on the mechanisms chosen by the other
designers. Changes in one designer's mechanism may therefore affect the
reporting incentives faced by team members and alter the incentive-compatible
mechanisms available to other designers. This dependence of each designer's set
of incentive-compatible mechanisms on the choices of the other designers is the
defining feature of a generalized game, a framework introduced by
\citet{debreu1952social} to study strategic environments in which each player's
feasible strategy set may depend on the strategies chosen by the other players.
These strategic interdependencies create significant challenges for establishing
equilibrium existence.

This difficulty is illustrated by a classic example in
\citet{myerson1982optimal}.\footnote{In particular, see Section~4 of
\citet{myerson1982optimal}, entitled ``Equilibria among Several
Principals.''} There, a principals' equilibrium in Myerson's sense fails to exist because the
correspondence assigning each principal its incentive-compatible
mechanisms, given the other principals' mechanisms, is not lower
hemicontinuous.\footnote{For more on the role of lower hemicontinuity in
generalized games, see \citet{tobias2022equilibrium}.} Intuitively, a
mechanism that is incentive compatible for one configuration of the other
principals' mechanisms may cease to be incentive compatible after an
arbitrarily small change in those mechanisms. Thus
incentive-compatible choices need not vary continuously with the
mechanisms chosen by the other principals, and exact best responses can
fail to survive small perturbations. This fragility is not peculiar to
Myerson's example; it is a general obstacle for establishing existence of equilibrium in interacting principal-agent
games.\footnote{Classical existence results for generalized games
rely on continuity of the feasible-strategy correspondence, with lower
hemicontinuity playing a central role. For additional details, see
\citet{debreu1982existence}, \citet[Chapter~19]{border1985fixed},
\citet{facchinei2010generalized}, \citet{dasgupta2015debreu}, and
\citet{tobias2022equilibrium}. Most closely related is
\citet{tao2024generalized}, who establish the existence of
pure-strategy Bayesian Nash equilibrium in generalized games with
uniformly continuous feasible-strategy correspondences.} In this paper, we develop a framework for establishing equilibrium in
interacting principal--agent games with interdependent incentive
constraints. Under broad conditions, we show that, for every
\(\varepsilon>0\), principals best respond exactly over mechanisms whose
agents' ex ante expected deviation regret is at most \(\varepsilon\). We
also establish exact equilibrium in an embedded one-member quasi-linear
team contest with voluntary participation when contest shares are
monotone in outputs, concave and uniformly Lipschitz in own output, and
a uniformly profitable pooling output exists.  This second result returns to the one-agent-per-principal
architecture underlying Myerson's example and asks which additional
economic structure restores exact equilibrium.  The contest and pooling
conditions provide that structure by preventing the discontinuity of the
exact feasible-mechanism correspondence that drives Myerson's
nonexistence result.

In the related context of standard Bayesian games, distributional
approaches to existence, as developed by
\citet{milgrom1985distributional} and applied to single-principal--agent
environments by \citet{kadan2017}, provide a powerful framework in which
induced laws can serve as the relevant choice objects. The generalized
game studied here, with interacting principals, retains this law-based
framework but introduces endogenous incentive feasibility. Each principal's feasible choices depend on which of its
mechanisms are incentive compatible against the mechanisms chosen by the
other principals. Each principal's mechanism can still be identified
with its own induced law over reported types and team decisions, and we
use exactly this identification. What changes is the feasible set:
whether a mechanism is incentive compatible depends on the other
principals' mechanisms, so the principals' problems cannot be solved one
at a time over fixed sets of represented outcomes. We therefore keep the
product structure of the mechanism strategy space and the endogenous
incentive constraints explicit. On this
space, we first establish a general existence theorem under
\(\varepsilon\)-relaxed incentive compatibility. We then turn to the special case that motivates Myerson's nonexistence example, in which each principal has a single privately informed agent. We specialize the decision space to output and prize-share choices in a one-member contest. This retains the central strategic structure of Myerson's setting while imposing primitive conditions under which exact feasibility varies continuously. An envelope representation, concavity of the contest technology, and a strictly profitable pooling output restore the continuity needed for an exact-equilibrium existence result. A conditional-mean reduction shows that the resulting equilibrium remains optimal against all stochastic mechanism deviations in the specialized primitive mechanism game.

The remainder of the paper is organized as follows. Section~2 discusses
the paper's relationship to the existing literature. Section~3 revisits
Myerson's nonexistence example and contrasts Myerson's quasi-equilibrium
with the \(\varepsilon\)-relaxed IC equilibrium studied here. Section~4
presents the general model. Section~5 develops the mechanism strategy
spaces and the topological assumptions used in the existence theorem.
Section~6 establishes existence of
\(\varepsilon\)-relaxed Bayesian--Nash principals' equilibria and applies
the theorem to Myerson's finite-type example. Section~7 develops
the exact-equilibrium result for the embedded one-member quasi-linear
team contest. Section~8 concludes. Appendices~A and~B contain the proofs
of the approximate and exact results, respectively.

\section{Related Literature}
Our approach connects to three related literatures. The first is the
literature on quasi- and approximate-equilibrium concepts, including
Myerson's quasi-equilibrium approach to multi-principal mechanism-design
games. The second studies distributional and representation-based
approaches to Bayesian games and principal--agent problems. The third is
the literature on team contests with endogenous prize sharing, which
provides a natural class of well-developed applications for our framework.

A related literature studies the existence of
\(\varepsilon\)-equilibria in games where exact equilibrium may be
difficult to obtain directly. Approximate equilibria have long served as
a useful intermediate concept when discontinuities or failures of
compactness prevent a direct equilibrium existence argument. Our use of
\(\varepsilon\)-Bayesian--Nash principals' equilibria differs from this
tradition in an important respect. In the settings studied here, the
central difficulty is not discontinuity of payoffs on a fixed strategy
space, but the fact that each principal's feasible strategy set is
itself endogenous through incentive compatibility. In this respect, our
framework is closer in spirit to the generalized-game literature than to
the standard literature on discontinuous games. It is complementary to
the better-reply security framework of \citet{reny1999existence}, which
addresses discontinuous payoff functions on fixed strategy spaces,
whereas we study discontinuities arising from endogenous
incentive-compatible feasible strategy correspondences.\footnote{Among
the many contributions to the equilibrium-existence literature for
discontinuous games are
\citet{dasgupta1986existence},
\citet{reny1999existence},
\citet{carmona2009existence},
\citet{barelli2013note},
\citet*{mclennan2011games},
\citet{he2016existence},
\citet{bich2017existence},
\citet{carbonell2018existence},
\citet{reny2020nash},
\citet{olszewski2023equilibrium}, and
\citet{prokopovych2023monotone}. Our approximate equilibrium result is
therefore motivated by endogenous feasibility rather than payoff
discontinuity.}

Various literatures use \(\varepsilon\)-equilibrium concepts to handle
settings in which exact equilibrium is difficult to characterize or may
fail to exist. For example, \(\varepsilon\)-equilibria have been used to
study the relationship between finitely and infinitely repeated games
\citep{radner1980collusive}, to approximate equilibria through restricted
strategy sets \citep{fudenberg1986limit}, to analyze perturbed games or
games with incomplete common knowledge or beliefs
\citep{monderer1989approximating,kajii1998robustness,jackson2012epsilon},
and to study computational complexity
\citep{goldberg2014query,babichenko2017communication}. Our approximation
is different in emphasis. The main difficulty in interacting
principal--agent games is that each principal's feasible set is generated
endogenously by incentive compatibility and may fail to vary lower
hemicontinuously with the mechanisms chosen by the other principals. We
therefore relax incentive compatibility itself: principals best respond
exactly over mechanisms for which agents' deviation gains are no larger
than \(\varepsilon\).

A separate literature studies these environments as extensive-form games.
\citet*{attar2023equilibrium} show that a perfect Bayesian equilibrium
exists in Myerson's example once principals need not select their
preferred continuation equilibrium in the agents' game.
\citet*{forges2024interacting} study perfect Bayesian equilibrium
outcomes of the extensive-form game in which each principal proposes a
mechanism to its own agents (an interacting principal--agent game in
our terminology, which they call the ``corporations' game''), establish
versions of the revelation principle, and give conditions linking
principals' equilibrium to perfect Bayesian equilibrium outcomes. Our
analysis retains Myerson's generalized-game formulation: incentive
compatibility determines each principal's feasible mechanisms under
truthful and obedient continuation play. Our formulation shares with
\citet*{forges2024interacting} the informational structure in which
agents observe only their own principal's mechanism, and their belief
restriction---agents in a deviating corporation believe that other
corporations remain on the equilibrium path---matches the conjectures
implicit in Myerson's formulation. The key difference is continuation
selection: a principals' equilibrium evaluates each deviation under the
deviating principal's preferred continuation, whereas a perfect
Bayesian equilibrium requires only that some sequentially rational
continuation deter it, possibly through untruthful reporting by the
deviator's agents. Our results build on Myerson's principals'
equilibrium concept for the generalized game rather than on perfect
Bayesian equilibrium of the extensive-form game.

The second strand, initiated by \citet{milgrom1985distributional},
represents strategies by the probability laws they induce over types and
actions rather than by the underlying strategy functions themselves.\footnote{Note that the senior editors at
\textit{Mathematics of Operations Research} recently included this paper,
see \citet{Schein26}, as one of 50 papers representing the journal's first
fifty years.}
Subsequent work, including \citet{Balder1988}, developed this
strategic-measure methodology as a tool for obtaining compactness and
continuity in games with infinite-dimensional strategy spaces. We draw on
the same topological insight: probability measures and stochastic kernels
often have better compactness properties than spaces of measurable
strategy functions. At the same time, interacting principal--agent games
require us to retain the product structure of the principals' mechanism
strategy spaces, because incentive compatibility for one principal is
defined relative to the mechanisms chosen by the other principals.

This representation-based approach was extended to mechanism design by
\citet*{kadan2017}. In the single-principal setting they represent each
mechanism by the outcome law induced under truthful and obedient play,
thereby replacing the primitive mechanism space with a compact
represented space while preserving the principal's optimization problem. Our setting differs because each principal's payoff and incentive
constraints are generated by the full profile of mechanisms. In
particular, incentive compatibility compares truthful behavior with the
payoffs generated by admissible unilateral deviations, and these
comparisons depend on the mechanisms chosen by the other principals. We
likewise use a law-based representation: each mechanism is identified
with its opponent-independent fixed-marginal type--decision law.
The difference is that the set of incentive-compatible laws for one
principal moves with the other principals' choices, so the analysis
must establish continuity properties of these endogenous feasible sets
rather than optimize over a fixed represented set.

Our methodological contribution is therefore not the representation
itself but the treatment of endogenous feasibility on it. We extend
Myerson's positive-tolerance construction to infinite-dimensional
mechanism spaces by expressing incentive feasibility through closed
fixed-deviation constraints from auxiliary classes that recover the full
deviation value, establishing upper semicontinuity of maximal regret in
opponents' mechanisms at each fixed own mechanism, and using an exactly
incentive-compatible seed, feasible against every opponents' profile, to
obtain lower hemicontinuity of the relaxed feasible sets for every
positive tolerance.

Our framework is also related to the literature on contests with
endogenous prize sharing.\footnote{Also related is the strand of
literature on team contests with deterministic team production, a public
good prize, and private information. This includes
\citet{barbieri2014group},
\citet{barbieri2016private},
\citet{eliaz2018simple},
\citet{barbieri2019group}, and
\citet{barbieri2021private}. See also
\citet{brookins2016equilibrium}, who, for the case in which each team's
winnings take the form of a public good, establishes equilibrium
existence for a range of information and production configurations. Also
related is \citet{nitzan2018}, which studies the related issue of
cost-sharing rules.} This literature has progressively expanded the
admissible class of prize-sharing rules, beginning with restricted
linear sharing rules
(\citet{nitzan1991collective},
\citet{nitzan2011prize}, and
\citet*{balart2016strategic}),
extending to more general sharing rules
(\citet{trevisan2020optimal},
\citet{simeonov2020individual}, and
\citet{kobayashi2021effort}),
and more recently allowing environments with multiple prizes
(\citet*{crutzen2020} and
\citet*{konishi2024allocation}).

Our two results address contest environments at different levels of
generality. The \(\varepsilon\)-relaxed theorem applies through the paper's
general mechanism-game framework once its compactness, continuity, and
seed conditions are satisfied. The exact theorem is deliberately more
specialized: it studies an embedded one-member quasi-linear contest and
uses concavity of the contest technology together with a uniformly
profitable pooling output to establish continuity of the exact
budget-feasible correspondence. Thus the exact result is a
contest-specific existence theorem, while the relaxed result applies to
the broader class of interacting principal--agent games developed below.

\section{Myerson (1982) Example}
\label{MyersonEx}

We begin by revisiting the example in Section~4 of
\citet{myerson1982optimal}, which illustrates why equilibrium may fail
to exist when multiple principals simultaneously design incentive
systems. The example shows that the correspondence assigning each
principal its incentive-compatible mechanisms, given the mechanisms
chosen by the other principals, need not be lower hemicontinuous. As a
result, exact principals' equilibrium may fail to exist.

Myerson responds to this difficulty with the notion of
quasi-equilibrium. A quasi-equilibrium keeps incentive compatibility
exact but weakens the principals' optimality requirement: a principal
may forgo a profitable exactly incentive-compatible deviation, provided
that the deviation is unsafe, in the sense that arbitrarily small
changes in the other principals' mechanisms could render it infeasible.

The concept studied in this paper makes the opposite trade-off.
Principals optimize exactly, against every feasible alternative, but
feasibility is defined by a relaxed incentive constraint: agents' ex
ante expected deviation regret may be as large as \(\varepsilon>0\).
Quasi-equilibrium relaxes what the principals must do;
\(\varepsilon\)-relaxed equilibrium relaxes what the agents must
accept. The two concepts are nonetheless closely connected. Myerson's
existence proof passes through fixed points with positive incentive
tolerances, and our construction adapts that step to compact mechanism
spaces. The tolerances are normalized differently: Myerson bounds each
type's probability-weighted deviation gain separately, whereas we bound
the sum of probability-weighted conditional regrets. At the same
numerical tolerance, our bound implies his probability-weighted
type-by-type bounds. By contrast, an unweighted type-by-type regret
bound implies our ex ante bound, but not conversely: the ex ante
restriction allows the tolerance budget to be allocated unevenly across
types. The example below illustrates both concepts.

Consider two competing teams, labeled \( j = 1, 2 \). Each team \( j \) consists of a principal \( p_j \) and a team member \( m_j \). The team member has private information (their ``type''), and the principal designs a mechanism to coordinate the team's actions. In this example, we focus on adverse selection where team members only report their types.

The game unfolds in five stages, which we describe below and which are illustrated in Figure~\ref{fig:mechanism-tree}.

\textit{Stage 0 (Mechanism Design):} Each principal \( p_j \) chooses a coordination mechanism for their team. Specifically, principal \( p_j \) selects a behavior strategy \( \alpha_j(\cdot|t_j') \) that specifies, for each possible type report \( t'_j \in \{\theta_A, \theta_B\} \) from member \( m_j \), a probability distribution over recommended actions from the set \( \{A, B, C\} \).

\textit{Stage 1 (Type Realization):} Nature independently draws a type \( t_j \in \{\theta_A, \theta_B\} \) for each team member \( m_j \), with \( P(t_j = \theta_A) = P(t_j = \theta_B) = \frac{1}{2} \). This corresponds to node \( x_1 \) in Figure~\ref{fig:mechanism-tree}.

\textit{Stage 2 (Type Report):} Each member \( m_j \) privately observes their type \( t_j \) (reaching node \( x_2 \) if \( t_j = \theta_A \) or node \( x_3 \) if \( t_j = \theta_B \)) and then sends a cheap-talk report \( t'_j \in \{\theta_A, \theta_B\} \) to their principal \( p_j \). This report is private within team \( j \): only principal \( p_j \) observes it.

\textit{Stage 3 (Mechanism-Selected Action):}
After observing the report \(t'_j\), principal \(p_j\) uses the
mechanism \(\alpha_j(\cdot\mid t'_j)\) to select an action
\(a_j\in\{A,B,C\}\), possibly stochastically. Although Myerson describes
\(a_j\) as a recommended action, the team member makes no subsequent
action choice: the selected action is implemented directly and enters
payoffs. In Figure~\ref{fig:mechanism-tree}, this corresponds to the
information sets at nodes \((x_4,x_6)\) following report
\(t'_j=\theta_A\) and nodes \((x_5,x_7)\) following report
\(t'_j=\theta_B\).

\textit{Stage 4 (Payoffs):} The game ends with payoffs realized at the terminal nodes. Each team's payoffs depend on three factors: their member's true type \( t_j \), their own recommended action \( a_j \), and the other team's recommended action \( a_{-j} \). In Figure~\ref{fig:mechanism-tree}, principal payoffs are shown in blue and member payoffs in red.

\begin{figure}[htbp]
\begin{center}
\def\myscale{1}
\scalebox{\myscale}[\myscale]{

\vspace{.25in}
\begin{tikzpicture}
      \fontsize{10pt}{0pt}\selectfont

      \draw [dashed] (3,-2) -- (3,2);
      \draw [dashed] (-3,-2) -- (-3,2);

        \draw (0,0) -- (0,-2);
        \draw (0,0) -- (0,2);
         \draw (0,4) to [bend left] (0,0);
        \draw [fill=white,draw=black] (0,0) circle (.25) node {\small $N$} node [anchor=west,xshift=.1in] {\small $x_1$};
        \draw[dotted] (0,4) -- (3,0);
        \draw[dotted] (0,4) -- (-3,0);
        \draw [fill=white,draw=blue] (0,4) circle (.25)
    node (main) {\small $p_j$}
    node [anchor=east,xshift=-.1in] {\small $x_0$}
    (main.east) node [anchor=west,xshift=0.1in] {\small Selection of Mechanism \color{blue}{$\alpha_j(\cdot| t'_j)$}};
    \node at (3,0) [draw=black,fill=white] {\small \color{black}{\color{blue}{$\alpha_j(\cdot|t'_j=\theta_B)$}}};
    \node at (-3,0) [draw=black,fill=white] {\small \color{black}{\color{blue}{$\alpha_j(\cdot|t'_j=\theta_A)$}}};
        \draw(-3,-2) -- (0,-2);
        \draw[red, line width=1mm] (-3,2) -- (0,2);
        \draw[red, line width=1mm]  (3,-2) -- (0,-2);
        \draw (3,2) -- (0,2);
        \draw [fill=white,draw=red] (0,2) circle (.25) node {\small $m_j$} node [anchor=south,yshift=.1in] {\small $x_2$};
        \draw [fill=white,draw=red] (0,-2) circle (.25) node {\small $m_j$} node [anchor=north,yshift=-.1in] {\small $x_3$};
        \node at (1.5,2) [draw=red,fill=white] {\small \color{red}{$t'_j=\theta_B$}};
        \node at (1.5,-2) [draw=red,fill=white] {\small \color{red}{$t'_j=\theta_B$}};
        \node at (-1.5,2) [draw=red,fill=white] {\small \color{red}{$t'_j=\theta_A$}};
        \node at (-1.5,-2) [draw=red,fill=white] {\small \color{red}{$t'_j=\theta_A$}};

        \node at (0,1) [draw=black,fill=white] {\small \color{black}{$t_j=\theta_A$ wp $\frac{1}{2}$}};
        \node at (0,-1) [draw=black,fill=white] {\small \color{black}{$t_j=\theta_B$ wp $\frac{1}{2}$}};

      \def\xPay{$z_j$}
      \def\yPay{0}
      \begin{scope}[xshift=3cm,yshift=-2cm]
        \draw (0,0) -- (2,1)  node [anchor=west] {(\color{red}\xPay\color{black},\color{blue}\yPay\color{black})};
        \def\xPay{1}
      \def\yPay{6}
        \draw (0,0) -- (2,0)  node [anchor=west] {(\color{red}\xPay\color{black},\color{blue}\yPay\color{black})};
              \def\xPay{0}
      \def\yPay{5}
\draw (0,0) -- (2,-1)  node [anchor=west] {(\color{red}\xPay\color{black},\color{blue}\yPay\color{black})};
     \draw (0,0) -- (2,0);
       \node at (1,-.75) [draw=blue,fill=white] {\small \color{blue}{$C$}};
           \node at (1,0) [draw=blue,fill=white] {\small \color{blue}{$B$}};
        \node at (1,.75) [draw=blue,fill=white] {\small \color{blue}{$A$}};
        \draw [fill=white,draw=blue] (0,0) circle (.25) node {\small $p_j$} node [anchor=north,yshift=-.1in] {\small $x_7$};
      \end{scope}

  \def\xPay{1}
      \def\yPay{6}
      \begin{scope}[xshift=3cm,yshift=2cm]
           \draw (0,0) -- (2,1)  node [anchor=west] {(\color{red}\xPay\color{black},\color{blue}\yPay\color{black})};
            \def\xPay{$z_j$}
      \def\yPay{0}
        \draw (0,0) -- (2,0)  node [anchor=west] {(\color{red}\xPay\color{black},\color{blue}\yPay\color{black})};
              \def\xPay{0}
      \def\yPay{5}
\draw (0,0) -- (2,-1)  node [anchor=west] {(\color{red}\xPay\color{black},\color{blue}\yPay\color{black})};
     \draw (0,0) -- (2,0);
        \node at (1,-.75) [draw=blue,fill=white] {\small \color{blue}{$C$}};
           \node at (1,0) [draw=blue,fill=white] {\small \color{blue}{$B$}};
        \node at (1,.75) [draw=blue,fill=white] {\small \color{blue}{$A$}};
        \draw [fill=white,draw=blue] (0,0) circle (.25) node {\small $p_j$} node [anchor=south,yshift=.1in] {\small $x_5$};
      \end{scope}

 \def\xPay{$z_j$}
      \def\yPay{0}
      \begin{scope}[xshift=-3cm,yshift=-2cm]
        \draw (0,0) -- (-2,1)  node [anchor=east] {(\color{red}\xPay\color{black},\color{blue}\yPay\color{black})};
            \def\xPay{1}
      \def\yPay{6}
        \draw (0,0) -- (-2,0)  node [anchor=east] {(\color{red}\xPay\color{black},\color{blue}\yPay\color{black})};
             \def\xPay{0}
      \def\yPay{5}
       \draw (0,0) -- (-2,0);
 \draw (0,0) -- (-2,-1)  node [anchor=east] {(\color{red}\xPay\color{black},\color{blue}\yPay\color{black})};
        \node at (-1,-.75) [draw=blue,fill=white] {\small \color{blue}{$C$}};
         \node at (-1,0) [draw=blue,fill=white] {\small \color{blue}{$B$}};
        \node at (-1,.75) [draw=blue,fill=white] {\small \color{blue}{$A$}};
        \draw [fill=white,draw=blue] (0,0) circle (.25) node {\small $p_j$} node [anchor=north,yshift=-.1in] {\small $x_6$};
      \end{scope}

     \def\xPay{1}
      \def\yPay{6}
      \begin{scope}[xshift=-3cm,yshift=2cm]
     \draw (0,0) -- (-2,1)  node [anchor=east] {(\color{red}\xPay\color{black},\color{blue}\yPay\color{black})};
             \def\xPay{0}
      \def\yPay{5}
       \draw (0,0) -- (-2,0);
 \draw (0,0) -- (-2,-1)  node [anchor=east] {(\color{red}\xPay\color{black},\color{blue}\yPay\color{black})};
  \def\xPay{$z_j$}
      \def\yPay{0}
        \draw (0,0) -- (-2,0)  node [anchor=east] {(\color{red}\xPay\color{black},\color{blue}\yPay\color{black})};
        \node at (-1,-.75) [draw=blue,fill=white] {\small \color{blue}{$C$}};
         \node at (-1,0) [draw=blue,fill=white] {\small \color{blue}{$B$}};
        \node at (-1,.75) [draw=blue,fill=white] {\small \color{blue}{$A$}};
        \draw [fill=white,draw=blue] (0,0) circle (.25) node {\small $p_j$} node [anchor=south,yshift=.1in] {\small $x_4$};
      \end{scope}

\end{tikzpicture}

}\end{center}
\begin{multicols}{2}

For \(j = 1\),
\[
{\color{red}z_1} =
\begin{cases}
\color{red}{2} & \text{if \(\color{blue}{p_2}\) chooses \(\color{blue}{A}\) or \(\color{blue}{B}\)}, \\
\color{red}{1} & \text{if \(\color{blue}{p_2}\) chooses \(\color{blue}{C}\)}.
\end{cases}
\]

\columnbreak

For \(j = 2\),
\[
{\color{red}z_2} =
\begin{cases}
\color{red}{2} & \text{if \(\color{blue}{p_1}\) chooses \(\color{blue}{C}\)}, \\
\color{red}{1} & \text{if \(\color{blue}{p_1}\) chooses \(\color{blue}{A}\) or \(\color{blue}{B}\)}.
\end{cases}
\]

\end{multicols}
\caption{Illustration of Team $j$'s Mechanism}
\label{fig:mechanism-tree}
\end{figure}

Given a profile of mechanisms \((\alpha_1(\cdot \mid t_1'), \alpha_2(\cdot \mid t_2'))\), the team \( j \) mechanism \(\alpha_j(\cdot \mid t_j')\) is said to be \textit{incentive compatible} if truthful reporting is a best response for team member \( m_j \), given the profile of mechanisms \((\alpha_1(\cdot \mid t_1'), \alpha_2(\cdot \mid t_2'))\) and given that all other agents report truthfully. That is, for each \( t_j \in \{\theta_A, \theta_B\} \), the report \( t_j'(t_j) = t_j \) maximizes member \(m_j\)'s expected payoff. In Figure~\ref{fig:mechanism-tree}, the truthful reporting paths \((x_2, x_4)\) and \((x_3, x_7)\) are emphasized in red. Team \( j \)'s \textit{incentive-compatible mechanism correspondence}, \( IC_j \), is a set-valued map that associates each mechanism of the other team, \(\alpha_{-j}(\cdot \mid t_{-j}')\), with the set of incentive-compatible mechanisms for team \( j \), denoted \( IC_j(\alpha_{-j}(\cdot \mid t_{-j}')) \).

This induces a generalized game among principals. Each principal \( j \)'s feasible-strategy correspondence is its incentive-compatible mechanism correspondence \( IC_j \). A profile of mechanisms \((\alpha_1(\cdot|t_1'), \alpha_2(\cdot|t_2'))\) is said to be a \textit{principals' equilibrium} if each principal \( j \)'s mechanism \(\alpha_j(\cdot|t_j')\) is a best response to \(\alpha_{-j}(\cdot|t_{-j}')\) among team \( j \)'s set of incentive-compatible mechanisms \( IC_j(\alpha_{-j}(\cdot|t_{-j}')) \). Thus, the source of the existence problem is not the principals'
objective functions but the endogenous feasible-strategy
correspondences induced by incentive compatibility.

The key difficulty is that these endogenous feasible-strategy
correspondences need not be lower hemicontinuous. We illustrate this by
focusing on team~1's incentive-compatible mechanism correspondence,
though an analogous argument applies to team~2.

Specifically, suppose team $2$'s mechanism chooses $C$ with probability $1$ for both
possible reports, i.e., $\alpha_2(C \mid t'_2) = 1$ for each $t'_2 \in \{\theta_A,
\theta_B\}$. Then team $1$'s set of incentive-compatible mechanisms,
$IC_1(\alpha_2(\cdot \mid t'_2))$, includes mechanisms in which $A$ is played
following report $\theta_A$ and $B$ is played following report $\theta_B$. However,
if team $2$'s mechanism places any strictly positive probability on $A$ or $B$ for
either report, then this mechanism is no longer in team $1$'s IC set---it
discontinuously exits the feasible set under arbitrarily small perturbations of team
$2$'s mechanism. Consequently, team $1$'s feasible-strategy correspondence fails to
be lower hemicontinuous: a mechanism that is feasible at the limit point cannot be
approximated by feasible mechanisms along sequences approaching that limit. Thus, an arbitrarily small change in an opponent's mechanism can eliminate a previously incentive-compatible mechanism.

In Myerson's example, this discontinuity also rules out a fixed point of the principals' best-response problem. The following cycle gives the intuition:
\begin{itemize}
    \item If principal \( 2 \) chooses \( C \) with probability \( 1 \) for both reports, then principal \( 1 \)'s best response is to choose \( A \) following report \( \theta_A \) and \( B \) following report \( \theta_B \).
    \item But if principal \( 1 \) chooses \( A \) following \( \theta_A \) and \( B \) following \( \theta_B \), then principal \( 2 \)'s best response is to choose \( A \) following \( \theta_A \) and \( B \) following \( \theta_B \).
    \item If principal \( 2 \) chooses \( A \) or \( B \) with strictly positive probability for either report, then principal \( 1 \)'s best response is to choose \( C \) with probability \( 1 \) for both reports.
    \item But if principal \( 1 \) chooses \( C \) with probability \( 1 \) for both reports, then principal \( 2 \)'s best response is also to choose \( C \) with probability \( 1 \) for both reports.
\end{itemize}
Thus the principals' best responses cycle because the incentive-compatible mechanism correspondence changes discontinuously; Myerson's full argument rules out the remaining mixed or partially separating candidates as well.
Myerson's example illustrates that the central obstacle to equilibrium
is the discontinuity of the endogenous feasible-strategy
correspondences generated by incentive compatibility. At first glance,
this appears to leave little hope for a general existence theory.
Remarkably, however, Myerson also establishes a positive existence
result. Proposition~4 proves the existence of a quasi-equilibrium in
the case of finite type and action spaces.
\begin{center}
\fbox{%
\parbox{0.9\linewidth}{
\textbf{Proposition 4 (\citet{myerson1982optimal}).}
A quasi-equilibrium must exist for any multi-principal problem
with finite type sets and finite decision sets.
}}
\end{center}

In Myerson's example, the always-\(C\) profile is a quasi-equilibrium
even though it is not an exact principals' equilibrium. At that
profile, principal~\(1\) has a profitable exactly incentive-compatible
deviation to separation, but the deviation is unsafe: any positive
probability that principal~\(2\) chooses \(A\) or \(B\) makes it
infeasible.

In an \(\varepsilon\)-relaxed IC equilibrium, by contrast, no principal
forgoes a profitable feasible deviation, safe or unsafe; what is
weakened is the agents' side, since truthful reporting may leave each
agent with expected regret up to \(\varepsilon\), and because the bound
is ex ante, the tolerance budget may be allocated unevenly across
types. The concepts also
differ in what they deliver: Myerson's yields an exactly
incentive-compatible profile that is optimal only against safe
deviations, whereas ours yields exact optimality at each fixed
\(\varepsilon>0\) over a relaxed feasible set.

We now illustrate this distinction in Myerson's example. Fix
\(0<\varepsilon<1/2\), and consider the following mechanism profile.

Principal \(1\)'s mechanism is
\[
\alpha_1(A\mid\theta_A)=\alpha_1(C\mid\theta_A)=\frac12,
\qquad
\alpha_1(B\mid\theta_A)=0,
\]
and
\[
\alpha_1(B\mid\theta_B)=\alpha_1(C\mid\theta_B)=\frac12,
\qquad
\alpha_1(A\mid\theta_B)=0.
\]
Principal \(2\)'s mechanism is
\[
\alpha_2(A\mid\theta_A)=2\varepsilon,
\qquad
\alpha_2(C\mid\theta_A)=1-2\varepsilon,
\qquad
\alpha_2(B\mid\theta_A)=0,
\]
and
\[
\alpha_2(B\mid\theta_B)=2\varepsilon,
\qquad
\alpha_2(C\mid\theta_B)=1-2\varepsilon,
\qquad
\alpha_2(A\mid\theta_B)=0.
\]

First note that this mechanism profile is \(\varepsilon\)-relaxed incentive
compatible. Consider team \(1\). If the agent's true type is
\(\theta_A\), truthful reporting yields action \(A\) with probability
\(1/2\) and \(C\) with probability
\(1/2\), giving expected payoff \(1/2\). If the agent instead reports
\(\theta_B\), the mechanism recommends \(B\) with probability \(1/2\) and \(C\) with probability
\(1/2\).
Conditional on the mismatched recommendation \(B\), the payoff is
\(z_1\), where
\[
z_1=
\begin{cases}
2 & \text{if principal \(2\) recommends \(A\) or \(B\),}\\
1 & \text{if principal \(2\) recommends \(C\).}
\end{cases}
\]
The recommendation \(C\) gives payoff \(0\).

Under principal \(2\)'s mechanism, \(A\) or \(B\) is recommended with
probability \(2\varepsilon\), so
\[
\mathbb E[z_1]=1+2\varepsilon.
\]
Thus the misreport payoff is
\[
\frac12(1+2\varepsilon)=\frac12+\varepsilon,
\]
and the deviation gain is exactly \(\varepsilon\). The same calculation
applies to type \(\theta_B\). Hence team \(1\)'s incentive constraint is
violated by at most \(\varepsilon\).

For team \(2\), if the agent's true type is \(\theta_A\), truthful
reporting yields action \(A\) with probability \(2\varepsilon\) and action \(C\) with probability
\(1-2\varepsilon\), giving
expected payoff \(2\varepsilon\). If the agent instead reports
\(\theta_B\), the mechanism recommends \(B\) with probability
\(2\varepsilon\) and \(C\) with probability
\(1-2\varepsilon\). Conditional on the mismatched recommendation \(B\), the payoff is
\(z_2\), where
\[
z_2=
\begin{cases}
2 & \text{if principal \(1\) recommends \(C\),}\\
1 & \text{if principal \(1\) recommends \(A\) or \(B\).}
\end{cases}
\]
The recommendation \(C\) gives payoff \(0\).

Under principal \(1\)'s mechanism, \(C\) is recommended with probability
\(1/2\), so
\[
\mathbb E[z_2]=2\cdot\frac12+1\cdot\frac12=\frac32.
\]
Thus the misreport payoff is
\[
2\varepsilon\cdot\frac32=3\varepsilon,
\]
while the truthful payoff is \(2\varepsilon\). The deviation gain is
\(\varepsilon\). The same calculation applies to type \(\theta_B\). Hence
team \(2\)'s incentive constraint is also violated by at most
\(\varepsilon\).

It remains to verify optimality against arbitrary, possibly asymmetric,
mechanism deviations. Fix a principal and let \(z>1\) denote the
expected payoff its agent obtains from an action mismatched with the
agent's true type. For an arbitrary mechanism, let \(x_A\) and \(x_B\)
be the probabilities of the type-matched actions following reports
\(\theta_A\) and \(\theta_B\), and let \(y_A\) and \(y_B\) be the
probabilities of the corresponding type-mismatched actions. Set
\[
X:=x_A+x_B,
\qquad
Y:=y_A+y_B.
\]
Adding the two types' gains from misreporting gives
\[
(z-1)(X-Y).
\]
Because ex ante regret is one half the sum of the positive parts of the
two type-specific gains, it is at least
\[
\frac12(z-1)(X-Y).
\]
Consequently, \(\varepsilon\)-IC requires
\[
X-Y\le \frac{2\varepsilon}{z-1}.
\]

The principal's expected payoff under this arbitrary mechanism is
\[
\frac12\bigl[6X+5(2-X-Y)\bigr]
=5+\frac12(X-5Y).
\]
Using the IC bound,
\[
5+\frac12(X-5Y)
\le
5+\frac{\varepsilon}{z-1}-2Y
\le
5+\frac{\varepsilon}{z-1}.
\]
Since also \(X\le2\), the principal's payoff is at most
\(5+\min\{1,\varepsilon/(z-1)\}\). This bound is attained by setting
\(Y=0\) and \(X=\min\{2,\,2\varepsilon/(z-1)\}\), so a symmetric best
response assigns matched-action probability
\(x=\min\{1,\varepsilon/(z-1)\}\) after each report. Each type's
misreporting gain is \((z-1)x=\min\{z-1,\varepsilon\}\), so this
mechanism satisfies the full ex ante regret constraint and attains
payoff \(5+x\). Against principal \(2\)'s
displayed mechanism, principal \(1\) faces \(z=1+2\varepsilon\), giving
\(x_1=1/2\). Against principal \(1\)'s mechanism, principal \(2\)
faces \(z=3/2\), giving \(x_2=2\varepsilon\). The displayed mechanisms
are therefore best responses over the full \(\varepsilon\)-IC sets.

Therefore this profile is an \(\varepsilon\)-relaxed IC equilibrium. It is
distinct from the profile in which both principals always recommend
\(C\). Although the constant-\(C\) profile satisfies \(\varepsilon\)-IC,
principal \(1\) can profitably deviate from it by increasing the
probability of type-matched recommendations without violating
\(\varepsilon\)-IC (or violating exact IC).

Finally, as \(\varepsilon\downarrow0\), this profile converges to
\[
x_1=\frac12,
\qquad
x_2=0.
\]
The limit profile is exact IC, since both deviation gains are zero. It
also satisfies Myerson's limiting optimality requirement directly. Along
the approximating sequence, principal~\(1\)'s equilibrium payoff is
\(5.5\), while its best exactly incentive-compatible deviation yields
\(5\); principal~\(2\)'s equilibrium payoff is
\(5+2\varepsilon\to5\), and its best exactly incentive-compatible
deviation payoff is \(5\). Hence each
principal's limiting payoff is at least the corresponding limiting value
from the exact-IC deviation problem. Thus this limit is a
quasi-equilibrium, alongside the always-\(C\) quasi-equilibrium discussed
by Myerson. It is not an exact principals' equilibrium. Given \(x_2=0\),
principal \(1\)'s exact IC constraint does not restrict \(x_1\), so
principal \(1\) can profitably deviate from \(x_1=1/2\) to the fully
separating rule \(x_1=1\). Thus the limit of
\(\varepsilon\)-relaxed IC equilibria need not be an exact equilibrium:
new profitable exact-IC deviations can appear at the limit. In this
example, the limit does not separate the relaxed construction from
quasi-equilibrium; rather, it illustrates that quasi-equilibria may be
nonunique and need not be exact equilibria.

The remainder of the paper develops two complementary responses to this
nonexistence problem within a general interacting-principals framework.
Working directly with the product of the principals' mechanism spaces, we
first identify conditions under which, for every \(\varepsilon>0\), there
exists an equilibrium in which principals best respond exactly over
\(\varepsilon\)-incentive-compatible mechanisms. We then specialize to an
embedded one-member quasi-linear contest. Concavity of the contest
technology and a uniformly profitable pooling output make the exact
budget-feasible correspondence continuous and yield an exact-equilibrium
result for that environment.
\section{Model}\label{modelsection}
We analyze an environment in which multiple teams with agency concerns interact strategically, while each team’s principal designs a mechanism for its own team.

To provide an overview, consider a game with \(N\) teams, each consisting
of \(n\) team members,\footnote{The assumption of equal team sizes is
made solely for notational convenience; our results extend directly to
teams of different sizes.} where teams are indexed by
\(j\in\{1,\ldots,N\}\) and team members by
\(i\in\{1,\ldots,n\}\). Each team faces both adverse selection, because
agents privately observe their types, and moral hazard, because agents'
realized individual decisions need not coincide with their
recommendations.

The game begins with each principal selecting a mechanism for their
team. The selected mechanism is observed by the members of that team,
but not by other teams. The continuation game then unfolds as follows.
Team members privately observe their types and report them to their
principal through cheap, unverifiable communication. The mechanism then
draws a team-level decision and recommended individual decisions for the
team members. After observing the team-level decision and their own individual
recommendations, team members choose their actual individual decisions.

The following subsections formalize the primitive spaces and each stage
of the game in turn. We then define the principals' mechanism-design
problems and the resulting generalized game.

\subsection{Primitive Spaces}
Let \(T\) denote each agent's type space. Let \(D_a\) denote each
agent's individual decision space, and let \(D_p\) denote the
principal's team-level decision space. We keep \(D_p\) abstract so that
the model can cover a range of applications; depending on the
environment, it may represent payment rules, prize-sharing rules,
production decisions, or other team-level allocations. The full team
decision space is
\[
D:=D_p\times D_a^n.
\]
Throughout the paper, all spaces are endowed with their Borel
\(\sigma\)-algebras, and all product spaces are equipped with the
product topology and the associated product Borel \(\sigma\)-algebra.
Unless otherwise stated, all maps are Borel measurable.
\subsection{First (Mechanism) Stage}

First, each principal selects a mechanism specifying how team decisions
are generated from the agents' reported types. Let
\[
\overline M:=\mathcal K(T^n,D)
\]
denote the set of stochastic kernels from reported team type profiles
to team decision vectors. A kernel representative for principal \(j\)'s
mechanism is an element
\[
m_j\in\overline M.
\]
Thus, after receiving the reported type profile
\(\mathbf t'_j\in T^n\), mechanism \(m_j\) draws
\[
d'_j=(d_{p,j},d'_{1,j},\ldots,d'_{n,j})\in D,
\]
where \(d_{p,j}\) is the team-level decision and \(d'_{i,j}\) is the
individual decision generated by the mechanism for agent \((i,j)\). The
prime on \(d'_{i,j}\) indicates that this is a mechanism-generated
individual recommendation; the agent's actual individual decision,
chosen later after observing its recommendation, is denoted
\(d_{i,j}\).

For notational convenience, let
\[
\widehat{\mathbf d}'=(d'_1,\ldots,d'_N)
\]
denote the profile of mechanism-generated team decision vectors.
We write \(\widehat{\mathbf d}'_{-j}\) for the profile of
mechanism-generated decision vectors of all teams other than \(j\).
Within team \(j\), we write
\[
\mathbf d'_{-i,j}
=
(d'_{1,j},\ldots,d'_{i-1,j},d'_{i+1,j},\ldots,d'_{n,j})
\]
for the individual recommendations of the other members of team \(j\).
The realized team-\(j\) decision vector after agent \((i,j)\)'s actual
individual decision is substituted for its recommendation is written
\[
(d_{p,j},\mathbf d'_{-i,j},d_{i,j}).
\]

\subsection{Second (Private-Type) Stage}

In the second stage, the joint type profile
\[
\widehat{\mathbf t}\in T^{nN}
\]
is drawn according to a probability measure
\[
H\in\mathcal P(T^{nN}).
\]
Each agent \((i,j)\) privately observes their own type \(t_{i,j}\).

For notational convenience, let \(\widehat{\mathbf t}_{-i,j}\) denote
the \((nN-1)\)-tuple of types of all agents other than agent \((i,j)\).
When referring only to the remaining members of team \(j\), we write
\(\mathbf t_{-i,j}\) for the corresponding \((n-1)\)-tuple of teammate
types.

Each team member then privately reports their type to their principal
through cheap, unverifiable communication. Since the team's selected
mechanism is observed within the team, a second-stage type-reporting
strategy for agent \((i,j)\) is a measurable function
\[
t'_{i,j}:\overline M\times T\to T,
\]
where
\[
t'_{i,j}(m_j,t_{i,j})
\]
denotes the reported type when agent \((i,j)\)'s true type is
\(t_{i,j}\) and team \(j\)'s selected mechanism is \(m_j\). Let \[
\mathcal T
=
\left\{
t'_{i,j}:\overline M\times T\to T
\;\middle|\;
t'_{i,j}\text{ is measurable}
\right\}
\]
denote the set of second-stage reporting strategies.

Let \(H_j\in\mathcal P(T^n)\) denote the marginal of \(H\) on team
\(j\)'s type profile and \(H_{i,j}\in\mathcal P(T)\) the marginal on
agent \((i,j)\)'s type. Two kernels in \(\overline M\) that agree at
\(H_j\)-almost every report profile generate the same joint law of team
\(j\)'s types and decisions under truthful reporting, so we identify
them. Principal \(j\)'s mechanism space \(M_j\) is the resulting set of
equivalence classes; we write \(\mathcal M:=\prod_{j=1}^N M_j\) for the
mechanism-profile space and \(M_{-j}:=\prod_{k\ne j}M_k\) for an
opponents' profile. Section~\ref{assumptionssection} equips \(M_j\)
with the Young--narrow topology.

\subsection{Third (Decision) Stage}

After the agents' reports are submitted, mechanism \(m_j\) draws the
team-level decision \(d_{p,j}\) and the profile of individual
recommendations \(\mathbf d'_j\). After observing the team-level decision
and their own individual recommendations, agents choose their actual
individual decisions.

Since the team's selected mechanism and the team-level decision are
observed within the team, a third-stage decision strategy for agent
\((i,j)\) is a measurable function
\[
d_{i,j}:\overline M\times T\times T\times D_p\times D_a\to D_a,
\]
where
\[
d_{i,j}(m_j,t_{i,j},t'_{i,j},d_{p,j},d'_{i,j})
\]
denotes the actual individual decision chosen as a function of team
\(j\)'s selected mechanism, the agent's true type, reported type, the
team-level decision, and the agent's individual recommendation. Let
\[
\mathcal D
=
\left\{
d_{i,j}:\overline M\times T\times T\times D_p\times D_a\to D_a
\;\middle|\;
d_{i,j}\text{ is measurable}
\right\}
\]
denote the set of third-stage decision strategies.

For \(m_j\in\overline M\), \(t,t'\in T\), and
\(B\in\mathcal B(D_p\times D_a)\), define
\[
Q_{i,j}^{m_j}(B\mid t,t')
:=
\int_{T^{n-1}}
m_j\!\left(
\{(d_{p,j},\mathbf d'_j):(d_{p,j},d'_{i,j})\in B\}
\,\middle|\,
t',\mathbf t_{-i,j}
\right)
H_{-i,j}(d\mathbf t_{-i,j}\mid t),
\]
where \(H_{-i,j}(\cdot\mid t)\) is the conditional law of the remaining
members of team \(j\)'s types given \(t_{i,j}=t\). Thus
\(Q_{i,j}^{m_j}(\cdot\mid t,t')\) is the marginal law of the
team-level decision and agent \((i,j)\)'s individual recommendation
when agent \((i,j)\)'s true type is \(t\), its report is \(t'\), and
the other team members report truthfully.

A pure strategy
\[
(t_{i,j}^{\prime\star},d_{i,j}^{\star})
\in \mathcal T\times\mathcal D
\]
is \emph{honest and obedient} for mechanism \(m_j\) if
\[
t_{i,j}^{\prime\star}(m_j,t)=t
\quad
\text{for }H_{i,j}\text{-almost every }t\in T,
\]
and, for \(H_{i,j}\)-almost every \(t\in T\),
\[
d_{i,j}^{\star}(m_j,t,t,d_{p,j},d'_{i,j})=d'_{i,j}
\]
for \(Q_{i,j}^{m_j}(\cdot\mid t,t)\)-almost every
\((d_{p,j},d'_{i,j})\in D_p\times D_a\).

\subsection{Principals' Initial Problems}
We now introduce the incentive-compatibility constraints that determine
the principals' endogenous feasible mechanism sets. For notational
simplicity, we assume all agents share a common von
Neumann--Morgenstern utility function
\[
u:T^{nN}\times D^{N}\to\mathbb R,
\]
with heterogeneity entering only through type differences.\footnote{All
results extend straightforwardly to settings with agent-specific utility
functions \(u_{i,j}\).}
We write \(u_{i,j}\) to denote the utility of agent \((i,j)\) evaluated
at a given realization of the type profile, the profile of team-level
decisions, and the profile of agents' actual individual decisions.

A mechanism profile \(\mathbf m\) satisfies \emph{Bayesian incentive
compatibility} for team \(j\) if every team member finds it optimal to
report truthfully and follow its recommended individual decision,
against every unilateral deviation available to it. A deviation by
agent \((i,j)\) has two components: it may misreport its type as
\(t'\ne t\), which alters team \(j\)'s mechanism-generated team-level
decision and individual recommendations (since these now respond to
the report profile \((\mathbf t_{-i,j},t')\) rather than the truth);
and, having observed the team-level decision \(d_{p,j}\) and its own
recommendation \(d'_{i,j}\), it may replace its own actual individual
decision using any measurable decision rule \(d_{i,j}\in\mathcal D\),
rather than following the recommendation. All other agents, on both
teams, are held at truthful reporting and obedient play throughout.

Write
\[
U_{i,j}(t;\mathbf m)
:=
E_{\mathbf m}\!\left[
u_{i,j}\!\left(\widehat{\mathbf t},\widehat{\mathbf d}^{\prime}\right)
\,\middle|\, t_{i,j}=t
\right]
\]
for agent \((i,j)\)'s conditional expected payoff from truthful
reporting and obedient play, given \(t_{i,j}=t\), and\footnote{For
simplicity we abstract from individual rationality in the general
model. Extending the existence argument to impose it would require
additional participation assumptions---for example, an exactly
incentive-compatible seed with a uniformly positive interim participation
margin over types and opponents' mechanisms, together with suitable
continuity of the participation constraints---or a separate
feasibility-continuity argument. The exact
contest result of Section~\ref{sec:exact-equilibrium} imposes interim
individual rationality.}
\[
D_{i,j}(t;\mathbf m)
:=
\operatorname*{ess\,sup}_{t'\in T}\ \sup_{d_{i,j}\in\mathcal D}\
E_{\mathbf m}^{\,t',d_{i,j}}\!\left[
u_{i,j}\!\left(
\widehat{\mathbf t},\widehat{\mathbf d}^{\prime}_{-j},
\left(d_{p,j},\mathbf d'_{-i,j},
d_{i,j}\!\left(m_j,t,t',d_{p,j},d'_{i,j}\right)\right)
\right)
\,\middle|\, t_{i,j}=t
\right]
\]
for its conditional expected payoff from the best version-independent
unilateral deviation---the essential supremum over reports, with the
agent free to choose any measurable decision rule after observing its
recommendation---given true type \(t\). The essential supremum is taken
with respect to \(H_{i,j}\). This is the appropriate object because a
mechanism represented by its fixed-marginal law is defined only up to an
\(H_{i,j}\)-null set of reports. Equivalently, incentive inequalities are
required for true types and reports in a common \(H_{i,j}\)-conull set,
as in the distributional formulation of \citet[Section~13.3 and
Proposition~13.4]{kadan2017}. Appendix~\ref{app:epsilon} establishes this
equivalence and shows that truthful reporting and obedience imply
\(D_{i,j}(t;\mathbf m)\ge U_{i,j}(t;\mathbf m)\) for almost every \(t\).
Bayesian incentive compatibility therefore requires
\begin{equation}
D_{i,j}(t;\mathbf m)=U_{i,j}(t;\mathbf m)
\qquad\text{for \(H_{i,j}\)-almost every \(t\in T\),}
\label{mainIC}
\end{equation}
for every agent \((i,j)\).

Because reports are evaluated through an essential supremum, a
deviation to a set of reports of probability zero is not tested
directly. Lemma~\ref{lem:safe-completion} in
Appendix~\ref{app:epsilon} shows that this entails no loss. For a given
mechanism profile, the deviation value of almost every type is
dominated by its essential supremum at every report outside a null set
(Lemma~\ref{lem:ic-common-conull-bridge}). Redirecting, coordinate by
coordinate, each agent's reports in that null set to a fixed report
outside it produces an \(H_j\)-equivalent version of the mechanism: a
redirected report generates exactly the decisions and recommendations of
the fixed report, and teammates' truthful reports fall in their null
sets with conditional probability zero, whether or not types are
correlated. The completion preserves truthful outcome laws and makes the
supremum over all reports and subsequent action rules equal to the
original essential-supremum deviation value for almost every true type.
It therefore preserves both exact incentive compatibility and the ex
ante regret defined below. Completing an equilibrium profile
consequently yields an equilibrium under unrestricted reporting
deviations: every unrestricted-feasible principal deviation is also
feasible in the conull-report game and has the same truthful principal
payoff.

Define the ex-ante truthful and best-deviation payoffs
\[
T_{i,j}(\mathbf m):=\int_TU_{i,j}(t;\mathbf m)\,H_{i,j}(dt),
\qquad
D_{i,j}(\mathbf m):=\int_TD_{i,j}(t;\mathbf m)\,H_{i,j}(dt),
\]
and agent \((i,j)\)'s ex-ante regret
\(R_{i,j}(\mathbf m):=D_{i,j}(\mathbf m)-T_{i,j}(\mathbf m)\ge0\), the
inequality inherited almost everywhere from \(D_{i,j}(t;\mathbf m)\ge
U_{i,j}(t;\mathbf m)\). Since this difference is nonnegative almost
everywhere, its integral vanishes only if the difference itself
vanishes \(H_{i,j}\)-almost everywhere; consequently
\begin{equation}
R_{i,j}(\mathbf m)=0
\qquad\Longleftrightarrow\qquad
\text{\eqref{mainIC} holds},
\label{eq:ic-ex-ante-equiv}
\end{equation}
an ex-ante characterization of exact incentive compatibility in the
spirit of \citet{kadan2017}.\footnote{The substantive content behind
\eqref{eq:ic-ex-ante-equiv} is not the equivalence itself, which is
the elementary fact that a nonnegative function with zero integral
vanishes a.e., but establishing that
\(D_{i,j}(\cdot;\mathbf m)\)---defined through an essential supremum
over reports and a supremum over decision rules---is itself measurable,
and that this condition is equivalent to IC on a common conull set of
true types and reports. These points are established in
Corollary~\ref{cor:ess-sup-report-kernel} and
Lemma~\ref{lem:ic-common-conull-bridge} in
Appendix~\ref{app:epsilon}.}

For \(\varepsilon\ge0\),
a mechanism profile \(\mathbf m\) is said to satisfy \emph{\(\varepsilon\)-relaxed
Bayesian incentive compatibility} for team \(j\) if
\[
R_{i,j}(\mathbf m)\le\varepsilon
\qquad\text{for every agent \((i,j)\), \(i=1,\ldots,n\).}
\]
At \(\varepsilon=0\) this is, by \eqref{eq:ic-ex-ante-equiv}, exact
Bayesian incentive compatibility \eqref{mainIC}.

When \(\varepsilon=0\), it makes no difference whether incentive
compatibility is imposed ex ante or type by type. If some set of types
with positive probability could gain by deviating, the agent could
deviate on exactly those types and raise its ex ante payoff; so
\(R_{i,j}(\mathbf m)=0\) holds exactly when truthful reporting and
obedience are optimal for almost every type.

When \(\varepsilon>0\), the two formulations differ. The bound
\(R_{i,j}(\mathbf m)\le\varepsilon\) limits the agent's deviation gain
on average across its types, not for each type separately. A type
that occurs with small probability can therefore have a deviation gain
well above \(\varepsilon\), as long as its contribution to the average is
small. Roughly, the rarer a type, the larger the temptation it may be
left with; what \(\varepsilon\)-relaxed incentive compatibility rules out
is a large temptation affecting a substantial share of types.

Accordingly, for each opponents' mechanism profile
\(\mathbf m_{-j}\in M_{-j}\) and each \(\varepsilon\ge0\), let
\[
IC_j^\varepsilon(\mathbf m_{-j})\subseteq M_j
\]
denote the set of mechanisms for principal \(j\) satisfying
\(\varepsilon\)-relaxed Bayesian incentive compatibility, and write
\(IC_j(\mathbf m_{-j}):=IC_j^0(\mathbf m_{-j})\) for the exact
incentive-compatible set.

Finally, principal \(j\)'s preferences are represented by a bounded
continuous von Neumann--Morgenstern utility function
\[
\pi_j:T^{nN}\times D^{N}\to\mathbb R,
\]
which may depend on the complete type profile and the full profile of
team decision vectors.

Each principal \(j\) best responds to \(\mathbf m_{-j}\) over its
\(\varepsilon\)-relaxed incentive-compatible feasible set. Accordingly,
principal \(j\)'s \(\varepsilon\)-relaxed best-response correspondence
is
\[
BR_j^\varepsilon(\mathbf m_{-j})
:=
\left\{
m_j\in IC_j^\varepsilon(\mathbf m_{-j})
:
E_{(m_j,\mathbf m_{-j})}[\pi_j]
=
\sup_{\tilde m_j\in IC_j^\varepsilon(\mathbf m_{-j})}
E_{(\tilde m_j,\mathbf m_{-j})}[\pi_j]
\right\},
\]
and we again write \(BR_j(\mathbf m_{-j}):=BR_j^0(\mathbf m_{-j})\).

We are now in a position to define our equilibrium concept.

\begin{definition}[\(\varepsilon\)-Relaxed Bayesian--Nash Principals' Equilibrium]
\label{def:BNPE}
For \(\varepsilon\ge0\), a mechanism profile
\[
\mathbf m^\varepsilon
=
(m_1^\varepsilon,\ldots,m_N^\varepsilon)
\in\mathcal M
\]
is an \emph{\(\varepsilon\)-relaxed Bayesian--Nash Principals'
Equilibrium} if, for every principal \(j=1,\ldots,N\),
\[
m_j^\varepsilon\in BR_j^\varepsilon(\mathbf m_{-j}^\varepsilon).
\]
Equivalently, each principal \(j\)'s mechanism is a best response to
\(\mathbf m_{-j}^\varepsilon\) over the corresponding
\(\varepsilon\)-relaxed incentive-compatible mechanism set
\(IC_j^\varepsilon(\mathbf m_{-j}^\varepsilon)\). \end{definition}
At \(\varepsilon=0\)
this is a \emph{Bayesian--Nash Principals' Equilibrium (BNPE)}. The
general theorem below establishes the relaxed concept for every
\(\varepsilon>0\). Exact equilibrium requires some additional structure
which is introduced separately.

\subsection{Interacting Principal-Agent Games}

We summarize the extensive-form game with interacting principals with agency concerns by
\[
\Gamma(N,n,T,D,H,u,\{\pi_j\}_{j=1}^N),
\]
where \(D=D_p\times D_a^n\) is the team decision space and $T$ is the agent type space. The environment
consists of \(N\) teams, each with \(n\) agents. Each principal initially
selects a mechanism \(m_j\in M_j\), while the
continuation game determines the induced outcome through the agents'
type reports, the mechanism-generated team decisions, and the agents'
actual individual decisions. Bayesian incentive compatibility determines
each principal's endogenous feasible mechanism correspondence, and the
principals then interact in the generalized game induced by these
correspondences.

The game proceeds as follows.

\begin{itemize}

\item \textbf{Mechanism Stage.}
The principal for each team selects a mechanism.

\item \textbf{Private-Type Stage.}
Nature draws the profile of agents' types according to the joint
distribution \(H\), and each agent privately observes their own type.

Agents simultaneously and privately report their types to their own
team principal.

\item \textbf{Decision Stage.}
Each mechanism draws a team-level decision and individual
recommendations for its agents based on the reported team type profile.

After observing the team-level decision and their own individual
recommendations, agents simultaneously choose their actual individual
decisions.

\end{itemize}

The formulation above defines the economic environment and the
corresponding generalized game among principals. To state equilibrium
existence results, we must also specify the topology on the principals'
strategy spaces.

\section{Mechanism Topology and Strategy Spaces}
\label{assumptionssection}

To study whether the game defined in Section~\ref{modelsection} has an
equilibrium, we need two further ingredients: maintained assumptions on
the primitives, and a notion of when two mechanisms are close.

We first state the maintained assumptions. They fall into four groups.
The type, decision, and team-decision spaces are compact. The joint
distribution of types has a well-behaved density, which allows types to
be correlated within and across teams. Agents' and principals' payoffs
are bounded and continuous. Finally, each team has a baseline decision
that no agent wants to disobey, which guarantees that every principal
always has at least one incentive-compatible mechanism.

We then describe the principals' strategy spaces. A mechanism matters
only through the joint distribution of reported types and team decisions
it generates, so two mechanisms that differ only on report profiles of
probability zero are treated as the same. Two mechanisms are close when
these joint distributions are close. This notion of closeness, the
Young--narrow topology, gives each principal a strategy space that is
compact and convex. Finally, we describe how principals' payoffs are
computed from a mechanism profile, by evaluating them under truthful
reporting and obedient play.

\subsection*{Assumptions}

\begin{assumption}[Compact metric primitive spaces]
\label{assump:compact-spaces}
The type space \(T\), the principal decision space \(D_p\), and the
individual decision space \(D_a\) are compact metric spaces, endowed with
their Borel \(\sigma\)-algebras. The team decision space is
\[
D:=D_p\times D_a^n.
\]
All finite product spaces are endowed with the product topology and the
associated product Borel \(\sigma\)-algebra.
\end{assumption}

Under Assumption~\ref{assump:compact-spaces}, the team decision space
\(D=D_p\times D_a^n\) and all finite type-decision product spaces are
compact metric spaces. Hence the associated probability spaces are
compact and metrizable under narrow convergence. In particular, the
fixed-marginal law representation
\[
H_j(d\mathbf t_j)m_j(dd_j\mid \mathbf t_j)
\]
places \(M_j\) inside a compact metrizable space of probability laws on
\(T^n\times D\). These compactness properties provide the ambient
topological setting for the mechanism spaces and the equilibrium
argument below.

\subsubsection*{Joint Type Law Assumptions}

\begin{assumption}[Continuous density of the joint type law]
\label{assump:absolute-continuity}
Let
\[
\overline H
:=
\bigotimes_{j=1}^N H_j
\]
denote the product of the team-level marginal type laws. The joint type
law \(H\) is absolutely continuous with respect to \(\overline H\), and
its Radon--Nikodym derivative
\[
h
:=
\frac{dH}{d\overline H}
\]
admits a bounded continuous version on \(T^{nN}\).
\end{assumption}
Under Assumption~\ref{assump:absolute-continuity}, the joint type law
admits the density representation
\[
H(d\widehat{\mathbf t})
=
h(\widehat{\mathbf t})\,
\overline H(d\widehat{\mathbf t}).
\]
This assumption permits correlated types across teams while giving a
continuous density representation relative to the product of team
marginals, a regularity property used below when verifying continuity of
payoff and incentive-evaluation maps.

For the interim incentive constraints, we also impose the corresponding
absolute-continuity condition within each team.

\begin{assumption}[Uniformly positive team-type densities]
\label{assump:team-type-density}
For each team \(j\), the type law \(H_j\) has a bounded continuous
density \(h_j\) relative to
\(\overline H_j=\bigotimes_{i=1}^n H_{i,j}\). There is a constant
\(c_j>0\) such that
\[
h_j(\mathbf t_j)\ge c_j
\quad\text{for every }\mathbf t_j\in T^n.
\]
\end{assumption}

Under Assumption~\ref{assump:team-type-density}, each team type law
has the representation
\[
H_j(d\mathbf t_j)
=
h_j(\mathbf t_j)\,\overline H_j(d\mathbf t_j).
\]
This permits correlation among agents' types within a team. The positive
lower bound makes \(H_j\) equivalent to \(\overline H_j\), so that the
fixed-marginal laws of mechanisms can be reweighted continuously
between the two reference measures.

\subsubsection*{Utility Function Assumptions}

\begin{assumption}[Bounded continuous utility functions]
\label{assump:bounded-continuous-utilities}
The agents' utility function
\[
u:T^{nN}\times D^N\to\mathbb R
\]
and the principals' utility functions
\[
\pi_j:T^{nN}\times D^N\to\mathbb R,
\qquad j=1,\ldots,N,
\]
are bounded and continuous.
\end{assumption}

\subsubsection*{Baseline-Dominant Team Decisions}

The next assumption guarantees that each principal has a simple seed
mechanism that is incentive compatible independently of the mechanisms
chosen by the other principals. The assumption requires a team decision
under which each agent's recommended individual decision is a best
response, uniformly over realized types and the decisions generated by
the other teams.

\begin{assumption}[Baseline-dominant team decision]
\label{assump:baseline-dominant-team-decision}
For each team \(j\), there exists a team decision
\[
d_j^0=(d_{p,j}^0,d_{1,j}^0,\ldots,d_{n,j}^0)\in D
\]
such that, for every agent \((i,j)\), for \(H\)-almost every
\(\widehat{\mathbf t}\in T^{nN}\), and for every profile of other-team
decision vectors \(\mathbf d_{-j}\in D^{N-1}\),
\[
d_{i,j}^0
\in
\operatorname*{arg\,max}_{d\in D_a}
u_{i,j}\!\left(
\widehat{\mathbf t},
\mathbf d_{-j},
(d_{p,j}^0,\mathbf d_{-i,j}^0,d)
\right).
\]
\end{assumption}

Assumption~\ref{assump:baseline-dominant-team-decision} does not impose
an order on the decision space, separability of payoffs, or monotonicity
of utility in decisions. It only requires that each team have one
baseline decision vector under which every agent's recommended
individual decision is a best response, no matter what decisions are
generated by the other teams. Economically, \(d_{i,j}^0\) can be read as
a minimal-cost or default individual decision: when team \(j\)'s
team-level decision is fixed at \(d_{p,j}^0\) and teammates follow their
baseline recommendations, agent \((i,j)\) cannot gain by choosing any
other individual decision.

This assumption yields a canonical seed mechanism. For each team \(j\),
define \(m_j^0\in M\) by
\[
m_j^0(\cdot\mid\mathbf t'_j):=\delta_{d_j^0}
\quad\text{for every }\mathbf t'_j\in T^n.
\]
Thus \(m_j^0\) generates the same baseline team decision regardless of
the reported type profile.

\subsection*{Mechanism Strategy Spaces and Induced Laws}

We next introduce the principals' mechanism strategy spaces and the
topology used in the equilibrium analysis. We begin with the primitive
mechanism space and equip it with the Young--narrow topology on stochastic
kernels.\footnote{For background on weak convergence and Young--narrow
topologies for transition probabilities and strategic measures, see
\citet{Balder1988}, \citet{balder2021new},
\citet{milgrom1985distributional}, and \citet{saldi2026kernel}.} Under
this topology, a mechanism
\[
m_j\in M_j
\]
is identified with the fixed-marginal joint law
\[
H_j(d\mathbf t_j)\,m_j(dd_j\mid\mathbf t_j)
\in\mathcal P(T^n\times D).
\]
Thus convergence of mechanisms is narrow convergence of their associated
fixed-marginal type--decision laws.

After defining the topology on the individual and profile mechanism
spaces, we introduce the truthful-obedient laws used to evaluate
principals' payoffs. We then turn to incentive compatibility. Using an
agent-by-agent reduction in the style of \citet*{kadan2017}, we define agents' deviation values, regrets, and the resulting
\(\varepsilon\)-incentive-compatible correspondences. These law-valued
objects and deviation rules are induced by the mechanism profile and are
used only in evaluating payoffs and incentive constraints.

\subsubsection*{Principal Strategy Spaces}

We first define the topology used on stochastic kernels and then apply it
to the principals' mechanism spaces.

\begin{definition}[Young--narrow mechanism space]
\label{def:young-narrow}
Let \(S\) be a Polish space, let \(A\) be a compact metric space, and
fix a probability measure \(\rho\in\mathcal P(S)\). Identify two kernels
in \(\mathcal K(S,A)\) when they agree \(\rho\)-almost everywhere,
equivalently when they induce the same fixed-marginal law
\(\rho(ds)K(da\mid s)\). The \emph{Young--narrow topology} on the
resulting equivalence classes is the topology inherited from narrow
convergence of these fixed-marginal laws. Equivalently, it is the
coarsest topology making the maps
\[
\mathcal K(S,A)\ni K
\longmapsto
\int_S\int_A f(s,a)\,K(da\mid s)\rho(ds)
\in\mathbb R
\]
continuous for every bounded continuous \(f\in C_b(S\times A)\).
Thus \(K^\ell\to K\) if and only if the fixed-marginal joint laws
\[
\rho(ds)K^\ell(da\mid s)
\]
converge narrowly to
\[
\rho(ds)K(da\mid s)
\]
in \(\mathcal P(S\times A)\).
\end{definition}

Every principal uses the same primitive kernel form
\(\overline M=\mathcal K(T^n,D)\). For principal \(j\), define \(M_j\)
by identifying kernel representatives that agree \(H_j\)-almost
everywhere and equipping the resulting classes with the Young--narrow
topology induced by \(H_j\). Equivalently,
\[
M_j
\cong
\left\{\nu\in\mathcal P(T^n\times D):
(\operatorname{pr}_{T^n})_\#\nu=H_j\right\}.
\]
Thus \(m_j^\ell\to m_j\) in \(M_j\) if and only if
the fixed-marginal laws
\[
H_j(d\mathbf t_j)\,m_j^\ell(dd'_j\mid\mathbf t_j)
\]
converge narrowly to
\[
H_j(d\mathbf t_j)\,m_j(dd'_j\mid\mathbf t_j)
\]
in \(\mathcal P(T^n\times D)\). The notation \(M_j\) records the
team-specific reference marginal used for the almost-everywhere
identification, while the primitive kernel form is common across teams.

Because \(T^n\) and \(D\) are compact metric spaces, the set of
probability laws on \(T^n\times D\) with fixed \(T^n\)-marginal \(H_j\)
is compact and metrizable under narrow convergence. Its identification
with \(M_j\) therefore makes \(M_j\) compact and metrizable under its
Young--narrow topology. It is also convex: if \(m^1,m^2\in M_j\) and
\(\lambda\in[0,1]\), then
\[
\lambda m^1+(1-\lambda)m^2\in M_j.
\]

The primitive mechanism-profile space is
\[
\mathcal M:=\prod_{j=1}^N M_j,
\]
equipped with the product of the coordinatewise Young--narrow
topologies. If \(d_{M,j}\) is any metric that metrizes the
Young--narrow topology induced by \(H_j\) on \(M_j\),
then the product topology on \(\mathcal M\) is metrized by
\[
d_{\mathcal M}(\mathbf m,\mathbf m')
:=
\max_{j=1,\ldots,N}d_{M,j}(m_j,m_j').
\]
Consequently, \(\mathcal M\) is a nonempty compact, convex, and
metrizable subset of a locally convex Hausdorff space.

\paragraph{Primitive equilibrium and law-based analysis.}
The principals choose primitive mechanisms, so the equilibrium concept
of Definition~\ref{def:BNPE} is stated on \(\mathcal M\). The
Young--narrow topology is introduced to study the payoff and incentive
consequences of these choices without comparing mechanism kernels
pointwise. It follows the strategic-measure approach of
\citet{milgrom1985distributional} and \citet{Balder1988}, as well as the
mechanism-space topology used by \citet*{kadan2017}, while preserving
the product structure required by the generalized game among
principals.

The truthful-obedient laws, report kernels, and relaxed decision rules
introduced below are induced or auxiliary objects used to evaluate
payoffs and incentive constraints; they are not additional strategic
coordinates. The general existence result analyzes the
\(\varepsilon\)-incentive-compatible correspondence through these objects.
The exact-equilibrium result specializes the model to the embedded
one-member quasi-linear contest described in
Section~\ref{sec:exact-equilibrium}.  In that environment, concavity of
the contest technology and a uniformly profitable pooling output provide
the budget slack used to establish continuity of the exact feasible
correspondence.

\subsubsection*{Evaluation of Mechanism Profiles}

Although the primitive mechanism space is topologized by the
Young--narrow topology, principals' payoffs and agents' incentive
constraints are evaluated using objects induced by a mechanism profile.
We first define the truthful-obedient laws used to evaluate principals'
payoffs. The following subsection then introduces the report kernels and
relaxed decision rules used to evaluate agents' unilateral deviations.

For each team \(j\), define the team-level truthful-obedient outcome
space by
\[
X_j:=T^n\times D.
\]
A generic element
\[
x_j=(\mathbf t_j,d'_j)\in X_j
\]
records team \(j\)'s true type profile and the team decision vector
generated by its mechanism under truthful reporting. Writing
\[
d'_j=(d_{p,j},d'_{1,j},\ldots,d'_{n,j}),
\]
the coordinate \(d_{p,j}\) is the team-level decision and
\(d'_{i,j}\) is agent \((i,j)\)'s mechanism-generated individual
recommendation. Under truthful-obedient play, each agent's actual
individual decision coincides with its recommendation.

For a mechanism \(m_j\in M_j\), the corresponding team-level
truthful-obedient law is
\[
\mu_j(m_j)\in\mathcal P(X_j),
\]
defined by
\begin{equation}
\label{eq:mu-j}
\mu_j(m_j)(d\mathbf t_j,dd'_j)
:=
H_j(d\mathbf t_j)\,
m_j(dd'_j\mid\mathbf t_j).
\end{equation}
This is precisely the fixed-marginal law representing an element of
\(M_j\).

The aggregate truthful-obedient outcome space is
\[
X:=T^{nN}\times D^N.
\]
For a mechanism profile
\[
\mathbf m=(m_1,\ldots,m_N)\in\mathcal M=\prod_{j=1}^N M_j,
\]
the aggregate truthful-obedient law
\[
\mu(\mathbf m)\in\mathcal P(X)
\]
is defined by
\begin{equation}
\label{eq:mu}
\mu(\mathbf m)
:=
H(d\widehat{\mathbf t})\,
\bigotimes_{j=1}^N
m_j(dd'_j\mid\mathbf t_j).
\end{equation}
Thus \(\mu(\mathbf m)\) records the joint distribution of the full type
profile and the mechanism-generated team decision vectors when all
agents report truthfully and follow their recommendations.

If principal \(j\)'s primitive payoff is \(\pi_j\), its expected payoff
at \(\mathbf m\) is evaluated under the truthful-obedient law:
\[
V_j(\mathbf m)
:=
\int_X \pi_j(x)\,\mu(\mathbf m)(dx).
\]

The truthful-obedient law does not by itself determine whether a
mechanism is incentive compatible, because incentive compatibility
requires comparison with admissible unilateral deviations. The
corresponding deviation values and incentive-compatible mechanism
correspondences are introduced in the analysis below.

\subsubsection*{Topology on Induced Laws}

All spaces of probability laws appearing below are equipped with the
narrow topology. Thus a sequence
\(\{\nu^\ell\}\subseteq\mathcal P(S)\) converges narrowly to
\(\nu\in\mathcal P(S)\) if
\[
\int_S f\,d\nu^\ell
\longrightarrow
\int_S f\,d\nu
\]
for every bounded continuous function \(f\in C_b(S)\). Because the
primitive outcome spaces are compact metric spaces, their
probability-law spaces are compact and metrizable under narrow
convergence.

The Young--narrow topology on each \(M_j\) introduced above is precisely the
topology obtained by applying narrow convergence to the fixed-marginal
laws
\[
H_j(d\mathbf t_j)\,m_j(dd'_j\mid\mathbf t_j).
\]
The deviation values and incentive-compatible mechanism
correspondences are defined and analyzed directly below.

\section{Equilibrium Existence}
\label{sec:abstract-existence}

The principals' interaction is a generalized game. Each principal chooses
a mechanism from the same primitive mechanism space
\[
M_j\cong\left\{\nu\in\mathcal P(T^n\times D):
(\operatorname{pr}_{T^n})_\#\nu=H_j\right\},
\]
but the mechanisms available as incentive-compatible choices depend on
the mechanisms selected by the other principals. Thus the underlying
strategy space is fixed, while the feasible-strategy correspondence is
endogenous.

We work directly with the primitive mechanism-profile space
\[
\mathcal M:=\prod_{j=1}^N M_j.
\]
For principal \(j\), write
\[
\mathbf m_{-j}\in M_{-j}:=\prod_{k\ne j}M_k
\]
for an opponents' mechanism profile. Given \(\mathbf m_{-j}\), principal
\(j\)'s feasible choices are described by the
\(\varepsilon\)-incentive-compatible correspondence
\[
IC_j^\varepsilon(\mathbf m_{-j})\subseteq M_j.
\]
The generalized-game feature enters through the dependence of this
correspondence on \(\mathbf m_{-j}\). Principal \(j\)'s objective is its
expected payoff \(V_j(m_j,\mathbf m_{-j})\), evaluated under the
truthful-obedient law induced by the mechanism profile.

We first establish a general
existence result under \(\varepsilon\)-relaxed incentive compatibility.
For every \(\varepsilon>0\), principals optimize exactly over mechanisms
in \(IC_j^\varepsilon(\mathbf m_{-j})\); only the agents' incentive
constraints are relaxed. We then return to the one-agent-per-principal structure of Myerson's
example and ask when exact equilibrium can be restored. In a contest
in which principals choose outputs and prize shares, concavity of the
contest technology and a uniformly profitable pooling output give each
principal room to adjust its mechanism when opponents change theirs.
Exact feasibility then no longer shrinks abruptly under small changes
in the opponents' mechanisms, as it does in Myerson's example, and an
exact equilibrium exists.

The fixed-point argument uses the compact and convex structure already
present in the primitive mechanism spaces. As shown in
Lemma~\ref{lem:mechanism-space-compact-convex} in
Appendix~\ref{app:epsilon}, the Young--narrow representation identifies
each \(M_j\) with a nonempty compact, convex, and metrizable subset
of a locally convex Hausdorff topological vector space. Consequently,
\(\mathcal M=\prod_jM_j\) has the same properties.
The remaining work is to establish the corresponding properties of the
endogenous feasible sets and the principals' objectives. For the
\(\varepsilon\)-relaxed result, the appendix evaluates each
agent's deviations through an agent-by-agent reduction in the style of
\citet*{kadan2017}, who show that a single principal's multi-agent
mechanism is incentive compatible exactly when each induced
single-agent problem is (their Proposition~13.4). This reduction is used to prove
nonemptiness, compactness, convexity, and continuity of
\(IC_j^\varepsilon\), together with continuity and quasiconcavity of
\(V_j\). We begin by defining the deviation values, regrets, and
\(\varepsilon\)-incentive-compatible correspondences used in that
argument.

\subsection{Approximate Equilibrium}

\begin{namedtheorem}[Approximate equilibrium]
\label{thm:approximate-equilibrium}
Under the maintained assumptions of
Section~\ref{assumptionssection}, for every \(\varepsilon>0\) there
exists an \(\varepsilon\)-relaxed Bayesian--Nash principals'
equilibrium
\[
\mathbf m^\varepsilon\in\mathcal M.
\]
The equilibrium classes admit kernel representatives forming an
equilibrium when agents' deviation values are computed over all
measurable reporting and subsequent action strategies.
\end{namedtheorem}

The theorem is proved as Lemma~\ref{lem:epsilon-BNPE} in
Appendix~\ref{app:epsilon}; its final statement follows from
Lemma~\ref{lem:safe-completion}. The proof treats the interaction among
principals as a generalized game on the compact convex strategy space
\(\mathcal M=\prod_jM_j\). Nonemptiness of
\(IC_j^\varepsilon(\mathbf m_{-j})\) follows from the
baseline-dominant seed mechanism of Lemma~\ref{lem:seed-IC}. The
appendix then establishes closed graph, convex-valuedness, and lower
hemicontinuity of the \(\varepsilon\)-incentive-compatible
correspondence, together with continuity and quasiconcavity of principal
payoffs. Berge's maximum theorem therefore gives well-behaved
best-response correspondences, and the
Kakutani--Fan--Glicksberg theorem yields an equilibrium.

The restriction \(\varepsilon>0\) is used in establishing lower
hemicontinuity of the incentive-compatible correspondence. Mixing a
feasible mechanism with the baseline-dominant seed creates strict
\(\varepsilon\)-slack, which persists under sufficiently small changes
in the opponents' mechanisms. This argument does not generally extend
to \(\varepsilon=0\), where the incentive constraints may bind.

\subsubsection*{Comparison with Myerson's Proposition~4}
\label{sec:finite-example}

Theorem~\ref{thm:approximate-equilibrium} and Myerson's Proposition~4
address the same underlying nonexistence problem in different ways. Both
consider generalized principal--agent games in which the mechanisms
available to one principal depend on the mechanisms selected by the
others. Their assumptions and equilibrium concepts, however, are
different, so the two results are not nested.

\paragraph{Primitive environments.}
Myerson's Proposition~4 is formulated for finite type and decision
spaces. Theorem~\ref{thm:approximate-equilibrium} permits the compact
metric type, action, and outcome spaces of
Assumption~\ref{assump:compact-spaces}; finite spaces are included as a
special case of this topological framework.

\paragraph{Assumptions.}
Finiteness supplies the compactness and continuity needed in Myerson's
argument without additional topological regularity conditions. Our
general theorem instead imposes the density conditions in
Assumptions~\ref{assump:absolute-continuity}
and~\ref{assump:team-type-density}. In a finite environment, the first
holds automatically and the second requires that every profile of
within-team types with positive marginal probabilities have positive
probability; both hold, for example, when types are independent and
every type has positive probability.

Theorem~\ref{thm:approximate-equilibrium} also assumes a
baseline-dominant team decision
(Assumption~\ref{assump:baseline-dominant-team-decision}). This is a
substantive restriction in environments with an obedience margin,
because it requires a recommendation that no agent wishes to disobey,
uniformly over types and other teams' decisions. It holds automatically
when agents have no separate private action, as in Myerson's leading
example below.

\paragraph{Equilibrium concepts.}
As discussed in Section~\ref{MyersonEx}, quasi-equilibrium keeps
incentive compatibility exact and relaxes the principals' optimality
requirement, whereas \(\varepsilon\)-relaxed equilibrium keeps
principals' optimization exact and relaxes the agents' incentive
constraints, with a tolerance on ex ante regret that is normalized
differently from Myerson's type-by-type tolerance. That section also shows that a limit
of \(\varepsilon\)-relaxed equilibria in Myerson's example can be a
quasi-equilibrium without being an exact principals' equilibrium.

\paragraph{Application to Myerson's example.}
We now verify that Myerson's example satisfies the maintained
assumptions and apply Theorem~\ref{thm:approximate-equilibrium}
directly.

\paragraph{Environment.}
There are \(N=2\) teams, each containing a single agent, so \(n=1\).
Each agent's type space is
\[
T=\{\theta_A,\theta_B\},
\]
with
\[
H_j(\theta_A)=H_j(\theta_B)=\tfrac12,
\]
and types are independent across teams.

The mechanism directly implements one of the decisions \(A\), \(B\), or
\(C\). Accordingly,
\[
D_p=\{A,B,C\},
\qquad
D_a=\{\ast\},
\qquad
D=D_p\times D_a\cong\{A,B,C\},
\]
where \(\ast\) is a dummy individual action, reflecting that the agent
makes no decision after the mechanism selects the outcome. A mechanism
is therefore a kernel
\[
m_j\in\overline M=\mathcal K(T,D);
\]
because every type has positive probability, the identification that
defines \(M_j\) is trivial here, and \(M_j=\overline M\).

Agent and principal utilities are given by
\[
\begin{array}{c|cc}
 & \theta_A & \theta_B \\\hline
A & (1,6) & (z_j,0) \\
B & (z_j,0) & (1,6) \\
C & (0,5) & (0,5),
\end{array}
\]
where each ordered pair reports the agent's payoff first and the
principal's payoff second. The variable \(z_j\) depends on the other
team's realized decision. For team \(1\),
\[
z_1=
\begin{cases}
2, & \text{if team \(2\)'s decision is \(A\) or \(B\)},\\
1, & \text{if team \(2\)'s decision is \(C\)},
\end{cases}
\]
whereas for team \(2\),
\[
z_2=
\begin{cases}
2, & \text{if team \(1\)'s decision is \(C\)},\\
1, & \text{if team \(1\)'s decision is \(A\) or \(B\)}.
\end{cases}
\]

Since \(T\), \(D_p\), and \(D_a\) are finite,
Assumption~\ref{assump:compact-spaces} holds automatically. Every
function on the primitive spaces is bounded and continuous, so
Assumption~\ref{assump:bounded-continuous-utilities} also holds.
Finally, independence and full support imply the maintained density
conditions.

\paragraph{Baseline dominance holds automatically.}
In Myerson's example, the decisions \(A\), \(B\), and \(C\) are
implemented directly by the mechanism; they are not private actions
chosen by the agent after receiving a recommendation. Thus
Assumption~\ref{assump:baseline-dominant-team-decision} restricts only
the singleton individual-action coordinate.

Because \(D_a=\{\ast\}\), for every \(d_{p,j}^0\in D_p\),
\[
\ast
\in
\operatorname*{arg\,max}_{d\in D_a}
u_{i,j}\!\left(
\widehat{\mathbf t},
\mathbf d_{-j},
(d_{p,j}^0,d)
\right)
\]
for every type profile and every decision of the other team. Baseline
dominance therefore holds automatically.

Choosing any fixed \(d_{p,j}^0\), and in particular \(C\), gives the
constant mechanism
\[
m_j^0(\cdot\mid t')
=
\delta_{(d_{p,j}^0,\ast)}
\qquad
\text{for every }t'\in T.
\]
Because this mechanism is report independent and there is no separate
action deviation, it is exactly incentive compatible against every
opponents' mechanism profile. Thus Myerson's example satisfies the seed
condition used in Theorem~\ref{thm:approximate-equilibrium}.

\begin{corollary}[Existence in Myerson's environment]
\label{cor:myerson-existence}
In the environment above, for every \(\varepsilon>0\), there exists an
\(\varepsilon\)-relaxed Bayesian--Nash principals' equilibrium.
\end{corollary}

\begin{proof}
The compactness, density, and utility assumptions hold as established
above. Because \(D_a\) is a singleton, the baseline-dominance assumption
holds automatically. The constant mechanism constructed above is report
independent and therefore exactly incentive compatible against every
opponents' mechanism profile. Hence Lemma~\ref{lem:seed-IC} supplies a
nonempty \(\varepsilon\)-incentive-compatible choice set for every
principal and every opponents' profile.
Theorem~\ref{thm:approximate-equilibrium} therefore applies directly.
\end{proof}

\paragraph{Discussion.}
Myerson's Proposition~3 illustrates the failure of lower
hemicontinuity directly. If the opponent's mechanism assigns zero
probability to the relevant adverse decision, full separation is
optimal and exactly incentive compatible. Once the opponent assigns any
positive probability to that decision, the optimal exactly
incentive-compatible mechanism jumps to the constant-\(C\) mechanism.
Thus feasible mechanisms at the limiting opponents' profile cannot
generally be approximated by corresponding feasible mechanisms at
nearby profiles.

The \(\varepsilon\)-incentive-compatible correspondence removes this
obstruction for every \(\varepsilon>0\). The interior-push argument of
Lemma~\ref{lem:epsilon-ic-lhc} establishes lower hemicontinuity, even
though it may fail at \(\varepsilon=0\). The worked example in
Section~\ref{MyersonEx} calculates the resulting relaxed equilibria and
relates them to Myerson's quasi-equilibrium construction. The limiting
profile is exact IC and is itself a quasi-equilibrium, but it is not an
exact principals' equilibrium; the example therefore also illustrates
nonuniqueness of quasi-equilibrium.

\section{Exact Equilibrium}
\label{sec:exact-equilibrium}

The preceding section establishes equilibrium existence in the general
model by allowing arbitrarily small ex ante expected incentive regret.
We now return to the one-agent-per-principal architecture underlying
Myerson's example and ask which additional economic structure restores
exact equilibrium.  We study a structured contest version of that
environment in which principals compete for team winnings and divide
those winnings with their privately informed members.  Concavity of the
contest technology and the availability of a uniformly profitable
pooling output provide the budget slack needed to make the exact
feasible-mechanism correspondence continuous.

The result applies to teams containing one privately informed member.
This restriction removes the joint reward-delivery problem that arises
with several members while preserving the defining strategic features
of Myerson's setting: principals choose mechanisms simultaneously,
agents report private information, and the feasibility of one
principal's mechanism depends on the mechanisms selected by the other
principals through the contest for team winnings.

There are \(N\ge2\) teams. Team \(j\)'s member has type
\(t_j\in T=[0,1]\), drawn independently across teams from a nonatomic law
\(H_j\) that is absolutely continuous with respect to Lebesgue measure
and has full support on \([0,1]\).
The principal assigns an output \(y_j\in Y=[0,\bar y]\). Producing output
\(y_j\) costs type \(t_j\) the amount \(c(t_j)y_j\), where
\(c:T\to\mathbb R_{++}\) is continuously differentiable and uniformly
strictly decreasing:
\[
c'(t)\le-\underline c<0.
\]
There is no subsequent private action, so the member's only deviation is
a misreport.

The teams divide a fixed prize \(V>0\). Given the output profile
\(\mathbf y\in Y^N\), team \(j\)'s winnings are
\[
w_j(\mathbf y)=V\sigma_j(\mathbf y),
\qquad
\sum_{j=1}^N\sigma_j(\mathbf y)=1.
\]
The member receives a reward \(p_j\in[0,w_j]\) and has utility
\(p_j-c(t_j)y_j\); principal \(j\) keeps the residual
\(w_j-p_j\). Thus rewards are quasi-linear for both parties, but are
constrained pointwise by realized team winnings. Participation is
voluntary: each member can decline and obtain utility zero, so
mechanisms must be individually rational at the interim stage. Without
this constraint the problem would be degenerate, since assigning maximal
output with no reward after every report is exactly incentive compatible
and weakly dominant for every principal.

This contest is embedded in the general mechanism game by taking the
team-level decision to be an output/share pair
\((y_j,\gamma_j)\in Y\times[0,1]\), with realized reward
\(p_j=\gamma_jw_j(\mathbf y)\), and by making the individual-action
space a singleton. An unrestricted deviation in the specialized
primitive game is therefore an incentive-compatible, individually
rational stochastic kernel from reports into output/share pairs.

We impose two additional primitive conditions. First, for every fixed
rival-output profile, \(y_j\mapsto\sigma_j(y_j,\mathbf y_{-j})\) is
nondecreasing, concave, and uniformly Lipschitz; for every rival
\(j'\ne j\), the share \(\sigma_j\) is nonincreasing in \(y_{j'}\).
Second, there is an output \(y^0\in(0,\bar y]\) satisfying
\[
V\sigma_j(y^0,\bar y,\ldots,\bar y)>c(0)y^0
\qquad\text{for every }j.
\]
The second condition supplies a pooling mechanism that is strictly
feasible uniformly over the rivals' mechanisms. Smoothed Tullock shares
provide a leading example; Appendix~\ref{app:exact-contest} gives a
primitive verification.

Strictly decreasing marginal cost gives single crossing. Consequently,
truthful reporting is characterized by a nondecreasing output rule and
the envelope formula. Quasi-linearity then reduces reward feasibility to
whether the envelope-implied interim mean reward lies between zero and
interim expected winnings. Whenever it does, a proportional
win-contingent reward schedule delivers that mean while satisfying
\(0\le p_j\le w_j\) after every realization.

The profitable pooling mechanism creates strict \emph{budget} slack,
not strict incentive slack. Mixing any feasible mechanism toward this
seed moves it into the interior of the reward-feasible set. That margin
persists under nearby opponents' mechanisms and yields lower
hemicontinuity of the exact feasible-mechanism correspondence.

\begin{namedtheorem}[Exact equilibrium in the one-member contest]
\label{thm:exact-contest-main}
In the one-member contest environment above, with interim individual
rationality, suppose the contest shares
satisfy the stated monotonicity, concavity, and uniform Lipschitz
conditions and that the profitable pooling condition holds. Then there
exists an exact Bayesian--Nash principals' equilibrium in the embedded
primitive mechanism game, including against all stochastic output/share
kernel deviations. It is implementable by monotone output rules and
pointwise feasible win-contingent reward schedules.
\end{namedtheorem}

Appendix~\ref{app:exact-contest} proves the theorem. The proof first
uses a conditional-mean argument to show that every stochastic
output/share deviation is weakly dominated by a deterministic mechanism.
It then parameterizes incentive-compatible deterministic mechanisms by a
monotone output rule, with the lowest type's rent normalized to zero by
individual rationality, establishes compactness
of that reduced strategy space, characterizes pointwise reward delivery, proves
continuity of the exact feasible correspondence by the budget-slack
interior push, and applies Berge's maximum theorem and the
Kakutani--Fan--Glicksberg fixed-point theorem.

The restriction to one member per team is substantive. With several
members, the individual envelope formulas determine several interim
reward means that must be delivered simultaneously from one realized
team prize. Nonnegativity and an average budget inequality alone do not
generally characterize joint deliverability. Extending the exact theorem
to \(n>1\) therefore requires a separate interim-feasibility result for
the vector of member rewards. With several members, taking the team-level
decision to consist of member outputs and a reward-share vector in the
simplex, and taking the individual-action space to be a singleton,
embeds the contest in the general model. The \(\varepsilon\)-relaxed
theorem applies whenever the maintained assumptions of
Section~\ref{assumptionssection} hold.

\section{Conclusion}
When multiple principals simultaneously design mechanisms for their
respective teams, strategic interaction gives rise to a generalized game
in which each principal's incentive-compatible feasible set depends
endogenously on the mechanisms selected by competing principals. As
\citet{myerson1982optimal} shows, this endogenous feasibility correspondence
need not be lower hemicontinuous, and exact principals' equilibrium can
fail to exist as a result.

Our first result develops an existence theory for a relaxed equilibrium
concept in which principals best respond exactly over mechanisms whose
incentive constraints have ex ante expected regret at most
\(\varepsilon\). Working directly with primitive stochastic mechanism
kernels, represented by their fixed-marginal laws under the
Young--narrow topology, Theorem~\ref{thm:approximate-equilibrium}
establishes an \(\varepsilon\)-relaxed Bayesian--Nash principals'
equilibrium for every \(\varepsilon>0\). The result permits compact
multidimensional type and decision spaces and accommodates both private
reports and subsequent private decisions. Its assumptions provide the
compactness, continuity, and exactly incentive-compatible seed needed by
the generalized-game fixed-point argument.

Section~\ref{sec:finite-example} compares this result with Myerson's
own Proposition~4, which establishes existence of a quasi-equilibrium
for finite multi-principal environments more generally, without a
baseline-dominance-type hypothesis. Myerson's proof uses equilibria of
relaxed-IC problems to construct a limiting exact-IC quasi-equilibrium;
our theorem instead establishes exact best responses over the relaxed
feasible set at every fixed \(\varepsilon>0\). Section~\ref{MyersonEx}
shows that a limit of the latter equilibria can be a quasi-equilibrium
without being an exact principals' equilibrium.

Our second result shows that the relaxation can be dispensed with in an
embedded quasi-linear contest environment with one privately informed
member per team and voluntary participation. There, single crossing and the envelope formula reduce
incentive compatibility to monotone output assignment. Concavity of the
contest technology permits every stochastic output/share deviation to be
replaced by a weakly better deterministic conditional-mean mechanism,
while a uniformly profitable pooling output supplies the budget slack
needed for continuity of the exact feasible correspondence. The fixed
point is therefore an exact equilibrium of the specialized primitive
mechanism game, not merely of its reduced monotone-allocation form.

The one-member restriction isolates the remaining obstacle to a broader
exact contest theorem. With several members, multiple interim reward
means must be delivered jointly from a single realized team prize.
Characterizing that joint interim-feasibility set, and establishing its
continuity as rival mechanisms vary, is the principal direction for
extending the exact result.

\appendix

\section{Approximate Equilibrium Existence}
\label{app:epsilon}

This appendix establishes existence of an
\(\varepsilon\)-relaxed Bayesian--Nash principals' equilibrium on the
primitive mechanism strategy space
\[
\mathcal M=\prod_{j=1}^N M_j.
\]
Each principal's strategy space \(M_j\), the set of \(H_j\)-equivalence
classes of kernels in \(\overline M=\mathcal K(T^n,D)\), is equipped
with the Young-narrow topology induced by the team type law \(H_j\).
Under the maintained compactness assumptions, \(M_j\) is a nonempty
compact convex subset of a locally convex Hausdorff topological vector
space.

The proof treats the principals' interaction as a generalized game. For
each opponents' profile \(\mathbf m_{-j}\), principal \(j\)'s feasible
set is the \(\varepsilon\)-relaxed incentive-compatible correspondence
\[
IC_j^\varepsilon(\mathbf m_{-j})
=
\left\{
m_j\in M_j:
D_{i,j}(m_j,\mathbf m_{-j})-T_{i,j}(m_j,\mathbf m_{-j})\le \varepsilon
\text{ for every } i=1,\ldots,n
\right\}.
\]
The tolerance bounds each agent's regret averaged over its own type;
at \(\varepsilon=0\), Lemma~\ref{lem:ic-common-conull-bridge}
identifies this condition with the interim IC definition in the model.
The appendix is organized around verifying that these correspondences are
nonempty, compact, convex, lower hemicontinuous, and have closed graph.
It also shows that principal payoff functions are continuous and
quasiconcave in each principal's own mechanism.

The key analytical step is the treatment of incentive compatibility.
Rather than assuming continuity of represented law maps, we verify the
needed continuity and convexity properties directly on the primitive
mechanism space. The truthful-obedient payoff is continuous by the
fixed-marginal Young-narrow representation. The deviation payoff is
handled using the single-agent reduction of \citet*{kadan2017},
which expresses each
agent's deviation problem through the agent-specific disintegration of
the product-marginal mechanism law. This formulation avoids pointwise
evaluation of mechanism kernels on null report histories and yields the
closedness and convexity properties needed for the
\(\varepsilon\)-IC correspondence.

Once these ingredients are in place, the generalized-game fixed-point
argument yields, for every \(\varepsilon>0\), a mechanism profile
\[
\mathbf m^\varepsilon\in\mathcal M
\]
such that each principal best responds over its
\(\varepsilon\)-relaxed incentive-compatible mechanism set. This is the
\(\varepsilon\)-relaxed Bayesian--Nash principals' equilibrium established
by the final fixed-point step.

We begin with the compactness and convexity of the primitive mechanism
spaces.

\begin{lemma}[Compactness and convexity of the mechanism space]
\label{lem:mechanism-space-compact-convex}
Suppose \(T\) and \(D\) are compact metric. For each team \(j\),
let \(M_j\) be the set of \(H_j\)-equivalence classes of kernels in
\(\overline M=\mathcal K(T^n,D)\),
equipped with the Young-narrow topology induced by the team type law
\(H_j\). Then \(M_j\) is a nonempty compact convex metrizable subset of a
locally convex Hausdorff topological vector space.
\end{lemma}

\begin{proof}
For each \(m_j\in M_j\), define the associated fixed-marginal joint law
\[
\mu_j(m_j)(d\mathbf t_j,dd'_j)
:=
H_j(d\mathbf t_j)m_j(dd'_j\mid \mathbf t_j)
\in\mathcal P(T^n\times D).
\]
Since \(T^n\) and \(D\) are compact metric, every law with
\(T^n\)-marginal \(H_j\) disintegrates as \(H_j\otimes m_j\) for some
kernel \(m_j\), unique up to \(H_j\)-null sets. This identifies \(M_j\),
under the Young-narrow topology induced by \(H_j\), with the
fixed-marginal set
\[
\mathfrak L_j
:=
\left\{
\nu\in\mathcal P(T^n\times D):
(\operatorname{pr}_{T^n})_\#\nu=H_j
\right\}.
\]

We first show that \(\mathfrak L_j\) is compact under narrow convergence.
Because \(H_j\) is fixed and \(T^n\times D\) is compact, the family
\(\mathfrak L_j\) is tight, so it is relatively compact by Prokhorov's
theorem.

The set \(\mathfrak L_j\) is closed under narrow convergence. If
\(\nu^\ell\to\nu\) narrowly and
\((\operatorname{pr}_{T^n})_\#\nu^\ell=H_j\) for every \(\ell\), then
continuity of the projection map implies
\[
(\operatorname{pr}_{T^n})_\#\nu^\ell
\to
(\operatorname{pr}_{T^n})_\#\nu
\]
narrowly in \(\mathcal P(T^n)\). Hence
\[
(\operatorname{pr}_{T^n})_\#\nu=H_j,
\]
so \(\nu\in\mathfrak L_j\). Therefore \(\mathfrak L_j\) is compact.

The set \(\mathfrak L_j\) is convex because the fixed-marginal constraint
is affine: if \(\nu^1,\nu^2\in\mathfrak L_j\) and
\(\lambda\in[0,1]\), then
\[
(\operatorname{pr}_{T^n})_\#
\bigl(\lambda\nu^1+(1-\lambda)\nu^2\bigr)
=
\lambda H_j+(1-\lambda)H_j
=
H_j.
\]
It is nonempty because, for any \(d^0\in D\), the constant kernel
\(m_j^0(\cdot\mid\mathbf t_j):=\delta_{d^0}\) generates an element of
\(\mathfrak L_j\).

Since \(T^n\times D\) is compact metric, \(\mathcal P(T^n\times D)\) is
metrizable under narrow convergence; hence the compact subspace
\(\mathfrak L_j\) is metrizable. Finally, \(\mathfrak L_j\) is a convex
subset of the locally convex Hausdorff topological vector space of finite
signed Borel measures on \(T^n\times D\), endowed with the topology of
convergence against bounded continuous functions. Through the
fixed-marginal identification above, the same compactness, convexity, and
metrizability properties hold for \(M_j\) under the Young-narrow
topology.
\end{proof}

\subsection{Nonemptiness of the Incentive-Compatible Choice Correspondence}
\label{app:ic-nonemptiness}

The approximate-existence theorem requires the
\(\varepsilon\)-incentive-compatible choice correspondence to be nonempty:
\[
IC_j^\varepsilon(\mathbf m_{-j})\neq\varnothing
\]
for every principal \(j\) and every opponents' mechanism profile
\(\mathbf m_{-j}\). This subsection shows that nonemptiness follows from
the maintained primitives. The argument uses the baseline seed mechanism
from Assumption~\ref{assump:baseline-dominant-team-decision}, which selects a
constant team decision vector
\[
d_j^0=(d_{p,j}^0,d_{1,j}^0,\ldots,d_{n,j}^0)\in D
\]
after every reported type profile.

\begin{lemma}[Seed mechanism is incentive compatible]
\label{lem:seed-IC}
Under Assumption~\ref{assump:baseline-dominant-team-decision}, for every
principal \(j\) and every opponents' mechanism profile
\[
\mathbf m_{-j}\in M_{-j},
\]
the seed mechanism
\[
m_j^0(\cdot\mid \mathbf t'_j):=\delta_{d_j^0}
\quad\text{for every }\mathbf t'_j\in T^n
\]
is incentive compatible against \(\mathbf m_{-j}\). Hence, for every
\(\varepsilon\ge 0\),
\[
m_j^0\in IC_j^\varepsilon(\mathbf m_{-j}).
\]
\end{lemma}

\begin{proof}
Fix \(j\), an opponents' mechanism profile \(\mathbf m_{-j}\), and an
agent \((i,j)\). Under the seed mechanism \(m_j^0\), the team decision
vector is constant:
\[
m_j^0(\cdot\mid \mathbf t'_j)=\delta_{d_j^0}
\quad\text{for every }\mathbf t'_j\in T^n.
\]
Thus agent \((i,j)\)'s report cannot affect the team-level decision
\(d_{p,j}^0\) or any individual recommendation in team \(j\). In
particular, regardless of the report, agent \((i,j)\)'s mechanism-generated
individual recommendation is \(d_{i,j}^0\).

Consider any type profile \(\widehat{\mathbf t}\) in the conull set on
which Assumption~\ref{assump:baseline-dominant-team-decision} holds, any realized
profile of other teams' decision vectors \(\mathbf d_{-j}\in D^{N-1}\),
and any alternative actual individual decision \(d\in D_a\). By
Assumption~\ref{assump:baseline-dominant-team-decision},
\[
u_{i,j}\!\left(
\widehat{\mathbf t},
\mathbf d_{-j},
(d_{p,j}^0,\mathbf d_{-i,j}^0,d_{i,j}^0)
\right)
\ge
u_{i,j}\!\left(
\widehat{\mathbf t},
\mathbf d_{-j},
(d_{p,j}^0,\mathbf d_{-i,j}^0,d)
\right).
\]
Since this inequality holds pointwise in
\((\widehat{\mathbf t},\mathbf d_{-j})\), integrating over the type law
and over the decision laws induced by \(\mathbf m_{-j}\) preserves the
inequality.

Therefore, after any report, following the seed recommendation
\(d_{i,j}^0\) gives agent \((i,j)\) at least as much expected utility as
any alternative actual individual decision. Because the report also does
not affect the decision vector generated by \(m_j^0\), no unilateral
report-and-decision deviation is profitable. Hence truthful reporting and
obedient play are incentive compatible for agent \((i,j)\).

The argument applies to every agent in team \(j\), so \(m_j^0\) is
incentive compatible against \(\mathbf m_{-j}\). Consequently
\(m_j^0\in IC_j(\mathbf m_{-j})\) in the model's interim sense, and in
fact the dominance established above holds for every report and for
almost every realized type, so it holds on the common conull set of
Lemma~\ref{lem:ic-common-conull-bridge}. By that lemma, exact IC in this
sense is equivalent to \(R_{i,j}(m_j^0,\mathbf m_{-j})\le0\). Since
\(\varepsilon\ge0\), this gives
\(R_{i,j}(m_j^0,\mathbf m_{-j})\le0\le\varepsilon\), i.e.\
\(m_j^0\in IC_j^\varepsilon(\mathbf m_{-j})\) for every \(\varepsilon\ge0\).
\end{proof}

\subsection{Continuity and Convexity of Incentive Constraints}
\label{app:evaluationmap}

This subsection verifies the topological properties of the
\(\varepsilon\)-incentive-compatible correspondence. The main point is that
we do not require continuity of the maximal deviation envelope. Instead,
we write incentive compatibility as an intersection of fixed-deviation
inequalities, each of which is continuous under the Young-narrow
topology.

\begin{remark}[Essential-supremum report deviations]
\label{rem:ess-sup-report-deviations}
The incentive-compatibility condition uses the essential supremum over
reported types, rather than a pointwise supremum over all reports. This
follows the distributional interpretation of mechanisms in
\citet*{kadan2017}: mechanisms are identified through their
fixed-marginal induced laws, so report-contingent responses are defined
only up to \(H_{i,j}\)-null sets. Consequently, reporting a single null
type has no version-independent payoff. The appropriate
version-independent report-deviation value is therefore the essential
supremum over reports.

This should not be interpreted as imposing a domination restriction on
the deviation outcome law. The absolute-continuity representation used
below is only a representation of the essential supremum over reports.
It applies to the report marginal only; the agent's decision rule after
receiving a recommendation is unrestricted, and no domination of the
deviation outcome law by the truthful-obedient law is imposed.
\end{remark}

\begin{corollary}[Report-kernel representation]
\label{cor:ess-sup-report-kernel}
Let \(G:T\times T\to\mathbb R\) be bounded and jointly measurable. Then
\[
\int_T
\operatorname*{ess\,sup}_{t'\in T}G(t,t')\,H_{i,j}(dt)
=
\sup_{\chi\in\mathcal K^{ac}_{i,j}}
\int_T\int_T
G(t,t')\,\chi(dt'\mid t)\,H_{i,j}(dt),
\]
where \(\mathcal K^{ac}_{i,j}\) is the set of stochastic kernels
\(\chi:T\times\mathcal B(T)\to[0,1]\) such that, for
\(H_{i,j}\)-almost every \(t\), there exists a measurable density
\(g_\chi(t,\cdot)\in L^1(H_{i,j})\) satisfying
\[
\chi(dt'\mid t)=g_\chi(t,t')\,H_{i,j}(dt').
\]
\end{corollary}

\begin{proof}
For every \(\chi\in\mathcal K^{ac}_{i,j}\), the inner integral is
bounded above by
\[
\operatorname*{ess\,sup}_{t'\in T}G(t,t')
\]
for \(H_{i,j}\)-almost every \(t\), so the supremum over
\(\mathcal K^{ac}_{i,j}\) is bounded above by the left-hand side.

For the reverse inequality, fix \(\delta>0\). With
\[
\bar G(t):=\operatorname*{ess\,sup}_{t'\in T}G(t,t'),
\]
the map \(t\mapsto\bar G(t)\) is measurable. Indeed, for every
\(c\in\mathbb Q\), Tonelli's theorem makes
\[
t\longmapsto
H_{i,j}\bigl(\{t'\in T:G(t,t')>c\}\bigr)
\]
measurable, and
\[
\bar G(t)
=
\inf\left\{
c\in\mathbb Q:
H_{i,j}\bigl(\{t'\in T:G(t,t')>c\}\bigr)=0
\right\}.
\]
Consequently, the set
\[
A_\delta
:=
\{(t,t')\in T\times T:G(t,t')>\bar G(t)-\delta\}.
\]
The set \(A_\delta\) is measurable, and
\(H_{i,j}(A_\delta(t))>0\) for every \(t\). Define
\[
\chi_\delta(B\mid t)
:=
\frac{H_{i,j}(B\cap A_\delta(t))}
{H_{i,j}(A_\delta(t))}
\qquad
\text{for }B\in\mathcal B(T).
\]
Then \(\chi_\delta\in\mathcal K^{ac}_{i,j}\): for each
\(B\in\mathcal B(T)\), the map
\[
t\longmapsto \chi_\delta(B\mid t)
\]
is measurable by Tonelli's theorem, since \(A_\delta\) is jointly
measurable and \(H_{i,j}(A_\delta(t))\) is measurable and strictly
positive; and \(\chi_\delta(\cdot\mid t)\ll H_{i,j}\) for every \(t\).
Moreover,
\[
\int_T\int_T
G(t,t')\,\chi_\delta(dt'\mid t)\,H_{i,j}(dt)
\ge
\int_T \bar G(t)\,H_{i,j}(dt)-\delta.
\]
Letting \(\delta\downarrow0\) gives the result.
\end{proof}

\subsection{Relaxed Deviations and Closedness of the
\texorpdfstring{\(\varepsilon\)}{epsilon}-IC Correspondence}
\label{sec:relaxed-deviations-closedness}

Throughout this subsection, fix an agent \((i,j)\). Following the
single-agent reduction of \citet{kadan2017}, use the product of the
individual type marginals as the reference law. Write
\[
\bar\mu_j(m_j)(d\mathbf t_j,dd'_j)
:=
\overline H_j(d\mathbf t_j)m_j(dd'_j\mid\mathbf t_j).
\]
Disintegrate this reference law from agent \((i,j)\)'s viewpoint:
\begin{equation}
\label{eq:krs-disintegration}
\bar\mu_j(m_j)(d\mathbf t_j,dd'_j)
=
H_{i,j}(dt_i)\,
\alpha^{m_j}_{i,j}\!\big(d(d_{p,j},d'_{i,j})\mid t_i\big)\,
\beta^{m_j}_{i,j}\!\big(
d(\mathbf t_{-i,j},\mathbf d'_{-i,j})
\mid t_i,d_{p,j},d'_{i,j}
\big),
\end{equation}
where
\[
\mathbf t_j=(t_i,\mathbf t_{-i,j})\in T^n,
\qquad
d'_j=(d_{p,j},d'_{i,j},\mathbf d'_{-i,j})\in D,
\]
and \(d'_{i,j}\in D_a\) is agent \((i,j)\)'s mechanism-generated
individual recommendation.
The truthful-obedient team law is
\(\mu_j(m_j)=h_j\bar\mu_j(m_j)\). In the deviation formulas below,
the report indexes the mechanism through \(\alpha^{m_j}_{i,j}\) and
\(\beta^{m_j}_{i,j}\), while \(h_j\) is evaluated at the agent's true
type and the other team members' types.

Write the team-\(k\) truthful-obedient law as
\[
\mu_k(m_k)(d\mathbf t_k,d\mathbf d_k)
:=
H_k(d\mathbf t_k)\,m_k(d\mathbf d_k\mid\mathbf t_k),
\]
and let the other-teams' truthful-obedient law be the product
\[
\mu_{-j}(\mathbf m_{-j})
:=
\bigotimes_{k\neq j}\mu_k(m_k)
\in\mathcal P\big(T^{n(N-1)}\times D^{N-1}\big),
\]
whose type marginal \(\bigotimes_{k\neq j}H_k\) does not depend on
\(\mathbf m_{-j}\). All cross-team coupling of types is carried explicitly
by the bounded continuous density \(h=dH/d\!\bigotimes_j H_j\).

Let the weighted reduced payoff, obtained by integrating out the
truthful-obedient decisions of the other teams and incorporating the
within-team type density, be
\begin{equation}
\label{eq:reduced-utility}
\begin{aligned}
\bar u_{i,j}\big(\mathbf t_j,d^a_j;\mathbf m_{-j}\big)
&:=
\int
u_{i,j}\big((\mathbf t_j,\mathbf t_{-j});d^a_j,\mathbf d_{-j}\big)\,
h_j(\mathbf t_j)\,
h(\mathbf t_j,\mathbf t_{-j})\,
\mu_{-j}(\mathbf m_{-j})\big(d(\mathbf t_{-j},\mathbf d_{-j})\big),
\\
&\qquad d^a_j:=(d_{p,j},a,\mathbf d'_{-i,j}).
\end{aligned}
\end{equation}
For fixed \(\mathbf m_{-j}\), this function is bounded and continuous in
\((d_{p,j},a,\mathbf d'_{-i,j})\), since \(h\) and \(h_j\) are bounded
and depend only on types.

A deviator with true type \(t\) observes, at the decision stage, the
tuple
\[
o:=(t,t',d_{p,j},d'_{i,j})
\in
\mathcal O:=T\times T\times D_p\times D_a,
\]
consisting of its type, its report, the team-level decision, and its own
individual recommendation. Define the weighted observed payoff
\begin{equation}
\label{eq:W-integrand}
W_{\mathbf m}(o,a)
:=
\int
\bar u_{i,j}\big((t,\mathbf t_{-i,j}),
(d_{p,j},a,\mathbf d'_{-i,j});\mathbf m_{-j}\big)\,
\beta^{m_j}_{i,j}\!\big(
d(\mathbf t_{-i,j},\mathbf d'_{-i,j})
\mid t',d_{p,j},d'_{i,j}
\big).
\end{equation}
For each \(\mathbf m\), \(W_{\mathbf m}\) is a Carathéodory integrand on
\(\mathcal O\times D_a\): measurable in \(o\) and continuous in \(a\),
with
\[
|W_{\mathbf m}|\le
\|u_{i,j}\|_\infty\,\|h\|_\infty\,\|h_j\|_\infty.
\]
For \(H_{i,j}\)-almost every \(t\), the truthful-obedient interim payoff is
\[
U_{i,j}(t;\mathbf m)
:=
\int_{D_p\times D_a}
W_{\mathbf m}\big((t,t,d_{p,j},d'_{i,j}),d'_{i,j}\big)\,
\alpha^{m_j}_{i,j}\!\big(d(d_{p,j},d'_{i,j})\mid t\big).
\]
Write
\(T_{i,j}(\mathbf m):=\int_T U_{i,j}(t;\mathbf m)\,H_{i,j}(dt)\).

\paragraph{Relaxed decision rules.}
Let \(\mathcal R_{i,j}\) denote the set of measurable relaxed decision
rules, that is, transition kernels
\[
\rho:\mathcal O\times\mathcal B(D_a)\to[0,1],
\qquad
\rho(da\mid t,t',d_{p,j},d'_{i,j}).
\]
Let
\[
\mathcal R^c_{i,j}\subseteq\mathcal R_{i,j}
\]
denote the continuous relaxed rules, namely those for which
\(o\mapsto\rho(\cdot\mid o)\) is weakly continuous on the full
observation space
\(\mathcal O=T\times T\times D_p\times D_a\). A measurable deterministic
decision rule
\[
d_{i,j}:\mathcal O\to D_a
\]
is identified with the Dirac kernel
\[
\rho(\cdot\mid o)=\delta_{d_{i,j}(o)}.
\]
Write \(\mathfrak D_{i,j}\) for this deterministic class.

For a report kernel \(\chi\in\mathcal K^{ac}_{i,j}\) and a relaxed
decision rule \(\rho\in\mathcal R_{i,j}\), define the relaxed deviation
value by
\begin{equation}
\label{eq:relaxed-deviation-value}
V^{\mathrm{dev}}_{i,j}(\mathbf m;\chi,\rho)
:=
\int_T\int_T\int_{D_p\times D_a}\int_{D_a}
W_{\mathbf m}(o,a)\,
\rho(da\mid o)\,
\alpha^{m_j}_{i,j}\!\big(d(d_{p,j},d'_{i,j})\mid t'\big)\,
\chi(dt'\mid t)\,H_{i,j}(dt).
\end{equation}
The fixed-deviation regret is
\[
r_{i,j}(\mathbf m;\chi,\rho)
:=
V^{\mathrm{dev}}_{i,j}(\mathbf m;\chi,\rho)
-
T_{i,j}(\mathbf m).
\]
The economic maximal deviation value uses all measurable deterministic
rules:
\begin{equation}
\label{eq:economic-envelope}
D_{i,j}(\mathbf m)
:=
\sup_{\chi\in\mathcal K^{ac}_{i,j}}
\sup_{d_{i,j}\in\mathfrak D_{i,j}}
V^{\mathrm{dev}}_{i,j}\big(\mathbf m;\chi,\delta_{d_{i,j}}\big).
\end{equation}
The associated \emph{envelope regret} is
\begin{equation}
\label{eq:envelope-regret}
R_{i,j}(\mathbf m)
:=
D_{i,j}(\mathbf m)-T_{i,j}(\mathbf m),
\end{equation}
so that incentive compatibility up to tolerance \(\varepsilon\) reads
\(R_{i,j}(\mathbf m)\le\varepsilon\). Thus
\begin{equation}
\label{eq:ic-eps-def}
IC_j^\varepsilon(\mathbf m_{-j})
=
\left\{
m_j\in M_j:
R_{i,j}(\mathbf m)\le\varepsilon
\text{ for all }i=1,\ldots,n
\right\}.
\end{equation}

\begin{remark}[Relaxation does not change incentive compatibility]
\label{rem:relaxation-not-restriction}
The passage to relaxed decision rules does not alter the
incentive-compatibility constraint. By Lemma~\ref{lem:value-density}
below, for each fixed mechanism profile \(\mathbf m\) and report kernel
\(\chi\), the supremum over measurable deterministic decision rules
equals the supremum over measurable relaxed decision rules and also
equals the supremum over continuous relaxed decision rules. The
relaxed-continuous class \(\mathcal R^c_{i,j}\) is therefore used only as
a proof device for closedness; it is not a restriction on the agent's
economic deviations, and in particular imposes no domination of the
realized deviation law by the truthful-obedient law.
\end{remark}

\begin{lemma}[Value density of continuous relaxed rules]
\label{lem:value-density}
Fix \(\mathbf m\in\mathcal M=\prod_{k=1}^N M_k\) and
\(\chi\in\mathcal K^{ac}_{i,j}\). Then
\[
\sup_{d_{i,j}\in\mathfrak D_{i,j}}
V^{\mathrm{dev}}_{i,j}\big(\mathbf m;\chi,\delta_{d_{i,j}}\big)
=
\sup_{\rho\in\mathcal R_{i,j}}
V^{\mathrm{dev}}_{i,j}(\mathbf m;\chi,\rho)
=
\sup_{\rho\in\mathcal R^c_{i,j}}
V^{\mathrm{dev}}_{i,j}(\mathbf m;\chi,\rho)
=
\int_{\mathcal O}\sup_{a\in D_a}W_{\mathbf m}(o,a)\,
\lambda_{\mathbf m,\chi}(do),
\]
where
\[
\lambda_{\mathbf m,\chi}(do)
=
\alpha^{m_j}_{i,j}\!\big(d(d_{p,j},d'_{i,j})\mid t'\big)\,
\chi(dt'\mid t)\,H_{i,j}(dt).
\]
\end{lemma}

\begin{proof}
\emph{Exactness of relaxation.} Since \(W_{\mathbf m}\) is a
Carathéodory integrand and \(D_a\) is compact metric, the measurable
maximum theorem \citep[Theorem~18.19]{aliprantis2006infinite} yields
measurability of
\[
o\longmapsto \sup_{a\in D_a}W_{\mathbf m}(o,a)
\]
and a \(\lambda_{\mathbf m,\chi}\)-measurable maximizer
\[
a^*(o)\in\operatorname*{arg\,max}_{a\in D_a}W_{\mathbf m}(o,a).
\]
Because \(\mathcal O\) and \(D_a\) are standard Borel spaces, this
selector has a Borel version that agrees with it
\(\lambda_{\mathbf m,\chi}\)-almost everywhere; use that version below.
The deterministic rule \(d_{i,j}=a^*\) attains
\[
\int_{\mathcal O}\sup_{a\in D_a}W_{\mathbf m}(o,a)\,
\lambda_{\mathbf m,\chi}(do),
\]
while
\[
\int_{\mathcal O}\int_{D_a}
W_{\mathbf m}(o,a)\,\rho(da\mid o)\,\lambda_{\mathbf m,\chi}(do)
\le
\int_{\mathcal O}\sup_{a\in D_a}W_{\mathbf m}(o,a)\,
\lambda_{\mathbf m,\chi}(do)
\]
for every \(\rho\in\mathcal R_{i,j}\). Hence the deterministic supremum
over \(\mathfrak D_{i,j}\) coincides with the relaxed supremum over
\(\mathcal R_{i,j}\), and both equal the displayed integral (the last term
of the chain). The remaining equality, with the supremum over the
continuous class \(\mathcal R^c_{i,j}\), is established next.

\emph{Stable density of the continuous class.} We show directly that
continuous relaxed rules attain the same supremum as all
measurable relaxed rules, by regularizing the deterministic maximizer
\(a^*\) from the exactness step.

Fix \(\eta>0\). Since \(\mathcal O\) is compact metric (hence Polish) and
\(D_a\) is compact metric (hence second countable), Lusin's theorem
\citep[Theorem~12.8]{aliprantis2006infinite} applied to the Borel map
\(a^*\) and the Borel probability measure \(\lambda_{\mathbf m,\chi}\)
yields a compact set \(K_\eta\subset\mathcal O\) with
\(\lambda_{\mathbf m,\chi}(\mathcal O\setminus K_\eta)<\eta\) on which
\(a^*\) is continuous. Composing with the Dirac embedding
\(a\mapsto\delta_a\), which is continuous from \(D_a\) into
\(\mathcal P(D_a)\) under the narrow topology, gives a continuous map
\(K_\eta\to\mathcal P(D_a)\), \(o\mapsto\delta_{a^*(o)}\). Since
\(\mathcal P(D_a)\) is a convex, compact, metrizable subset of the
locally convex space of finite signed Borel measures on \(D_a\),
Dugundji's extension theorem \citep{dugundji1951} extends this map to a
continuous function \(\rho_\eta:\mathcal O\to\mathcal P(D_a)\) agreeing
with \(\delta_{a^*(o)}\) on \(K_\eta\). In particular
\(\rho_\eta\in\mathcal R^c_{i,j}\): it is weakly continuous jointly in
all coordinates of \(o=(t,t',d_{p,j},d'_{i,j})\).

By construction \(\rho_\eta(\cdot\mid o)=\delta_{a^*(o)}\) for
\(o\in K_\eta\), so
\[
\left|
\int_{\mathcal O}\int_{D_a}W_{\mathbf m}(o,a)\,\rho_\eta(da\mid o)\,
\lambda_{\mathbf m,\chi}(do)
-
\int_{\mathcal O}\sup_{a\in D_a}W_{\mathbf m}(o,a)\,
\lambda_{\mathbf m,\chi}(do)
\right|
\le
2\|W_{\mathbf m}\|_\infty\,
\lambda_{\mathbf m,\chi}(\mathcal O\setminus K_\eta)
<2\|W_{\mathbf m}\|_\infty\,\eta,
\]
using that the two integrands agree on \(K_\eta\) and are each bounded
by \(\|W_{\mathbf m}\|_\infty\) off it. Letting \(\eta\downarrow0\) gives
\[
\sup_{\rho\in\mathcal R^c_{i,j}}
V^{\mathrm{dev}}_{i,j}(\mathbf m;\chi,\rho)
\ge
\int_{\mathcal O}\sup_{a\in D_a}W_{\mathbf m}(o,a)\,
\lambda_{\mathbf m,\chi}(do).
\]
Combined with the reverse inequality, which holds because
\(\mathcal R^c_{i,j}\subseteq\mathcal R_{i,j}\) and the displayed
integral already equals \(\sup_{\rho\in\mathcal R_{i,j}}
V^{\mathrm{dev}}_{i,j}(\mathbf m;\chi,\rho)\) by the exactness step
above, this gives
\[
\sup_{\rho\in\mathcal R^c_{i,j}}
V^{\mathrm{dev}}_{i,j}(\mathbf m;\chi,\rho)
=
\sup_{\rho\in\mathcal R_{i,j}}
V^{\mathrm{dev}}_{i,j}(\mathbf m;\chi,\rho).
\]
\end{proof}

\begin{lemma}[Common-conull-set interpretation of report IC]
\label{lem:ic-common-conull-bridge}
For a fixed mechanism profile \(\mathbf m\), define
\[
G_{i,j}^{\mathbf m}(t,r)
:=
\int_{D_p\times D_a}
\max_{a\in D_a} W_{\mathbf m}((t,r,d_{p,j},d'_{i,j}),a)\,
\alpha^{m_j}_{i,j}\!\big(d(d_{p,j},d'_{i,j})\mid r\big),
\qquad
E_{i,j}^{\mathbf m}(t)
:=\operatorname*{ess\,sup}_{r\in T}G_{i,j}^{\mathbf m}(t,r).
\]
For every \(\eta\ge0\), the interim inequalities
\(U_{i,j}(t;\mathbf m)+\eta\ge G_{i,j}^{\mathbf m}(t,r)\)
for all \(t,r\) in a common \(H_{i,j}\)-conull set hold if and only if
\(U_{i,j}(t;\mathbf m)+\eta\ge E_{i,j}^{\mathbf m}(t)\)
for \(H_{i,j}\)-almost every \(t\). Moreover,
\(E_{i,j}^{\mathbf m}(t)\ge U_{i,j}(t;\mathbf m)\) almost everywhere.
Consequently, exact IC for agent \((i,j)\) as defined in the model is
equivalent to \(R_{i,j}(\mathbf m)\le0\).
\end{lemma}

\begin{proof}
The measurable maximum theorem makes \(G_{i,j}^{\mathbf m}\) jointly
measurable and identifies it with the optimal report-\(r\) payoff when
the agent can choose its action after observing its recommendation.
Compactness of the primitive spaces and continuity of \(u_{i,j}\),
\(h\), and \(h_j\) give a modulus \(\omega\), independent of the report
\(r\), such that
\[
|G_{i,j}^{\mathbf m}(t,r)-G_{i,j}^{\mathbf m}(s,r)|
\le\omega(d_T(t,s))
\quad\text{for all }t,s,r\in T.
\]
Indeed the same bound holds uniformly in the observed decisions and
chosen action for \(W_{\mathbf m}\); taking the maximum and integrating
against the report-indexed kernel \(\alpha^{m_j}_{i,j}\) preserve it.
The essential supremum \(E_{i,j}^{\mathbf m}\) inherits this modulus.

Let \(S\) be the support of \(H_{i,j}\) and choose a countable dense
subset \(\{t_k\}\) of \(S\). For each \(k\), the set of reports
\(\{r:G_{i,j}^{\mathbf m}(t_k,r)>E_{i,j}^{\mathbf m}(t_k)\}\)
is null. Outside their union, the inequality
\(G_{i,j}^{\mathbf m}(t_k,r)\le E_{i,j}^{\mathbf m}(t_k)\) holds for
every \(k\), and the common modulus extends it to every \(t\in S\).
Intersect this conull report set with the conull set of true types
where \(U_{i,j}(t;\mathbf m)+\eta\ge E_{i,j}^{\mathbf m}(t)\).
The intersection is a single conull set on which all the stated
interim inequalities hold. The reverse implication follows by taking
the essential supremum over reports.

Taking the report \(r=t\) and then obeying the recommendation is an
available action choice. Thus
\(G_{i,j}^{\mathbf m}(t,t)\ge U_{i,j}(t;\mathbf m)\), and the preceding
common-report-set argument gives
\(E_{i,j}^{\mathbf m}(t)\ge U_{i,j}(t;\mathbf m)\) almost everywhere.
By Corollary~\ref{cor:ess-sup-report-kernel} and
Lemma~\ref{lem:value-density},
\(D_{i,j}(\mathbf m)=\int E_{i,j}^{\mathbf m}\,dH_{i,j}\).
Therefore \(R_{i,j}=\int(E_{i,j}^{\mathbf m}-U_{i,j})\,dH_{i,j}\)
is nonnegative and vanishes exactly when the interim IC inequalities
hold on a common conull set.
\end{proof}

\begin{remark}[Meaning of the tolerance]
\label{rem:ex-ante-ic-tolerance}
For \(\varepsilon>0\), the condition \(R_{i,j}\le\varepsilon\) bounds
the agent's expected regret across its types. It does not assert the
pointwise interim bound
\(U_{i,j}(t;\mathbf m)+\varepsilon\ge E_{i,j}^{\mathbf m}(t)\)
at almost every type.
\end{remark}

\begin{lemma}[Safe completion and unrestricted reporting deviations]
\label{lem:safe-completion}
Fix a mechanism profile \(\mathbf m\) and kernel versions of its
coordinates. For each agent \((i,j)\) there is an \(H_{i,j}\)-null Borel
set \(N_{i,j}\subset T\) such that
\(G_{i,j}^{\mathbf m}(t,r)\le E_{i,j}^{\mathbf m}(t)\) for every \(t\) in
the support of \(H_{i,j}\) and every \(r\notin N_{i,j}\), and such that
for every \(r\notin N_{i,j}\) the disintegration~\eqref{eq:krs-disintegration}
identifies \(G_{i,j}^{\mathbf m}(\cdot,r)\) with the report-\(r\)
deviation value of the given kernel. Choose
\(r^0_{i,j}\in T\setminus N_{i,j}\), let \(P_{i,j}(r)=r\) for
\(r\notin N_{i,j}\) and \(P_{i,j}(r)=r^0_{i,j}\) for \(r\in N_{i,j}\), and
define the completed kernels
\[
\widehat m_j(\,\cdot\mid\mathbf r_j)
:=
m_j\big(\,\cdot\mid P_{1,j}(r_{1,j}),\ldots,P_{n,j}(r_{n,j})\big),
\qquad j=1,\ldots,N .
\]
Then:
\begin{enumerate}[label=(\alph*)]
\item \(\widehat m_j\) and \(m_j\) agree \(H_j\)-almost everywhere, so they
  define the same element of \(M_j\); truthful outcome laws, truthful
  interim utilities, and principals' payoffs are unchanged.
\item For \(H_{i,j}\)-almost every true type \(t\), the supremum of agent
  \((i,j)\)'s deviation payoff under \(\widehat{\mathbf m}\) over
  \emph{all} reports \(r\in T\) and all measurable subsequent action rules
  equals \(E_{i,j}^{\mathbf m}(t)\). Hence exact incentive compatibility
  against all reports, and the ex ante regret computed with that
  all-report supremum, coincide with~\eqref{mainIC} and with
  \(R_{i,j}(\mathbf m)\), respectively.
\item If \(\mathbf m\) is an \(\varepsilon\)-relaxed Bayesian--Nash
  principals' equilibrium for some \(\varepsilon\ge0\), then
  \(\widehat{\mathbf m}\) is an equilibrium of the game in which agents'
  regret is computed with the supremum over all reports: each
  \(\widehat m_j\) satisfies that unrestricted \(\varepsilon\)-constraint,
  and no principal has a profitable deviation among kernels satisfying
  it against \(\widehat{\mathbf m}_{-j}\).
\end{enumerate}
\end{lemma}

\begin{proof}
The first property of \(N_{i,j}\) is established, before intersecting
with any set of true types, in the countable-dense-set step of the proof
of Lemma~\ref{lem:ic-common-conull-bridge}; for the second, the given
kernel and the disintegration~\eqref{eq:krs-disintegration} reconstruct
the same product-reference joint measure, so Fubini's theorem and
uniqueness of disintegration give equality of the report-indexed
measures off a report-null set; using a countable determining class makes
this a single Borel exceptional set.
Assumption~\ref{assump:team-type-density} makes \(H_j\) equivalent to
the product reference law \(\overline H_j\), so an \(H_j\)-equivalent
kernel may be used under that reference law. Conditional deviation
payoffs are evaluated using the maintained density representation of the
type law, and all identifications with conditional expectations are made
on a common conull set of true types, independent of the report.

(a) The kernels differ only on
\(\bigcup_{k}\{\mathbf r_j:r_{k,j}\in N_{k,j}\}\), which is
\(H_j\)-null because each \(N_{k,j}\) is null for the marginal
\(H_{k,j}\).

(b) For each teammate \(k\neq i\),
\(H_{k,j}(N_{k,j})=0\) implies
\(\Pr(t_{k,j}\in N_{k,j}\mid t_{i,j}=t)=0\) for \(H_{i,j}\)-almost every
\(t\); this uses only conditional probability, not independence, and
the same holds for null-set changes in other teams' mechanisms. Taking
the finite union gives a conull set \(T^*_{i,j}\) of true types, not
depending on the report, on which no other agent's truthful report is
redirected. For \(t\in T^*_{i,j}\) and any own report \(r\), the
completed profile therefore generates exactly the team decision and
recommendation distribution that the original profile generates after
report \(P_{i,j}(r)\). The agent's actual report supplies no further
payoff-relevant information, and every measurable action rule remains
available, so its maximized report-and-action payoff equals
\(G_{i,j}^{\mathbf m}(t,P_{i,j}(r))\). Since \(P_{i,j}(r)\notin N_{i,j}\),
this is at most \(E_{i,j}^{\mathbf m}(t)\) for \(t\) in the support.
Conversely, the reports outside \(N_{i,j}\) form a conull set on which
the completed and original values coincide, so the supremum is at least
the original essential supremum. Thus the all-report supremum equals
\(E_{i,j}^{\mathbf m}(t)\) for almost every \(t\), and since truthful
utility is unchanged almost everywhere, the all-report regret equals
\(R_{i,j}(\mathbf m)\).

(c) By (b), each \(\widehat m_j\) satisfies the unrestricted
\(\varepsilon\)-constraint. Let \(\tilde m_j\) be any kernel satisfying
it against \(\widehat{\mathbf m}_{-j}\). An essential supremum cannot
exceed the all-report supremum, and the essential-supremum evaluation
depends on \(\widehat{\mathbf m}_{-j}\) only through its fixed-marginal
class, which by (a) is \(\mathbf m_{-j}\). Hence the class of
\(\tilde m_j\) belongs to \(IC_j^\varepsilon(\mathbf m_{-j})\), and the
equilibrium property of \(\mathbf m\) gives
\(V_j(\tilde m_j,\widehat{\mathbf m}_{-j})=V_j(\tilde m_j,\mathbf m_{-j})
\le V_j(\mathbf m)=V_j(\widehat{\mathbf m})\). The completion depends on
the profile \(\mathbf m\), but the argument needs completed
representatives only at the equilibrium profile; deviations are
evaluated through their fixed-marginal classes and need not be
completed.
\end{proof}

\begin{remark}[Measurability of all-report deviation values]
\label{rem:all-report-measurable}
For an arbitrary Borel kernel representative, define
\(G^{\mathrm{ker}}_{i,j}(t,r)\) directly as the optimized report-\(r\)
payoff, including at exceptional reports, optimizing over measurable
policies that map the observed team decision and recommendation into an
action. Under the product reference law the kernel and these policies do
not depend on the true type, which enters only through the bounded
continuous utilities and density weights. Their uniform continuity gives
a modulus \(\omega\), independent of the report and the action policy,
with
\(|G^{\mathrm{ker}}_{i,j}(t,r)-G^{\mathrm{ker}}_{i,j}(s,r)|\le\omega(d_T(t,s))\)
for all \(r\) and all \(t,s\) in the support of \(H_{i,j}\). Taking the
supremum over reports preserves the modulus, so
\(D^{\mathrm{all}}_{i,j}(t):=\sup_{r\in T}G^{\mathrm{ker}}_{i,j}(t,r)\) is
continuous on that support, and the all-report regret
\(\int(D^{\mathrm{all}}_{i,j}-U_{i,j})\,dH_{i,j}\) is an ordinary integral.
The same modulus connects this pointwise value to measurable reporting
strategies: approximately optimal report-and-action policies at a
countable dense set of true types approximate \(D^{\mathrm{all}}_{i,j}\)
everywhere on the support, and choosing an approximately best member of
that countable family measurably gives an admissible reporting and
action strategy.
\end{remark}

\begin{lemma}[Narrow convergence with an integrable fixed-marginal weight]
\label{lem:product-young-integrable-weight}
Let \(S\) and \(Z\) be compact metric spaces and let
\(\xi^k\to\xi\) narrowly in \(\mathcal P(S\times Z)\). Suppose every
\(\xi^k\), as well as \(\xi\), has the same \(S\)-marginal
\(\lambda\). If \(g\in L^1(\lambda)\) and \(\varphi\) is bounded and
continuous on \(S\times Z\), then
\[
\int \varphi(s,z)g(s)\,\xi^k(ds,dz)
\longrightarrow
\int \varphi(s,z)g(s)\,\xi(ds,dz).
\]
\end{lemma}

\begin{proof}
The conclusion follows from narrow convergence when \(g\) is bounded
and continuous. For a general
\(g\in L^1(\lambda)\), choose bounded continuous functions \(g_r\) such
that \(\lVert g_r-g\rVert_{L^1(\lambda)}\to0\). Replacing \(g\) by
\(g_r\) changes either integral by at most
\[
\lVert\varphi\rVert_\infty
\lVert g_r-g\rVert_{L^1(\lambda)},
\]
uniformly in \(k\). First let \(k\to\infty\) for fixed \(r\), and then
let \(r\to\infty\). The error estimate is uniform precisely because the
\(S\)-marginal is \(\lambda\) for every \(k\).
\end{proof}

\begin{lemma}[Fixed-deviation regret continuity]
\label{lem:fixed-deviation-continuity}
For every \(\chi\in\mathcal K^{ac}_{i,j}\) and every
\(\rho\in\mathcal R^c_{i,j}\), the map
\[
\mathbf m\longmapsto
r_{i,j}(\mathbf m;\chi,\rho)
=
V^{\mathrm{dev}}_{i,j}(\mathbf m;\chi,\rho)
-
T_{i,j}(\mathbf m)
\]
is continuous on \(\mathcal M=\prod_{k=1}^N M_k\) for the product
Young--narrow topology.
\end{lemma}

\begin{proof}
Fix \(\chi\in\mathcal K^{ac}_{i,j}\) and \(\rho\in\mathcal R^c_{i,j}\). Two
distinct type variables appear in \eqref{eq:relaxed-deviation-value}: the
deviator's \emph{true type} \(t_i\), which enters the utility and type
densities, and the \emph{report-index type} \(\tau\)---the value the
mechanism conditions on---which indexes the product-reference disintegration
\eqref{eq:krs-disintegration}. We write the report variable as \(\tau\)
throughout.

Because \(\chi(\cdot\mid t_i)\ll H_{i,j}\) for \(H_{i,j}\)-almost every
\(t_i\), write
\[
\chi(d\tau\mid t_i)=g_\chi(t_i,\tau)\,H_{i,j}(d\tau)
\]
for a jointly measurable density \(g_\chi\). The report-index \(\tau\)
enters \eqref{eq:relaxed-deviation-value} only through the kernels
\(\alpha^{m_j}_{i,j}(\cdot\mid\tau)\) and
\(\beta^{m_j}_{i,j}(\cdot\mid\tau,\cdot)\), so integrating it against
\(H_{i,j}\) reconstitutes the team-\(j\) fixed-marginal law: by
\eqref{eq:krs-disintegration},
\[
H_{i,j}(d\tau)\,
\alpha^{m_j}_{i,j}\!\big(d(d_{p,j},d'_{i,j})\mid\tau\big)\,
\beta^{m_j}_{i,j}\!\big(d(\mathbf t_{-i,j},\mathbf d'_{-i,j})\mid\tau,d_{p,j},d'_{i,j}\big)
=\bar\mu_j(m_j)\big(d(\tau,\mathbf t_{-i,j}),\,dd'_j\big),
\]
with \(d'_j=(d_{p,j},d'_{i,j},\mathbf d'_{-i,j})\). Substituting the density
and using this identity, \eqref{eq:relaxed-deviation-value} becomes an
integral against the product
\(H_{i,j}\otimes\bar\mu_j(m_j)\otimes\mu_{-j}(\mathbf m_{-j})\),
\[
V^{\mathrm{dev}}_{i,j}(\mathbf m;\chi,\rho)
=\int \Phi_\rho\;g_\chi(t_i,\tau)\,h_j(t_i,\mathbf t_{-i,j})\,
h(\widehat{\mathbf t})\;
H_{i,j}(dt_i)\,
\bar\mu_j(m_j)\big(d(\tau,\mathbf t_{-i,j},d'_j)\big)\,
\mu_{-j}(\mathbf m_{-j})\big(d(\mathbf t_{-j},\mathbf d_{-j})\big),
\]
where \(\widehat{\mathbf t}=((t_i,\mathbf t_{-i,j}),\mathbf t_{-j})\) carries
the deviator's \emph{true} own type \(t_i\), and
\[
\Phi_\rho
=
\int_{D_a}
u_{i,j}\big(\widehat{\mathbf t};(d_{p,j},a,\mathbf d'_{-i,j}),\mathbf d_{-j}\big)\,
\rho\big(da\mid t_i,\tau,d_{p,j},d'_{i,j}\big).
\]
Thus the report-conditional expression has become an ordinary integral
against the product-reference fixed-marginal law \(\bar\mu_j(m_j)\).
The within-team density is evaluated at the true type \(t_i\), not at
the report-index \(\tau\).

For fixed \(\rho\in\mathcal R^c_{i,j}\), weak continuity of \(\rho\),
continuity of \(u_{i,j}\), and compactness of \(D_a\) imply that
\(\Phi_\rho h_jh\) is bounded and continuous in all type and decision
coordinates. The only merely measurable factor is the fixed report
density \(g_\chi\). Moreover,
\[
\sup_{\text{decision coordinates}}\big|\Phi_\rho\,g_\chi\,h_j\,h\big|
\le
\|u_{i,j}\|_\infty\,\|h_j\|_\infty\,\|h\|_\infty\,g_\chi(t_i,\tau),
\]
which is integrable against \(H_{i,j}(dt_i)\,H_{i,j}(d\tau)\), since
\(\chi\) is a stochastic kernel.

Because \(h_j\) is continuous and bounded away from zero,
\(\bar\mu_j(m_j)=h_j^{-1}\mu_j(m_j)\); hence Young-narrow convergence of
\(\mu_j(m_j^k)\) implies Young-narrow convergence of
\(\bar\mu_j(m_j^k)\). Therefore, if \(\mathbf m^k\to\mathbf m\) in the
product Young-narrow topology, then \(\bar\mu_j(m_j^k)\to\bar\mu_j(m_j)\) and
\(\mu_{-j}(\mathbf m_{-j}^k)\to\mu_{-j}(\mathbf m_{-j})\) narrowly while
\(H_{i,j}\) is fixed. Their product laws converge narrowly and have the
same joint marginal on the type coordinates
\((t_i,\tau,\mathbf t_{-i,j},\mathbf t_{-j})\). Applying
Lemma~\ref{lem:product-young-integrable-weight} with fixed-marginal
weight \(g_\chi\) gives
\[
V^{\mathrm{dev}}_{i,j}(\mathbf m^k;\chi,\rho)
\longrightarrow
V^{\mathrm{dev}}_{i,j}(\mathbf m;\chi,\rho).
\]
The truthful-obedient payoff \(T_{i,j}\) is a Carath\'eodory integral
against \(\mu(\mathbf m)=h\,\bigotimes_j\mu_j(m_j)\) and is likewise jointly
continuous. Hence \(r_{i,j}(\cdot;\chi,\rho)\) is continuous.
\end{proof}

\begin{proposition}[Closed graph of the \texorpdfstring{\(\varepsilon\)}{epsilon}-IC correspondence]
\label{prop:ic-closed-graph}
For every \(\varepsilon\ge0\), the correspondence
\[
\mathbf m_{-j}\longmapsto IC_j^\varepsilon(\mathbf m_{-j})
\]
has closed graph in \(M_{-j}\times M_j\) under the product Young-narrow
topology.
\end{proposition}

\begin{proof}
By Lemma~\ref{lem:value-density}, for every mechanism profile
\(\mathbf m=(m_j,\mathbf m_{-j})\), the economic deviation envelope
\eqref{eq:economic-envelope} satisfies
\[
D_{i,j}(\mathbf m)
=
\sup_{\chi\in\mathcal K^{ac}_{i,j}}
\sup_{\rho\in\mathcal R^c_{i,j}}
V^{\mathrm{dev}}_{i,j}(\mathbf m;\chi,\rho).
\]
Consequently,
\[
\operatorname{Gr}(IC_j^\varepsilon)
=
\bigcap_{i=1}^n
\bigcap_{\chi\in\mathcal K^{ac}_{i,j}}
\bigcap_{\rho\in\mathcal R^c_{i,j}}
\left\{
(\mathbf m_{-j},m_j)\in M_{-j}\times M_j:
r_{i,j}(m_j,\mathbf m_{-j};\chi,\rho)\le\varepsilon
\right\}.
\]
Each set in the intersection is closed, since
\(r_{i,j}(\cdot;\chi,\rho)\) is continuous by
Lemma~\ref{lem:fixed-deviation-continuity}. Therefore
\[
\operatorname{Gr}(IC_j^\varepsilon)
\]
is closed. The argument uses only continuity of the individual
fixed-deviation regret functionals; it does not require continuity of
the envelope \(D_{i,j}\), which is generally only lower semicontinuous.
\end{proof}

\begin{remark}[Where each hypothesis is used]
\label{rem:hypothesis-roles}
Two distinct regularizations enter, at two different coordinates. The
report kernel's absolute continuity
\[
\chi(\cdot\mid t)\ll H_{i,j}
\]
controls the report coordinate, turning the essential supremum over
reports into an integral against the fixed-marginal law. The joint weak
continuity of \(\rho\in\mathcal R^c_{i,j}\) on \(\mathcal O\) keeps the
integrand \(\Phi_\rho\) continuous in the type, report, and
observed-decision coordinates. Neither is a
restriction on the agent's economic deviations: the report side is the
version-independent essential supremum, and the decision side is
value-dense by Lemma~\ref{lem:value-density}.
\end{remark}

\begin{lemma}[Opponent-direction upper semicontinuity of regret]
\label{lem:opponent-regret-usc}
For every principal $j$, every agent $i$, and every fixed own-team
mechanism $m_j\in M_j$, the regret map
\[
\mathbf m_{-j}\longmapsto R_{i,j}(m_j,\mathbf m_{-j})
=D_{i,j}(m_j,\mathbf m_{-j})-T_{i,j}(m_j,\mathbf m_{-j})
\]
is upper semicontinuous on $M_{-j}=\prod_{k\neq j}M_k$ under the product
Young-narrow topology.
\end{lemma}

\begin{proof}
Fix $i$, $j$, and $m_j$. Let
$\ell^{m_j}_{\chi,\rho}\in\mathcal P(T^n\times D)$ be the law of the
team-$j$ type profile and realized decision vector
$(\mathbf t_j,d^{\mathrm{dev}}_j)$ generated by the deviation
$(\chi,\rho)$, namely the marginal on $(\mathbf t_j,d^{\mathrm{dev}}_j)$ of
\[
H_{i,j}(dt_i)\,\chi(dt'\mid t_i)\,
\alpha^{m_j}_{i,j}\!\big(d(d_{p,j},d'_{i,j})\mid t'\big)\,
\rho\big(da\mid t_i,t',d_{p,j},d'_{i,j}\big)\,
\beta^{m_j}_{i,j}\!\big(d(\mathbf t_{-i,j},\mathbf d'_{-i,j})\mid t',d_{p,j},d'_{i,j}\big)\,
h_j(t_i,\mathbf t_{-i,j}),
\]
where $\mathbf t_j=(t_i,\mathbf t_{-i,j})$ and
$d^{\mathrm{dev}}_j=(d_{p,j},a,\mathbf d'_{-i,j})$. This is a probability
measure on $T^n\times D$ that does not depend on the opponents' variables:
under the product-reference disintegration, the marginal of
$\mathbf t_{-i,j}$ before reweighting is
$\bigotimes_{k\ne i}H_{k,j}$, and
$\int h_j(t_i,\mathbf t_{-i,j})\,
\bigotimes_{k\ne i}H_{k,j}(d\mathbf t_{-i,j})=1$ for
$H_{i,j}$-almost every $t_i$.
The report $t'$ and the recommendation $d'_{i,j}$ have been integrated
out, so $\ell^{m_j}_{\chi,\rho}$ lives on $T^n\times D$, not on
$T^n\times T\times D$.
Because opponents enter agent $(i,j)$'s deviation problem only through
their truthful-obedient law $\mu_{-j}(\mathbf m_{-j})$, the fixed-deviation
value can be written as an integral over the opponents' type--decision
space $Z_{-j}:=T^{n(N-1)}\times D^{N-1}$,
\[
V^{\mathrm{dev}}_{i,j}(m_j,\mathbf m_{-j};\chi,\rho)
=\int_{Z_{-j}}\Psi^{m_j}_{\chi,\rho}(\mathbf t_{-j},\mathbf d_{-j})\,
\mu_{-j}(\mathbf m_{-j})\big(d(\mathbf t_{-j},\mathbf d_{-j})\big),
\]
where
\[
\Psi^{m_j}_{\chi,\rho}(\mathbf t_{-j},\mathbf d_{-j})
=\int u_{i,j}\big((\mathbf t_j,\mathbf t_{-j});d^{\mathrm{dev}}_j,\mathbf d_{-j}\big)\,
h(\mathbf t_j,\mathbf t_{-j})\,\ell^{m_j}_{\chi,\rho}(d\mathbf t_j,dd^{\mathrm{dev}}_j).
\]
This is a re-expansion of the deviation value in \eqref{eq:relaxed-deviation-value} by Fubini's theorem, integrating the team-$j$ variables against $\ell^{m_j}_{\chi,\rho}$ and the opponents' variables against $\mu_{-j}(\mathbf m_{-j})$; the density $h$ is the single cross-team weight already carried by $\bar u_{i,j}$ in \eqref{eq:reduced-utility}, not an additional weighting.

\emph{Step 1: the family $\mathcal F=\{\Psi^{m_j}_{\chi,\rho}\}$ is uniformly
bounded and equicontinuous.}
Since $|u_{i,j}|\le\|u_{i,j}\|_\infty$, every member of $\mathcal F$ is
bounded by $\|u_{i,j}\|_\infty\,\|h\|_\infty$. The product
$\widehat T\times D$ is compact metric, so $u_{i,j}$ and $h$ are uniformly
continuous; hence
$(\mathbf t_j,\mathbf t_{-j},d^{\mathrm{dev}}_j,\mathbf d_{-j})\mapsto u_{i,j}h$
admits a modulus of continuity $\omega$ in $(\mathbf t_{-j},\mathbf d_{-j})$
that is uniform over $(\mathbf t_j,d^{\mathrm{dev}}_j)$. Averaging against the
probability measure $\ell^{m_j}_{\chi,\rho}$ preserves this modulus, so each
$\Psi^{m_j}_{\chi,\rho}$ has modulus $\le\omega$ in
$(\mathbf t_{-j},\mathbf d_{-j})$, uniformly over $(\chi,\rho)$. Thus
$\mathcal F$ is uniformly bounded and equicontinuous on the compact metric
space $Z_{-j}$.

\emph{Step 2: envelope representation and convergence of the measures.}
By the value-density lemma~\ref{lem:value-density},
\[
D_{i,j}(m_j,\mathbf m_{-j})
=\sup_{\chi\in\mathcal K^{ac}_{i,j}}\ \sup_{\rho\in\mathcal R^c_{i,j}}
V^{\mathrm{dev}}_{i,j}(m_j,\mathbf m_{-j};\chi,\rho)
=\sup_{\Psi\in\mathcal F}\int_{Z_{-j}}\Psi\,d\mu_{-j}(\mathbf m_{-j}).
\]
If $\mathbf m_{-j}^k\to\mathbf m_{-j}$ in the product Young-narrow topology,
then each opponent truthful-obedient law $\mu_k(m_k^k)$ converges narrowly
and, finite products preserving narrow convergence,
$\mu_{-j}(\mathbf m_{-j}^k)\to\mu_{-j}(\mathbf m_{-j})$ narrowly on the
compact space $Z_{-j}$. The density $h$ enters only through the integrand
$\Psi^{m_j}_{\chi,\rho}$ and the equicontinuity of $\mathcal F$, not through
the convergence of the measures.

\emph{Step 3: upper semicontinuity.}
Set $L:=\limsup_k D_{i,j}(m_j,\mathbf m_{-j}^k)$, which is finite since
$D_{i,j}$ is bounded by $\|u_{i,j}\|_\infty\|h\|_\infty$. By definition of
$\limsup$, there is a subsequence $(k_m)_{m\in\mathbb N}$ of the full index
sequence along which
\[
D_{i,j}(m_j,\mathbf m_{-j}^{k_m})\longrightarrow L.
\]
For each $m$, choose $\Psi_{k_m}\in\mathcal F$ with
\[
\int\Psi_{k_m}\,d\mu_{-j}(\mathbf m_{-j}^{k_m})\ge D_{i,j}(m_j,\mathbf m_{-j}^{k_m})-1/m
\];
by construction, \[
\int\Psi_{k_m}\,d\mu_{-j}(\mathbf m_{-j}^{k_m})\to L
\] as
well. Only now do we invoke Arzel\`a--Ascoli: $\mathcal F$ is relatively
compact in $C(Z_{-j})$, so along a further subsequence $(k_{m_\ell})_\ell$
of $(k_m)_m$, $\Psi_{k_{m_\ell}}\to\Psi_\ast$ uniformly. Then
\[
\Big|\int\Psi_{k_{m_\ell}}\,d\mu_{-j}(\mathbf m_{-j}^{k_{m_\ell}})
-\int\Psi_\ast\,d\mu_{-j}(\mathbf m_{-j})\Big|
\le\|\Psi_{k_{m_\ell}}-\Psi_\ast\|_\infty
+\Big|\int\Psi_\ast\,d\mu_{-j}(\mathbf m_{-j}^{k_{m_\ell}})
-\int\Psi_\ast\,d\mu_{-j}(\mathbf m_{-j})\Big|\to0,
\]
using uniform convergence of the integrands and narrow convergence of the
measures along the subsequence $(k_{m_\ell})_\ell$ (inherited from
$\mathbf m_{-j}^k\to\mathbf m_{-j}$). Since $(k_{m_\ell})_\ell$ is itself a
subsequence of $(k_m)_m$, along which
\[
\int\Psi_{k_m}\,d\mu_{-j}(\mathbf m_{-j}^{k_m})\to L
\], the same sequence
of integrals has two limits along $(k_{m_\ell})_\ell$, namely $L$ and
$\int\Psi_\ast\,d\mu_{-j}(\mathbf m_{-j})$, so
\[
\int\Psi_\ast\,d\mu_{-j}(\mathbf m_{-j})=L.
\]
On the other hand, each $\Psi_{k_{m_\ell}}\in\mathcal F$, so
$\int\Psi_{k_{m_\ell}}\,d\mu_{-j}(\mathbf m_{-j})\le\sup_{\Psi\in\mathcal F}
\int\Psi\,d\mu_{-j}(\mathbf m_{-j})=D_{i,j}(m_j,\mathbf m_{-j})$ for every
$\ell$, and integration against the fixed limit measure
$\mu_{-j}(\mathbf m_{-j})$ is sup-norm continuous, so letting
$\ell\to\infty$,
\[
\int\Psi_\ast\,d\mu_{-j}(\mathbf m_{-j})
=\lim_\ell\int\Psi_{k_{m_\ell}}\,d\mu_{-j}(\mathbf m_{-j})
\le D_{i,j}(m_j,\mathbf m_{-j}).
\]
Combining the last two displays, $L\le D_{i,j}(m_j,\mathbf m_{-j})$, i.e.
$\limsup_k D_{i,j}(m_j,\mathbf m_{-j}^k)\le D_{i,j}(m_j,\mathbf m_{-j})$,
so $D_{i,j}(m_j,\cdot)$ is upper semicontinuous. The truthful value
$T_{i,j}(m_j,\cdot)$ is an integral of a bounded continuous integrand
against $\mu_{-j}(\mathbf m_{-j})$, hence continuous; therefore
$R_{i,j}(m_j,\cdot)=D_{i,j}(m_j,\cdot)-T_{i,j}(m_j,\cdot)$ is upper
semicontinuous.
\end{proof}

\begin{lemma}[Convex-valuedness of $\varepsilon$-IC]
\label{lem:epsilon-ic-convex-valued}
For every $\varepsilon\ge0$, every principal $j$, and every
$\mathbf m_{-j}$, the set $IC_j^\varepsilon(\mathbf m_{-j})$ is convex.
\end{lemma}

\begin{proof}
Fix $\mathbf m_{-j}$ and an agent $i$. The space $M_j$ is
convex and $m_j\mapsto\mu_j(m_j)$ is affine, so
$T_{i,j}(\cdot,\mathbf m_{-j})$ is affine in $m_j$. For each fixed
$\chi\in\mathcal K^{ac}_{i,j}$ and $\rho\in\mathcal R^c_{i,j}$, the value
$V^{\mathrm{dev}}_{i,j}(\cdot,\mathbf m_{-j};\chi,\rho)$ integrates a fixed
bounded functional against the deviation law; as the own-type marginal
$H_{i,j}$ is held fixed, that law is affine in $m_j$, so
$V^{\mathrm{dev}}_{i,j}(\cdot,\mathbf m_{-j};\chi,\rho)$ is affine in $m_j$.
By the value-density lemma~\ref{lem:value-density}, for every $m_j$,
\[
D_{i,j}(m_j,\mathbf m_{-j})
=\sup_{\chi\in\mathcal K^{ac}_{i,j}}\ \sup_{\rho\in\mathcal R^c_{i,j}}
V^{\mathrm{dev}}_{i,j}(m_j,\mathbf m_{-j};\chi,\rho),
\]
a pointwise supremum of a single family of functionals each affine in
$m_j$; hence $D_{i,j}(\cdot,\mathbf m_{-j})$ is convex. Therefore
$R_{i,j}(\cdot,\mathbf m_{-j})=D_{i,j}-T_{i,j}$ is convex, and
\[
IC_j^\varepsilon(\mathbf m_{-j})
=\bigcap_{i=1}^{n}\big\{m_j\in M_j:\ R_{i,j}(m_j,\mathbf m_{-j})\le\varepsilon\big\}
\]
is an intersection of convex sublevel sets, hence convex.
\end{proof}

\begin{lemma}[Lower hemicontinuity of $\varepsilon$-IC]
\label{lem:epsilon-ic-lhc}
Suppose the seed mechanism $m_j^0$ satisfies
$R_{i,j}(m_j^0,\mathbf m_{-j})\le0$ for every agent $i$ and every
$\mathbf m_{-j}$. Then for every $\varepsilon>0$ and every principal $j$ the
correspondence $\mathbf m_{-j}\mapsto IC_j^\varepsilon(\mathbf m_{-j})$ is
lower hemicontinuous.
\end{lemma}

\begin{proof}
Fix $\mathbf m_{-j}$, let $m_j\in IC_j^\varepsilon(\mathbf m_{-j})$, and let
$\mathbf m_{-j}^k\to\mathbf m_{-j}$. For $\lambda\in(0,1]$ set
$m_j^\lambda=(1-\lambda)m_j+\lambda m_j^0\in M_j$. By convexity of each
$R_{i,j}(\cdot,\mathbf m_{-j})$ (Lemma~\ref{lem:epsilon-ic-convex-valued}),
using $R_{i,j}(m_j,\mathbf m_{-j})\le\varepsilon$ and
$R_{i,j}(m_j^0,\mathbf m_{-j})\le0$,
\[
R_{i,j}(m_j^\lambda,\mathbf m_{-j})
\le(1-\lambda)R_{i,j}(m_j,\mathbf m_{-j})+\lambda R_{i,j}(m_j^0,\mathbf m_{-j})
\le(1-\lambda)\varepsilon<\varepsilon,\qquad i=1,\dots,n.
\]
Fix $\lambda$. By Lemma~\ref{lem:opponent-regret-usc},
$R_{i,j}(m_j^\lambda,\cdot)$ is upper semicontinuous at $\mathbf m_{-j}$, so
\[
\limsup_{k}R_{i,j}(m_j^\lambda,\mathbf m_{-j}^k)
\le R_{i,j}(m_j^\lambda,\mathbf m_{-j})\le(1-\lambda)\varepsilon<\varepsilon .
\]
Hence there is a threshold $K(\lambda)$ with
$m_j^\lambda\in IC_j^\varepsilon(\mathbf m_{-j}^k)$ for all $k\ge K(\lambda)$.
Now take any sequence $\lambda_\ell\downarrow0$ and choose strictly
increasing thresholds $\bar K_\ell$ with $\bar K_\ell\ge K(\lambda_\ell)$.
Define the selection
\[
m_j^k:=
\begin{cases}
m_j^0, & k<\bar K_1,\\[2pt]
m_j^{\lambda_\ell}, & \bar K_\ell\le k<\bar K_{\ell+1}.
\end{cases}
\]
For $\bar K_\ell\le k<\bar K_{\ell+1}$ we have $k\ge\bar K_\ell\ge K(\lambda_\ell)$,
so $m_j^k=m_j^{\lambda_\ell}\in IC_j^\varepsilon(\mathbf m_{-j}^k)$; for
$k<\bar K_1$, $m_j^k=m_j^0\in IC_j^\varepsilon(\mathbf m_{-j}^k)$ by
Lemma~\ref{lem:seed-IC}. Since $\lambda_\ell\downarrow0$,
$\mu_j(m_j^{\lambda_\ell})=(1-\lambda_\ell)\mu_j(m_j)+\lambda_\ell\mu_j(m_j^0)
\to\mu_j(m_j)$, so $m_j^k\to m_j$. This selection establishes lower
hemicontinuity.
\end{proof}

\subsection*{Principal Payoffs and
\texorpdfstring{\(\varepsilon\)}{epsilon}-Best Responses}

The preceding sections establish the topological and convex structure of
the primitive mechanism strategy space. The equilibrium argument uses
the fact that each principal's strategy space \(M_j\) is convex. For \(m_j^1,m_j^2\in M_j\) and \(\lambda\in[0,1]\), the
own-coordinate mixture is the pointwise convex combination of kernels
\[
m_j^\lambda(\cdot\mid\mathbf t_j')
:=
\lambda m_j^1(\cdot\mid\mathbf t_j')
+
(1-\lambda)m_j^2(\cdot\mid\mathbf t_j').
\]
Equivalently, under the Young-narrow fixed-marginal representation,
\[
H_j(d\mathbf t_j)m_j^\lambda(dd'_j\mid\mathbf t_j)
=
\lambda H_j(d\mathbf t_j)m_j^1(dd'_j\mid\mathbf t_j)
+
(1-\lambda)H_j(d\mathbf t_j)m_j^2(dd'_j\mid\mathbf t_j).
\]
Thus \(m_j^\lambda\in M_j\), and convex combinations are taken directly
within the primitive mechanism space.

All subsequent convexity, compactness, continuity, and fixed-point
arguments use the compact convex structure of the primitive mechanism
spaces, the fixed-deviation closedness argument above, the convexity and
lower-hemicontinuity properties of the \(\varepsilon\)-IC correspondences,
and continuity and quasiconcavity of principal payoffs.

Although the exact incentive-compatible correspondence
\[
\mathbf m_{-j}\longmapsto IC_j(\mathbf m_{-j})
\]
need not be lower hemicontinuous, the \(\varepsilon\)-relaxed
correspondence has the closedness and convexity properties needed for
the generalized-game fixed-point argument. For each \(\varepsilon>0\),
principal \(j\)'s feasible set is the \(\varepsilon\)-relaxed
incentive-compatible set \(IC_j^\varepsilon(\mathbf m_{-j})\) of
\eqref{eq:ic-eps-def}.
The relaxation enlarges the exact IC set but leaves the principals'
best-response requirement exact: each principal optimizes over the
mechanisms whose incentive constraints fail, if at all, by at most
\(\varepsilon\).

We next establish the payoff properties used in the fixed-point
argument. For a mechanism profile
\[
\mathbf m=(m_j,\mathbf m_{-j})\in\mathcal M,
\]
the expected payoff functional of principal \(j\) is
\[
V_j(\mathbf m)
:=
E_{\mathbf m}[\pi_j]
=
\int_X \pi_j\,d\mu(\mathbf m).
\]

\begin{lemma}[Continuity and quasiconcavity of principal payoffs]
\label{prop:payoff}
For each principal \(j\), the functional
\[
V_j:\mathcal M\to\mathbb R
\]
is continuous. Moreover, for each fixed opponents' profile
\(\mathbf m_{-j}\), the map
\[
m_j\longmapsto V_j(m_j,\mathbf m_{-j})
\]
is affine, and hence quasiconcave, on \(M_j\).
\end{lemma}

\begin{proof}
Let \(\mathbf m^k\to\mathbf m\) in \(\mathcal M\). The aggregate
truthful-obedient law factors as
\[
\mu(\mathbf m)=h(\widehat{\mathbf t})\,\bigotimes_{j=1}^N\mu_j(m_j),
\]
with \(h\) bounded and continuous. Young-narrow convergence of the kernels
gives \(\mu_j(m_j^k)\to\mu_j(m_j)\) narrowly for each \(j\); finite products
preserve narrow convergence, and reweighting by the bounded continuous
density \(h\) preserves it as well, so
\[
\mu(\mathbf m^k)\Rightarrow \mu(\mathbf m)
\]
in \(\mathcal P(X)\). Since \(\pi_j\) is bounded and continuous,
\[
\int_X \pi_j\,d\mu(\mathbf m^k)
\longrightarrow
\int_X \pi_j\,d\mu(\mathbf m).
\]
Thus \(V_j\) is continuous.

Now fix \(\mathbf m_{-j}\). For \(m_j^1,m_j^2\in M_j\) and
\(\lambda\in[0,1]\), let
\[
m_j^\lambda
:=
\lambda m_j^1+(1-\lambda)m_j^2.
\]
Because the truthful-obedient law is affine in the own mechanism
coordinate,
\[
\mu(m_j^\lambda,\mathbf m_{-j})
=
\lambda \mu(m_j^1,\mathbf m_{-j})
+
(1-\lambda)\mu(m_j^2,\mathbf m_{-j}).
\]
Therefore
\[
V_j(m_j^\lambda,\mathbf m_{-j})
=
\lambda V_j(m_j^1,\mathbf m_{-j})
+
(1-\lambda)V_j(m_j^2,\mathbf m_{-j}).
\]
Hence \(V_j(\cdot,\mathbf m_{-j})\) is affine, and therefore
quasiconcave.
\end{proof}

\subsection*{Existence of \texorpdfstring{\(\varepsilon\)}{epsilon}-Equilibria}

We now establish existence of an \(\varepsilon\)-relaxed
Bayesian--Nash principals' equilibrium. The proof applies the
Kakutani--Fan--Glicksberg fixed-point theorem to the principals'
\(\varepsilon\)-best-response correspondence on the primitive mechanism
strategy space
\[
\mathcal M=\prod_{j=1}^N M_j.
\]

Recall the \(\varepsilon\)-relaxed incentive-compatible correspondence
\(IC_j^\varepsilon(\mathbf m_{-j})\) of \eqref{eq:ic-eps-def}.
By Lemma~\ref{lem:seed-IC}, Proposition~\ref{prop:ic-closed-graph}, and
Lemmas~\ref{lem:epsilon-ic-convex-valued} and
\ref{lem:epsilon-ic-lhc}, \(IC_j^\varepsilon\) has nonempty compact convex
values and is continuous.

\begin{definition}[\(\varepsilon\)-relaxed best responses]
\label{def:eps-BR}
For each principal \(j\), each \(\varepsilon>0\), and each other-principals'
mechanism profile
\[
\mathbf m_{-j}\in M_{-j},
\]
define the \emph{\(\varepsilon\)-relaxed best-response correspondence} by
\[
BR_j^\varepsilon(\mathbf m_{-j})
:=
\operatorname*{arg\,max}_{m_j\in
IC_j^\varepsilon(\mathbf m_{-j})}
V_j(m_j,\mathbf m_{-j})
\ \subseteq\ M_j.
\]
Thus principal \(j\) optimizes exactly over the primitive mechanisms
satisfying team \(j\)'s incentive constraints up to the tolerance
\(\varepsilon\), holding the other principals' mechanisms fixed.
\end{definition}

\begin{definition}[Joint \(\varepsilon\)-relaxed best-response correspondence]
\label{def:F-epsilon}
Since the primitive mechanism strategy space is the product
\[
\mathcal M=\prod_{j=1}^N M_j,
\]
the joint \(\varepsilon\)-relaxed best-response correspondence is
\[
F^\varepsilon:
\mathcal M\rightrightarrows\mathcal M,
\qquad
F^\varepsilon(\mathbf m)
:=
\prod_{j=1}^N
BR_j^\varepsilon(\mathbf m_{-j}).
\]
\end{definition}

A profile \(\mathbf m^\varepsilon\) is a fixed point of \(F^\varepsilon\) if
and only if
\[
m_j^\varepsilon
\in
BR_j^\varepsilon(\mathbf m_{-j}^\varepsilon)
\]
for every principal \(j\). Equivalently, every principal optimizes
exactly over its \(\varepsilon\)-relaxed incentive-compatible mechanism set.

\begin{theorem*}[Kakutani--Fan--Glicksberg fixed-point theorem
\citealp{fan1952fixed,glicksberg1952further}]
\label{thm:KFG}
Let \(K\) be a nonempty compact convex subset of a locally convex
Hausdorff topological vector space. If
\[
F:K\rightrightarrows K
\]
is upper hemicontinuous with nonempty compact convex values, then \(F\)
has a fixed point.
\end{theorem*}

We apply the theorem with
\[
K=\mathcal M=\prod_{j=1}^N M_j.
\]
By Lemma~\ref{lem:mechanism-space-compact-convex}, each \(M_j\) is a
nonempty compact convex subset of a locally convex Hausdorff topological
vector space. Hence the finite product \(\mathcal M\) is also a
nonempty compact convex subset of a locally convex Hausdorff topological
vector space.

The next lemmas verify that \(F^\varepsilon\) satisfies the hypotheses of
the Kakutani--Fan--Glicksberg theorem. Since finite products of upper
hemicontinuous correspondences with nonempty compact convex values
inherit these properties from their factors, it suffices to establish
them for the per-principal correspondences \(BR_j^\varepsilon\).

\begin{lemma}[Properties of \(\varepsilon\)-relaxed best responses]
\label{lem:epsilon-br-properties}
For each principal \(j\), the correspondence
\[
BR_j^\varepsilon:M_{-j}\rightrightarrows M_j
\]
has nonempty compact convex values and is upper hemicontinuous.
\end{lemma}

\begin{proof}
This proof uses the topological properties of the \(\varepsilon\)-IC
correspondence established above. By Lemma~\ref{lem:seed-IC},
\(IC_j^\varepsilon(\mathbf m_{-j})\) is nonempty for every
\(\mathbf m_{-j}\); indeed Lemma~\ref{lem:seed-IC} shows the seed
mechanism satisfies \(R_{i,j}(m_j^0,\mathbf m_{-j})\le0\) for every agent
\(i\) and every \(\mathbf m_{-j}\), which is precisely the hypothesis of
Lemma~\ref{lem:epsilon-ic-lhc}. By Proposition~\ref{prop:ic-closed-graph}
its graph is closed, so, since \(M_j\) is compact, its values are compact
and the correspondence is upper hemicontinuous. By
Lemma~\ref{lem:epsilon-ic-convex-valued} its values are convex, and by
Lemma~\ref{lem:epsilon-ic-lhc}, whose hypothesis was just verified, it is
lower hemicontinuous. Hence
\(IC_j^\varepsilon\) is a continuous correspondence with nonempty, compact,
convex values.

By Lemma~\ref{prop:payoff}, \(V_j\) is continuous on \(\mathcal M\) and
quasiconcave in \(m_j\). Berge's maximum theorem
\citep[Theorem~17.31]{aliprantis2006infinite} applied to
\(IC_j^\varepsilon\) then implies that \(BR_j^\varepsilon\) is upper
hemicontinuous with nonempty compact values. Convexity of the values
follows from convexity of \(IC_j^\varepsilon(\mathbf m_{-j})\) together with
quasiconcavity of \(V_j(\cdot,\mathbf m_{-j})\).
\end{proof}

\begin{lemma}[Existence of \(\varepsilon\)-relaxed BNPE]
\label{lem:epsilon-BNPE}
For every \(\varepsilon>0\), there exists a mechanism profile
\[
\mathbf m^\varepsilon\in\mathcal M
\]
such that
\[
\mathbf m^\varepsilon
\in
F^\varepsilon\!\left(\mathbf m^\varepsilon\right).
\]
Equivalently, \(\mathbf m^\varepsilon\) is an
\(\varepsilon\)-relaxed Bayesian--Nash principals' equilibrium.
\end{lemma}

\begin{proof}
Fix \(\varepsilon>0\). By Lemma~\ref{lem:mechanism-space-compact-convex},
each \(M_j\) is a nonempty compact convex subset of a locally convex
Hausdorff topological vector space. Hence the finite product
\[
\mathcal M=\prod_{j=1}^N M_j
\]
is also a nonempty compact convex subset of a locally convex Hausdorff
topological vector space.

By Lemma~\ref{lem:epsilon-br-properties}, each per-principal
correspondence
\[
BR_j^\varepsilon:M_{-j}\rightrightarrows M_j
\]
is upper hemicontinuous with nonempty compact convex values. Therefore
the joint correspondence
\[
F^\varepsilon(\mathbf m)
=
\prod_{j=1}^N BR_j^\varepsilon(\mathbf m_{-j})
\]
is upper hemicontinuous with nonempty compact convex values.

Thus
\[
F^\varepsilon:\mathcal M\rightrightarrows\mathcal M
\]
satisfies the hypotheses of the Kakutani--Fan--Glicksberg fixed-point
theorem. Hence there exists
\[
\mathbf m^\varepsilon\in\mathcal M
\]
such that
\[
\mathbf m^\varepsilon\in F^\varepsilon(\mathbf m^\varepsilon).
\]
Equivalently,
\[
m_j^\varepsilon\in BR_j^\varepsilon(\mathbf m_{-j}^\varepsilon)
\]
for every principal \(j\). Therefore \(\mathbf m^\varepsilon\) is an
\(\varepsilon\)-relaxed Bayesian--Nash principals' equilibrium.
\end{proof}

The proof of Theorem~\ref{thm:approximate-equilibrium} follows
immediately from Lemma~\ref{lem:epsilon-BNPE}.

\section{Proof of the Exact Contest Theorem}
\label{app:exact-contest}

This section shows that the primitive generalized-game/Berge--Kakutani architecture
underlying Appendix~\ref{app:epsilon} extends, under additional
structure and a participation constraint, to deliver \emph{exact}
equilibrium with no relaxation. We
first embed the contest as a specialization of the mechanism game in
the model, and then show that every stochastic mechanism deviation is
weakly dominated by a deterministic conditional-mean mechanism of the
form used in the reduced proof. The
structure is a primitive condition on the contest technology ---
Lipschitz continuity and concavity in own output, plus a separate
condition guaranteeing a profitable pooling seed --- together these
replace incentive slack (unavailable at $\varepsilon=0$) with budget
slack.

\subsection*{Environment}

There are $N\ge2$ teams $j=1,\ldots,N$, each with a single agent
($n=1$; see the Discussion). Agent $j$ has private type $t_j\in T=[0,1]$,
interpreted as an inverse marginal cost of effort, drawn independently
across teams, with each $H_j$ absolutely continuous with respect to
Lebesgue measure on $T$ and having full support on $T$ (in particular
nonatomic; absolute continuity is used for dominated-convergence
arguments below, matching $H_j$-a.e.\ and Lebesgue-a.e.\ statements,
while full support makes every conull report set dense). There is no obedience margin, so the
agent's only unilateral deviation is a misreport.

\paragraph{Cost.} Agent $j$'s cost of producing output $y$ at true type
$t$ is $c(t)\,y$, where $c:T\to\mathbb R_{++}$ is continuously
differentiable and \emph{uniformly strictly decreasing},
$c'(t)\le-\underline c<0$ for all $t$. Write
$\bar g:=\|c'\|_\infty\bar y+\|c\|_\infty\bar y$, a bound used in
Step~0 below.

\paragraph{Contest technology.} The $N$ teams divide a fixed prize
$V>0$ according to a continuous sharing rule
$\sigma_j:Y^N\to[0,1]$, $Y=[0,\bar y]$, $\sum_j\sigma_j\equiv1$, giving
team $j$'s deterministic winnings $w_j(\mathbf y):=V\,\sigma_j(\mathbf y)$.

\begin{assumption}[Lipschitz, concave contest technology]
\label{a:technology}
For every fixed $\mathbf y_{-j}\in Y^{N-1}$, $y_j\mapsto\sigma_j(y_j,\mathbf y_{-j})$
is nondecreasing, concave, and $L$-Lipschitz, with $L$ independent of
$\mathbf y_{-j}$. In addition, for every fixed $y_j$ and every
$j'\ne j$, $y_{j'}\mapsto\sigma_j(y_j,\mathbf y_{-j})$ is nonincreasing;
this is specific to symmetric-type contest shares (Tullock satisfies
it directly, since raising a rival's output only raises the
denominator) and is not implied by the other properties.
\end{assumption}

\begin{assumption}[Profitable pooling seed]
\label{a:participate-contest}
There exists $y^0\in(0,\bar y]$ such that
\[
V\,\sigma_j\bigl(y^0,\bar y,\ldots,\bar y\bigr)\ >\ c(0)\,y^0.
\]
\end{assumption}

This is stated directly on the primitive $\sigma_j$, not derived from
Assumption~\ref{a:technology}: Lipschitz continuity, concavity, and
monotonicity alone do not bound how little of the prize team $j$
secures near $y_j=0$, so a separate condition is needed.

\begin{remark}[Tullock verifies both assumptions]
\label{r:tullock}
The smoothed Tullock share $\sigma_j(\mathbf y)=(y_j+\eta)^\rho/
\sum_{j'}(y_{j'}+\eta)^\rho$, $\eta>0$, $\rho\in(0,1]$, satisfies
Assumption~\ref{a:technology}: it is $C^1$ on the compact box $Y^N$
(hence Lipschitz); nondecreasing and concave in $y_j$, since
$u\mapsto u/(u+\text{const})$ is increasing and concave and
$y_j\mapsto y_j^\rho$ is concave for $\rho\le1$, and a concave
nondecreasing function composed with a concave function is concave;
and nonincreasing in each rival's output directly, since raising
$y_{j'}$ raises only the denominator.

For Assumption~\ref{a:participate-contest}: with $\rho=1$,
\(
\sigma_j(y^0,\bar y,\ldots,\bar y)
=
(y^0+\eta)/[(y^0+\eta)+(N-1)(\bar y+\eta)]
\ge
y^0/[y^0+(N-1)(\bar y+\eta)]
\)
for every $\eta\ge0$. If $V>c(0)(N-1)\bar y$, then for $\eta$ small
enough and $y^0$ small enough,
$V\,y^0/[y^0+(N-1)(\bar y+\eta)]>c(0)y^0$ --- dividing by $y^0$, this
is $V>c(0)[y^0+(N-1)(\bar y+\eta)]$, which holds since the right side
converges to $c(0)(N-1)\bar y<V$ as $y^0,\eta\downarrow0$. So
Assumption~\ref{a:participate-contest} holds for Tullock ($\rho=1$)
whenever $V>c(0)(N-1)\bar y$; the constants $\eta,y^0$ used here are
specific to verifying the assumption for this particular technology,
not part of the assumption itself.
\end{remark}

\paragraph{Preferences.} Agent $j$'s utility is quasi-linear in a
realized reward $p_j\in[0,w_j]$: $u_j=p_j-c(t_j)\,y_j$. Principal $j$
keeps the residual of the team's winnings after paying the reward:
$\pi_j=w_j(\mathbf y)-p_j$. The realized reward is paid according to a
\emph{schedule} $p_j=R_j(t,w)$ constructed in Step~2 to satisfy
$R_j(t,w)\in[0,w]$ for every $w$ and every $t$.

\paragraph{Participation.} Each agent can decline to participate and
obtain utility zero. A mechanism is \emph{individually rational} if,
under truthful reporting, the agent's interim expected utility is
nonnegative for $H_j$-almost every type. Without this constraint the
exact problem is degenerate: the report-independent mechanism that
assigns output $\bar y$ and share $0$ after every report is exactly
incentive compatible and, because shares are nondecreasing in own
output and rewards are nonnegative, weakly dominant for every principal
against any opponents' mechanisms. Individual rationality rules out such
unpaid compulsory production and gives the reward-delivery constraint
below a substantive role.

\paragraph{Embedding in the general mechanism game.}
Specialize the model's team decision space to
\[
D_p=Y\times[0,1],
\qquad
D_a=\{*\}.
\]
A realized team-level decision $d_{p,j}=(y_j,\gamma_j)$ specifies
output and the fraction of team winnings paid to the agent. Given the
full decision profile, set
\[
p_j=\gamma_jV\sigma_j(\mathbf y),
\qquad
u_j(t_j,\mathbf d)=p_j-c(t_j)y_j,
\qquad
\pi_j(\mathbf d)=(1-\gamma_j)V\sigma_j(\mathbf y).
\]
Thus an unrestricted primitive mechanism in this specialized game is a
kernel
\[
K_j(dy_j,d\gamma_j\mid t'_j)\in
\mathcal K(T,Y\times[0,1]).
\]
The individual-action space is a singleton, so the general model's
obedience constraint is vacuous and unilateral deviations consist only
of report deviations. Principal~$j$ chooses among kernels that are
incentive compatible and individually rational, so the embedded game is
the general mechanism game with an added participation constraint. Mechanisms are identified through their
fixed-$H_j$ marginal laws, and IC is understood on the common conull
domain used in Appendix~\ref{app:epsilon}. The reduced strategy space
introduced next is not an exogenous restriction on this kernel game:
Proposition~\ref{p:lifting} proves that it contains a weakly better
response to every stochastic kernel deviation.

\subsection*{Step 0: the strategy space}

A reduced mechanism for principal $j$ is a pair $m_j=(q_j,v_0)$: an
allocation $q_j$ and a base rent $v_0$, the constant in the envelope
formula. The associated interim reward mean is defined in Step~1.
Individual rationality requires $v_0\ge0$; principal optimality then
permits the normalization $v_0=0$ without loss of best responses
(Remark~\ref{r:v0-range}).

\begin{definition}[Canonical allocations]
$Q_j$ is the set of $q_j:T\to Y$ that are nondecreasing,
right-continuous on $[0,1)$, and satisfy $q_j(1)=\lim_{s\uparrow1}q_j(s)$.
Equivalently, writing $q_j(t)=\nu_j([0,t])$ for the associated finite
Borel measure $\nu_j$ on $T$ of mass at most $\bar y$: $\nu_j$ places
no mass exactly at $t=1$. This costs nothing economically ($t=1$ is
reported with probability zero, since $H_j\ll$ Lebesgue), and buys a
genuine (not merely pseudo-) metric, as the next lemma shows.
\end{definition}

\begin{lemma}[Compactness]\label{l:compact-contest}
Metrize $Q_j$ by $d_1(q,q')=\int_0^1|q-q'|\,dt$ (Lebesgue measure), and
let $S_j:=Q_j\times\{0\}$, $m_j=(q_j,v_0)\in S_j$ metrized by
$d_1(q,q')+|v_0-v_0'|$. Then $d_1$ is a genuine metric on $Q_j$, and
$(S_j,d_1)$ is compact, convex, and metrizable.
\end{lemma}

\begin{proof}
\emph{$d_1$ separates $Q_j$.} Suppose $q,q'\in Q_j$ have $d_1(q,q')=0$,
i.e.\ $q=q'$ off a Lebesgue-null set $Z$. For any $t_0\in[0,1)$, since
$Z$ is null it cannot contain a whole interval $(t_0,t_0+\delta)$, so
there is a sequence $s_n\downarrow t_0$ with $q(s_n)=q'(s_n)$; right-continuity
of both at $t_0$ gives $q(t_0)=\lim_nq(s_n)=\lim_nq'(s_n)=q'(t_0)$. At
$t_0=1$, the same argument with $s_n\uparrow1$ and the defining
left-continuity convention at $1$ gives $q(1)=q'(1)$. So $q=q'$
everywhere.

\emph{Convexity.} A pointwise convex combination of functions that are
nondecreasing, right-continuous on $[0,1)$, and satisfy the convention
at $1$ retains all three properties.

\emph{Compactness.} Every $q\in Q_j$ has total variation at most
$\bar y$; by Helly's selection theorem every sequence has a
subsequence converging pointwise, at every continuity point of some
nondecreasing limit $\tilde q$, to $\tilde q$. Since $d_1$ cannot
detect the value of any function at the single point $t=1$ (Lebesgue-null),
replacing $\tilde q(1)$ by $\lim_{s\uparrow1}\tilde q(s)$ changes no
$d_1$-distance; call the result $q^*\in Q_j$. The discontinuity set of
$q^*$ within $[0,1)$ is countable hence Lebesgue-null. The subsequence
therefore converges to $q^*$ Lebesgue-a.e.\ on $[0,1]$ (the endpoint is
irrelevant for this statement) and, by dominated convergence, in
$d_1$. These properties are inherited by $S_j=Q_j\times\{0\}$.
Because the canonical convention makes the map into $L^1([0,1])$
injective, $S_j$ is thereby identified affinely with a compact convex
subset of the locally convex Banach space
$L^1([0,1])\times\mathbb R$.
\end{proof}

\begin{remark}[Why the normalization \(v_0=0\) is without loss]
\label{r:v0-range}
Write $g(t;q_j):=\int_0^t[-c'(s)]q_j(s)\,ds+c(t)q_j(t)$. Both terms
are nonnegative, so $g\in[0,\bar g]$ for every $t$ and $q_j$, and the
envelope gives $\bar r_j(t)=v_0+g(t;q_j)$ and interim utility
$V(t)=v_0+\int_0^t[-c'(s)]q_j(s)\,ds$, which is nondecreasing in $t$
with $V(0)=v_0$. Since $V$ is continuous and $H_j$ has full support,
nonnegativity for $H_j$-almost every type extends to every type,
including $0$; individual rationality is therefore equivalent to
$v_0\ge0$. For fixed $q_j$, lowering $v_0$ lowers every interim reward
mean by the same amount, which preserves incentive compatibility,
relaxes the upper budget constraint, keeps rewards nonnegative because
$g\ge0$, and strictly raises the principal's payoff. Hence every
individually rational mechanism with $v_0>0$ is strictly dominated by
the same allocation with $v_0=0$, and restricting to $v_0=0$ loses no
best responses.
\end{remark}

The fixed-point argument is carried out on
$\mathcal M=\prod_jS_j$. The underlying mechanism game nevertheless
allows every stochastic kernel in
$\mathcal K(T,Y\times[0,1])$; Proposition~\ref{p:lifting} is the
reduction that makes the smaller space sufficient for best responses.
No public randomization over reduced mechanisms is needed.

\subsection*{Step 1: single crossing and incentive compatibility}

\begin{lemma}[Single crossing]\label{l:sc-general}
For every $t''>t'$ and $y''>y'$,
\[
\bigl[u_j(t'',y'')-u_j(t'',y')\bigr]-\bigl[u_j(t',y'')-u_j(t',y')\bigr]
=-\bigl[c(t'')-c(t')\bigr](y''-y')\ \ge\ \underline c\,(t''-t')(y''-y')>0,
\]
using only that $c$ is strictly decreasing with slope bounded away
from zero; this is independent of the reward and hence of $\sigma_j$,
$V$, and the opponents' mechanisms.
\end{lemma}

For a measurable output rule $q_j:T\to Y$ and measurable interim
reward mean $\bar r_j:T\to\mathbb R$, write
\[
U(t',t;q_j,\bar r_j):=\bar r_j(t')-c(t)q_j(t'),
\qquad
V(t;q_j,\bar r_j):=U(t,t;q_j,\bar r_j).
\]

\begin{lemma}[Envelope characterization]\label{l:ic-contest}
Truthful reporting under $(q_j,\bar r_j)$ is optimal on a common
$H_j$-conull set of true types and reports if and only if $q_j$ agrees
$H_j$-almost everywhere with some $\tilde q_j\in Q_j$ and, for some
constant $v_0$,
\[
V(t)=v_0+\int_0^t\bigl[-c'(s)\bigr]\tilde q_j(s)\,ds,
\qquad
\bar r_j(t)=V(t)+c(t)\tilde q_j(t)
\]
for $H_j$-almost every $t$. Because an $H_j$-null set may have positive
Lebesgue measure, the integral uses the canonical monotone version
$\tilde q_j$, which full support determines uniquely. Replacing $q_j$
by $\tilde q_j$ and defining the displayed functions at every type gives an
equivalent fixed-marginal mechanism for which truthful reporting is
optimal pointwise. Accordingly, for $m_j=(q_j,v_0)\in S_j$ we henceforth
define $V(t;m_j)$ and $\bar r_j(t;m_j)$ by these formulas.
\end{lemma}

\begin{proof}
First suppose the IC inequalities hold pointwise. For $t''>t'$, adding
the inequalities for true types $t'$ and $t''$ gives
$[c(t')-c(t'')][q_j(t'')-q_j(t')]\ge0$, so $q_j$ is nondecreasing
since $c$ is strictly decreasing. Since $c\in C^1(T)$,
$U(t',\cdot;q_j,\bar r_j)$ is differentiable in $t$ with
$|\partial_tU(t',t)|\le\|c'\|_\infty\bar y$, uniformly in $t'$.
The truthful value is therefore Lipschitz. Since its supremum is attained
at $t'=t$ by hypothesis, the envelope theorem for optimization with a
parameter \citep[Theorem~2]{milgromsegal2002}, whose integrable-bound
hypothesis is exactly the uniform Lipschitz bound just established,
gives $V'(t)=-c'(t)q_j(t)$ for a.e.\ $t$, yielding the displayed
envelope formula.

Conversely, given nondecreasing $q_j$, define $V$ and $\bar r_j$ by
the displayed formulas. Using
$U(t',t)=\bar r_j(t')-c(t)q_j(t')=V(t')+q_j(t')[c(t')-c(t)]$,
for all $t,t'$,
\[
V(t)-U(t',t)=\bigl[V(t)-V(t')\bigr]+q_j(t')\bigl[c(t)-c(t')\bigr]
=\int_{t'}^t\bigl[-c'(s)\bigr]\bigl[q_j(s)-q_j(t')\bigr]\,ds\ \ge0,
\]
the integrand matching the sign of $ds$ since $q_j$ is nondecreasing
and $-c'\ge\underline c>0$; hence truthful reporting is optimal.

Under the common-conull formulation, apply the same two-IC-inequality
argument on the conull domain. Full support of $H_j$ makes that domain
dense in $T$. The increasing restriction of $q_j$ therefore has the
canonical right-continuous nondecreasing extension in $Q_j$ (with its
left limit assigned at $t=1$), and the uniformly Lipschitz truthful
value has a unique continuous extension. Writing $\tilde V$ for that
continuous extension and setting
$\tilde r(t):=\tilde V(t)+c(t)\tilde q_j(t)$, taking one-sided limits of
the IC inequalities along the conull domain (from the right below $1$,
and from the left at $1$) gives pointwise IC for $(\tilde q_j,\tilde r)$.
The pointwise envelope argument therefore yields
$\tilde V'(t)=-c'(t)\tilde q_j(t)$ for Lebesgue-almost every $t$. Since
$\tilde V$ is Lipschitz, integration gives
$\tilde V(t)=\tilde V(0)+\int_0^t[-c'(s)]\tilde q_j(s)\,ds$ for every
$t$. The extended value and reward agree with the original ones
$H_j$-almost everywhere. Changing the output and reward rules off the common
conull report set, together with the countably many jump points and
endpoints at which canonicalization may change the monotone restriction,
does not change their fixed-marginal laws, since $H_j$ is nonatomic. The
pointwise sufficiency argument then applies to the canonical extension.
\end{proof}

\begin{remark}[Relation to \citet{sinander2022}]
This is the classical quasi-linear special case of \citet{sinander2022}'s
Implementability Theorem (\S4.3--4.4), under the sign convention his
setup requires ($P=-\bar r_j$, since his agent's utility is strictly
decreasing in the payment). The proof above is self-contained and does
not rest on the citation.
\end{remark}

\subsection*{Step 2: from the interim mean to a realized schedule}

The envelope pins the interim \emph{mean} reward $\bar r_j(t;m_j)$;
what is actually paid, ex post, is $R_j(t,w)$ for the \emph{realized}
winnings $w$, which varies with the opponents' realized types even
after agent $j$'s report is fixed. Write
\[
\bar w_j(t;q_j,m_{-j}):=\mathbb E_{t_{-j}\sim H_{-j}}
\bigl[w_j\bigl(q_j(t),\mathbf y_{-j}(t_{-j})\bigr)\bigr]
\]
for the interim expected winnings when agent $j$ reports $t$ (this
depends on $m_j$ only through $q_j$, not through $v_0$). Independence
of types across teams is used here: conditional on agent $j$'s true
type and report, the rivals' type law remains $H_{-j}$. This same fact
ensures below that a report-indexed proportional share delivers the
same interim reward mean for every true type. Define
the \emph{budget slack}
\[
\mathrm B_j(m_j;m_{-j}):=\min\Bigl\{\inf_{t\in T}\bar r_j(t;m_j),\ \
\inf_{t\in T}\bigl[\bar w_j(t;q_j,m_{-j})-\bar r_j(t;m_j)\bigr]\Bigr\}.
\]

\begin{lemma}[A delivering schedule exists iff $\mathrm B_j\ge0$]
\label{l:delivering-schedule}
Fix $m_j\in S_j$, $m_{-j}$, and $t\in T$ with
$\bar r_j(t;m_j)\in[0,\bar w_j(t;q_j,m_{-j})]$. If
$\bar w_j(t;q_j,m_{-j})>0$, the schedule
\[
R_j(t,w):=\frac{\bar r_j(t;m_j)}{\bar w_j(t;q_j,m_{-j})}\,w
\]
satisfies $R_j(t,w)\in[0,w]$ for every realized $w\ge0$, and
$\mathbb E_{t_{-j}}\bigl[R_j\bigl(t,w_j(q_j(t),\mathbf y_{-j}(t_{-j}))\bigr)\bigr]
=\bar r_j(t;m_j)$. If $\bar w_j(t;q_j,m_{-j})=0$, then
$\bar r_j(t;m_j)=0$ is forced, and $R_j(t,w):=0$ delivers it
(feasible, since $w=0$ a.s.\ when $\bar w_j(t)=0$). Conversely, any
schedule $R_j(t,\cdot)\in[0,w]$ delivering the mean forces
$\bar r_j(t;m_j)\in[0,\bar w_j(t;q_j,m_{-j})]$. In either case,
$R_j(t,\cdot)$ is Borel measurable: $\bar r_j(\cdot;m_j)$ and
$\bar w_j(\cdot;q_j,m_{-j})$ are Borel, and the piecewise ratio,
defined as $0$ where the denominator vanishes, is Borel.
\end{lemma}

\begin{proof}
For $\bar w_j(t)>0$: $\gamma(t):=\bar r_j(t)/\bar w_j(t)\in[0,1]$ from
the hypothesis, so $R_j(t,w)=\gamma(t)w\in[0,w]$ for every $w\ge0$,
and $\mathbb E_{t_{-j}}[R_j(t,w_j(\cdot))]=\gamma(t)\bar w_j(t)=\bar r_j(t)$.
If $\bar w_j(t)=0$, nonnegativity of realized winnings implies
$w_j(q_j(t),\mathbf y_{-j}(t_{-j}))=0$ for $H_{-j}$-almost every
$t_{-j}$. The maintained mean inequality then forces $\bar r_j(t)=0$,
and the zero schedule is feasible and delivers this mean.

Conversely, if a feasible schedule delivers $\bar r_j(t)$, then
$0\le R_j(t,w)\le w$ pointwise. Taking conditional expectations over
$t_{-j}$ gives $0\le\bar r_j(t)\le\bar w_j(t)$. This proves necessity
as well as sufficiency. Finally, $\bar r_j(\cdot;m_j)$ is Borel by the
envelope formula, while $\bar w_j(\cdot;q_j,m_{-j})$ is Borel as the
integral of a bounded jointly measurable function. Hence the piecewise
ratio $\gamma(t)=\bar r_j(t)/\bar w_j(t)$ on $\{\bar w_j>0\}$, extended
by zero on $\{\bar w_j=0\}$, is Borel, and so is
$(t,w)\mapsto R_j(t,w)=\gamma(t)w$.
\end{proof}

By Lemma~\ref{l:delivering-schedule},
$\mathcal F_j(m_{-j})=\{m_j\in S_j:\mathrm B_j(m_j;m_{-j})\ge0\}$: exactly
when $\mathrm B_j\ge0$, the win-contingent proportional share is a
genuine delivering schedule at \emph{every} report, so $q_j$
nondecreasing and $\mathrm B_j\ge0$ together are necessary and
sufficient for $m_j$ to be implementable by some feasible schedule.

\begin{proposition}[Embedding and conditional-mean reduction]
\label{p:lifting}
(a) Every $m_j=(q_j,v_0)\in\mathcal F_j(m_{-j})$ lifts to the deterministic
kernel
\[
K_j^{m_j}(dy,d\gamma\mid t)
=\delta_{(q_j(t),\gamma_j(t))}(dy,d\gamma),
\qquad
\gamma_j(t):=
\begin{cases}
\bar r_j(t;m_j)/\bar w_j(t;q_j,m_{-j}),&\bar w_j(t;q_j,m_{-j})>0,\\
0,&\bar w_j(t;q_j,m_{-j})=0,
\end{cases}
\]
which is feasible, exactly incentive compatible, and individually
rational in the embedded mechanism game.

(b) Fix deterministic reduced opponents $m_{-j}$. For every feasible
stochastic kernel $K_j\in\mathcal K(T,Y\times[0,1])$ that is incentive
compatible on the common conull domain and individually rational, its
conditional-mean output and interim reward satisfy the reduced envelope
representation with a constant $v_0\ge0$. Lowering that constant to
zero yields a reduced mechanism in $\mathcal F_j(m_{-j})$ that is
incentive compatible and individually rational and gives principal $j$
at least the expected payoff generated by $K_j$. Consequently,
maximizing over $\mathcal F_j(m_{-j})$ has the same value as maximizing
over all feasible, incentive-compatible, individually rational
stochastic kernels in the embedded mechanism game.
\end{proposition}

\begin{proof}
Part (a) follows from Lemmas~\ref{l:delivering-schedule}
and~\ref{l:ic-contest}. The ratio lies in $[0,1]$, including the
zero-winnings case, so the displayed kernel takes values in the model's
decision space $Y\times[0,1]$. Individual rationality holds because
$V(t;m_j)=\int_0^t[-c'(s)]q_j(s)\,ds\ge0$. Independence across teams makes the
conditional law of rivals, and hence the delivered mean, independent of
the agent's true type once its report is fixed.

For part (b), let $(Y_j,\Gamma_j)$ denote the output/share pair drawn
from $K_j(\cdot\mid r)$ after report $r$. Define, for almost every
report,
\[
q_j(r):=\mathbb E[Y_j\mid r],
\qquad
a_j(r):=
\mathbb E\!\left[
\Gamma_jV\sigma_j(Y_j,\mathbf q_{-j}(t_{-j}))\mid r
\right],
\]
where the expectation is over $K_j(\cdot\mid r)$ and the independent
rival types. The payoff of true type $t$ from report $r$ under $K_j$ is
exactly
\[
a_j(r)-c(t)q_j(r).
\]
Thus replacing $K_j$ by any mechanism delivering output $q_j(r)$ and
mean reward $a_j(r)$ preserves every report payoff and every IC
inequality.

Feasibility of the original share and concavity of the contest
technology give, report by report,
\[
0\le a_j(r)
\le \mathbb E\!\left[V\sigma_j(Y_j,\mathbf q_{-j})\mid r\right]
\le \mathbb E_{t_{-j}}\!\left[
V\sigma_j(q_j(r),\mathbf q_{-j}(t_{-j}))
\right]
=\bar w_j(r;q_j,m_{-j}),
\]
where the second inequality is conditional Jensen. Hence the
proportional share
\(\gamma_j(r)=a_j(r)/\bar w_j(r;q_j,m_{-j})\), with value zero when
the denominator vanishes, is pointwise feasible and delivers $a_j(r)$.

Because $K_j$ is IC, Lemma~\ref{l:ic-contest} implies that $q_j$ has a
nondecreasing canonical version and that $a_j$ satisfies the envelope
formula for some $v_0$. Null-report modifications needed to choose that
version do not change the fixed-marginal mechanism. The original mean
inequalities hold on a conull, hence dense, set of reports. For the
canonical version, both $\bar r_j$ and $\bar w_j$ are right-continuous on
$[0,1)$ and left-continuous at $1$; taking one-sided limits from that
dense set therefore extends
$0\le\bar r_j(t)\le\bar w_j(t;q_j,m_{-j})$ to every $t\in T$.
Let $V$ denote the continuous envelope extension associated with the
canonical allocation; it agrees with the truthful interim utility under
$K_j$ for $H_j$-almost every type. Individual rationality of $K_j$ therefore
gives $V(t)\ge0$ for $H_j$-almost every $t$; since $V$ is continuous and full support makes this set
dense, $v_0=V(0)\ge0$. Replace $v_0$ by zero. This subtracts the same
nonnegative constant from every report's reward mean, so it preserves
all IC comparisons; it keeps rewards nonnegative and individual
rationality satisfied, because the new mean is $g\ge0$ and the new
interim utility is $\int_0^t[-c'(s)]q_j(s)\,ds\ge0$; and it relaxes the
upper budget constraint. The reduction is applied to interim reward
means, not realized payments; the new means are delivered by the
proportional schedule of Lemma~\ref{l:delivering-schedule}, so realized
rewards remain in $[0,w]$. By the mean inequalities above, the resulting
pair belongs to $\mathcal F_j(m_{-j})$, and the change weakly raises the
principal's payoff.

Finally, conditional Jensen also shows that the reduced principal's
payoff weakly exceeds the original one:
\[
\bar w_j(r;q_j,m_{-j})-a_j(r)
\ge
\mathbb E\!\left[
V\sigma_j(Y_j,\mathbf q_{-j})-
\Gamma_jV\sigma_j(Y_j,\mathbf q_{-j})\mid r
\right].
\]
Integrating over reports proves the best-response equivalence.
\end{proof}

\subsection*{Step 3: the budget-slack functional}

\begin{lemma}[Budget slack]\label{l:budget-contest}
$\mathrm B_j$ is bounded, jointly upper semicontinuous, continuous in
$m_{-j}$ at fixed $m_j\in S_j$, and concave in $m_j$. Under
Assumption~\ref{a:participate-contest}, the flat pooling mechanism
$m_j^0=(q_j^0,v_0^0)$, $q_j^0\equiv y^0$ (the $y^0$ of
Assumption~\ref{a:participate-contest}), $v_0^0=0$, satisfies
$\mathrm B_j(m_j^0;m_{-j})\ge s_0>0$ for every $m_{-j}$, where
\[
s_0:=\min\bigl\{c(0)y^0,\ V\sigma_j(y^0,\bar y,\ldots,\bar y)-c(0)y^0\bigr\}.
\]
\end{lemma}

\begin{proof}
\emph{Concavity.} Fix $m_{-j}$ and $m_j^\lambda:=(1-\lambda)m_j^1+
\lambda m_j^2\in S_j$, i.e.\ $q_j^\lambda=(1-\lambda)q_j^1+\lambda q_j^2$
and $v_0^\lambda=(1-\lambda)v_0^1+\lambda v_0^2$. For every $t$ and
every realization of $t_{-j}$, concavity of
$\sigma_j(\cdot,\mathbf y_{-j}(t_{-j}))$ (Assumption~\ref{a:technology})
gives $\sigma_j(q_j^\lambda(t),\cdot)\ge(1-\lambda)\sigma_j(q_j^1(t),\cdot)
+\lambda\sigma_j(q_j^2(t),\cdot)$ directly from concavity applied to
the pointwise-averaged argument; taking $\mathbb E_{t_{-j}}$ preserves
the inequality, so $\bar w_j(t;\cdot,m_{-j})$ is concave in $q_j$,
hence in $m_j$, for every $t$. The envelope reward
$\bar r_j(t;m_j)=V(t;m_j)+c(t)q_j(t)$ is \emph{affine} in $m_j$. Hence
$\mathrm B_j$, a minimum of infima of a concave and an affine
functional (both concave, pointwise-in-$t$ concavity being preserved
by taking infima over $t$), is concave in $m_j$.

\emph{Boundedness.} Since $v_0=0$, $\bar r_j(t;m_j)=g(t;q_j)\in[0,\bar g]$
for every $t,m_j$; and
$0\le\bar w_j\le V$. Both brackets of $\mathrm B_j$ are bounded,
uniformly over $S_j\times\mathcal M_{-j}$.

\emph{Continuity of interim winnings.} Write
$\bar w_j(t;q_j,m_{-j})=\bar w_j^*(q_j(t);m_{-j})$, where
\[
\bar w_j^*(y;m_{-j})
:=
\mathbb E_{t_{-j}}[w_j(y,\mathbf y_{-j}(t_{-j}))].
\]
We first verify that $\bar w_j^*$ is jointly continuous on
$Y\times\mathcal M_{-j}$. Suppose $y^k\to y$ and
$m_{-j}^k\to m_{-j}$. For every rival $j'$, convergence
$q_{j'}^k\to q_{j'}$ in $L^1$ of Lebesgue measure implies convergence
in Lebesgue measure. Since $H_{j'}\ll\mathrm{Leb}$, it also implies
convergence in $H_{j'}$-measure. Independence across teams then gives
convergence in $H_{-j}$-measure of the rival-output vector
$\mathbf y_{-j}^k(t_{-j})$ to $\mathbf y_{-j}(t_{-j})$. Uniform
continuity of $\sigma_j$ on the compact set $Y^N$ yields
\[
w_j(y^k,\mathbf y_{-j}^k(\cdot))
\longrightarrow
w_j(y,\mathbf y_{-j}(\cdot))
\quad\text{in }H_{-j}\text{-measure}.
\]
These random variables are uniformly bounded by $V$, so they converge
in $L^1(H_{-j})$. Their expectations therefore converge, proving joint
continuity of $\bar w_j^*$. Moreover,
$\bar w_j^*(\cdot;m_{-j})$ is $VL$-Lipschitz uniformly in $m_{-j}$ by
Assumption~\ref{a:technology}. Both brackets defining $\mathrm B_j$,
as functions of $t$, are consequently right-continuous on $[0,1)$ and
left-continuous at $t=1$.

\emph{Semicontinuity, own coordinate.}

Let $m_j^k=(q_j^k,v_0^k)\to m_j=(q_j,v_0)$ in $S_j$, and possibly
$m_{-j}^k\to m_{-j}$ simultaneously. At every continuity point $t<1$
of $q_j$ in the interior --- a dense set $C\subset(0,1)$ whose
complement in $(0,1)$ is at most countable --- monotonicity and $L^1$ convergence imply
$q_j^k(t)\to q_j(t)$. No convergence at the endpoint $t=1$ is claimed:
$L^1$ convergence does not control that value. At every $t\in C$,
both brackets converge pointwise to their limit values at $t$. Writing
$h^k,h$ for either bracket's pre-limit and limit functions,
$\inf_sh^k(s)\le h^k(t)$ for every $t$ (the definition of infimum,
with no exceptional set), so $\limsup_k\inf_sh^k(s)\le h(t)$ for every
$t\in C$; taking the infimum over $t\in C$,
$\limsup_k\inf_sh^k(s)\le\inf_{t\in C}h(t)$. The limit bracket $h$ is
right-continuous on $[0,1)$ and left-continuous at $1$: the integral
part of $\bar r_j$ is continuous, while $q_j$ is right-continuous on
$[0,1)$ and satisfies $q_j(1)=\lim_{s\uparrow1}q_j(s)$; composition
with the continuous function $\bar w_j^*$ preserves the corresponding
one-sided continuity. Therefore every $t_0<1$, including $t_0=0$, can
be approximated from the right by points of $C$, and $t_0=1$ can be approximated from the
left by points of $C$. It follows that
$\inf_{t\in C}h(t)=\inf_{t\in T}h(t)$. Hence
$\limsup_k\inf_sh^k(s)\le\inf_{t\in T}h(t)$
for each bracket, and $\min$ of upper semicontinuous functions is
upper semicontinuous, giving joint upper semicontinuity of
$\mathrm B_j$.

\emph{Continuity in the opponents.} Fix $m_j\in S_j$ and let
$m_{-j}^k\to m_{-j}$ alone. By joint continuity of $\bar w_j^*$ on
the compact $Y\times\mathcal M_{-j}$, convergence is uniform in $y$:
\[
\sup_{y\in Y}|\bar w_j^*(y;m_{-j}^k)-\bar w_j^*(y;m_{-j})|\to0,
\qquad
\sup_{t\in T}|\bar w_j(t;q_j,m_{-j}^k)-\bar w_j(t;q_j,m_{-j})|\to0.
\]
since $\bar r_j(\cdot;m_j)$ does not move, this gives
$\mathrm B_j(m_j;m_{-j}^k)\to\mathrm B_j(m_j;m_{-j})$, as
Lemma~\ref{l:ic-cont-contest}'s interior-push step requires.

\emph{Seed.} At $m_j^0$, the envelope gives $\bar r_j^0(t)\equiv
y^0c(0)$ for every $t$. Rival output is bounded, $y_{j'}\le\bar y$ for
every $j'\ne j$ regardless of mechanism, so by the rival-monotonicity
clause of Assumption~\ref{a:technology}, $\sigma_j(y^0,\mathbf y_{-j})
\ge\sigma_j(y^0,\bar y,\ldots,\bar y)$ for every $m_{-j}$. The lower
margin of $\mathrm B_j$ equals $c(0)y^0$; the upper margin is at least
$V\sigma_j(y^0,\bar y,\ldots,\bar y)-c(0)y^0$, positive by
Assumption~\ref{a:participate-contest}. So
$\mathrm B_j(m_j^0;m_{-j})\ge\min\{c(0)y^0,\
V\sigma_j(y^0,\bar y,\ldots,\bar y)-c(0)y^0\}=s_0>0$, uniformly in
$m_{-j}$.
\end{proof}

\subsection*{Step 4: the reduced feasible correspondence is continuous}

\begin{lemma}\label{l:ic-cont-contest}
$m_{-j}\mapsto\mathcal F_j(m_{-j})$ has nonempty, compact, convex values and
is continuous.
\end{lemma}

\begin{proof}
$m_j^0\in\mathcal F_j(m_{-j})$ for every $m_{-j}$; $\mathcal F_j(m_{-j})=S_j\cap
\{\mathrm B_j(\cdot;m_{-j})\ge0\}$ is nonempty compact convex by
Lemma~\ref{l:budget-contest}; upper hemicontinuity from closedness of
the graph. It remains to prove lower hemicontinuity.

Fix $m_{-j}$, let $m_j\in\mathcal F_j(m_{-j})$, and consider any sequence
$m_{-j}^k\to m_{-j}$. For $\rho\in(0,1]$, define the interior push
\[
m_j^\rho:=(1-\rho)m_j+\rho m_j^0\in S_j.
\]
Concavity of $\mathrm B_j$ in the own mechanism gives
\[
\mathrm B_j(m_j^\rho;m_{-j})
\ge
(1-\rho)\mathrm B_j(m_j;m_{-j})
+\rho\mathrm B_j(m_j^0;m_{-j})
\ge \rho s_0>0.
\]
By continuity of $\mathrm B_j(m_j^\rho;\cdot)$ in the opponents'
profile, for every fixed $\rho\in(0,1]$ there is an integer $K(\rho)$
such that
\[
\mathrm B_j(m_j^\rho;m_{-j}^k)\ge\tfrac12\rho s_0>0
\quad\text{for every }k\ge K(\rho).
\]
Hence $m_j^\rho\in\mathcal F_j(m_{-j}^k)$ whenever $k\ge K(\rho)$.

Set $\rho_\ell:=1/\ell$ and define the nondecreasing running thresholds
\[
\widehat K_\ell
:=
\max\{\ell,K(\rho_1),\ldots,K(\rho_\ell)\}.
\]
Then $\widehat K_\ell\to\infty$ and
$\widehat K_\ell\ge K(\rho_\ell)$. Define
\[
m_j^k:=
\begin{cases}
m_j^0, & k<\widehat K_1,\\
m_j^{\rho_{\ell(k)}}, & k\ge\widehat K_1,
\end{cases}
\qquad
\ell(k):=\max\{\ell:\widehat K_\ell\le k\}.
\]
The maximum defining $\ell(k)$ is finite, and
$\ell(k)\to\infty$. For $k<\widehat K_1$, the uniform seed property
gives $m_j^k=m_j^0\in\mathcal F_j(m_{-j}^k)$. For
$k\ge\widehat K_1$, we have
$k\ge\widehat K_{\ell(k)}\ge K(\rho_{\ell(k)})$, and therefore
$m_j^k=m_j^{\rho_{\ell(k)}}\in\mathcal F_j(m_{-j}^k)$. Finally,
$\rho_{\ell(k)}\downarrow0$, so $m_j^k\to m_j$. This proves lower
hemicontinuity and completes the proof.
\end{proof}

\subsection*{Step 5: the principal's payoff, and the theorem}

\begin{lemma}[Principal payoff]\label{l:payoff-contest}
$V_j(m_j,m_{-j}):=\mathbb E_t\bigl[\bar w_j(t;q_j,m_{-j})\bigr]
-\mathbb E_t[\bar r_j(t;m_j)]$ is continuous on
$S_j\times\mathcal M_{-j}$ and concave in $m_j$.
\end{lemma}

\begin{proof}
Continuity: pointwise convergence established in
Lemma~\ref{l:budget-contest}'s proof holds Lebesgue-a.e., hence
$H_j$-a.e.\ ($H_j\ll$ Lebesgue), so dominated convergence gives
convergence of $\mathbb E_t[\cdot]$. Concavity: as in
Lemma~\ref{l:budget-contest}, $\mathbb E_t[\bar w_j(t;\cdot,m_{-j})]$
is concave in $m_j$ and $\mathbb E_t[\bar r_j(t;\cdot)]$ affine, so
$V_j$ is concave.
\end{proof}

Concavity, not affineness, is what Berge's and Kakutani's theorems
need: it gives the best-response correspondence convex values without
requiring the payoff to be linear in the own mechanism, and without
any convexifying device.

\begin{repeatnamedtheorem}[Exact equilibrium in the one-member contest]
{thm:exact-contest-main}
Under the environment above ($n=1$, $N\ge2$),
Assumption~\ref{a:technology}, and Assumption~\ref{a:participate-contest},
there exists $m^*\in\mathcal M=\prod_jS_j$ such that each $m_j^*$
maximizes $V_j(\cdot,m_{-j}^*)$ over $\mathcal F_j(m_{-j}^*)$: an exact
Bayesian--Nash equilibrium of the embedded primitive mechanism game
with participation, including against all stochastic kernel deviations
(Proposition~\ref{p:lifting}), implemented by the win-contingent
proportional-share schedule of Lemma~\ref{l:delivering-schedule}.
\end{repeatnamedtheorem}

\begin{proof}
By Lemma~\ref{l:ic-cont-contest}, $\mathcal F_j$ is continuous with nonempty
compact convex values; by Lemma~\ref{l:payoff-contest}, $V_j$ is
continuous and concave in $m_j$. Berge's maximum theorem
\citep[Theorem~17.31]{aliprantis2006infinite} gives an
upper hemicontinuous best-response correspondence $\beta_j(m_{-j})=
\arg\max_{\mathcal F_j(m_{-j})}V_j(\cdot,m_{-j})$ with nonempty compact convex
values. By Lemma~\ref{l:compact-contest},
$\mathcal M=\prod_jS_j$ is a compact convex subset of the locally
convex Hausdorff space
$\prod_j(L^1([0,1])\times\mathbb R)$; the product correspondence
$F=\prod_j\beta_j$ is an upper-hemicontinuous self-map of $\mathcal M$
with the same properties, and the Kakutani--Fan--Glicksberg theorem
\citep{fan1952fixed,glicksberg1952further}
gives a fixed point $m^*\in F(m^*)$. Proposition~\ref{p:lifting}
transfers this to the underlying mechanism game.
\end{proof}

\subsection*{Discussion}

Three features drive the exact-existence result. First, strict single
crossing converts incentive compatibility into monotonicity of the
allocation and the envelope formula. Canonical right-continuous
representatives make the resulting monotone-allocation space compact
under the $L^1$ metric while preserving a pointwise interpretation of
truthful reporting. Second, quasi-linearity makes the interim mean
reward sufficient for both the agent's incentives and the principal's
expected payoff. Proposition~\ref{p:lifting} therefore shows that
optimization over the reduced feasible set suffices against all
admissible stochastic output/share kernel deviations. Third, the profitable pooling
mechanism supplies uniform \emph{budget} slack. Mixing toward that seed
makes a boundary mechanism strictly feasible, and continuity in the
opponents then yields lower hemicontinuity of the exact feasible-set
correspondence. No strict incentive slack is required.

The participation constraint is essential: without it, unpaid maximal
production is a dominant mechanism for every principal and existence is
immediate. The remaining conditions are sufficient rather than minimal.
In particular,
the theorem uses quasi-linearity, a Lipschitz-concave contest
technology, and a separately assumed profitable pooling seed. It also
restricts attention to one agent per team. With several agents, the
individual envelope formulas determine several interim reward means
that must be delivered simultaneously from one realized team prize.
An average budget inequality is generally not sufficient for that
joint implementation problem. Extending the theorem to $n>1$ therefore
requires a separate characterization of jointly deliverable interim
rewards, together with continuity of that correspondence as the
opponents' mechanisms vary.

\newpage

\bibliographystyle{plainnat}
\bibliography{references_merged}

\begin{thebibliography}{76}
\providecommand{\natexlab}[1]{#1}
\providecommand{\url}[1]{\texttt{#1}}
\expandafter\ifx\csname urlstyle\endcsname\relax
  \providecommand{\doi}[1]{doi: #1}\else
  \providecommand{\doi}{doi: \begingroup \urlstyle{rm}\Url}\fi

\bibitem[Admati and Perry(1991)]{admati1991joint}
Anat~R. Admati and Motty Perry.
\newblock Joint projects without commitment.
\newblock \emph{The Review of Economic Studies}, 58\penalty0 (2):\penalty0
  259--276, 1991.

\bibitem[Aliprantis and Border(2006)]{aliprantis2006infinite}
Charalambos~D. Aliprantis and Kim~C. Border.
\newblock \emph{Infinite Dimensional Analysis: A Hitchhiker's Guide}.
\newblock Springer, Berlin, Heidelberg, 3rd edition, 2006.

\bibitem[Antón et~al.(2023)Antón, Ederer, Giné, and Schmalz]{anton2023}
Miguel Antón, Florian Ederer, Mireia Giné, and Martin Schmalz.
\newblock Common ownership, competition, and top management incentives.
\newblock \emph{Journal of Political Economy}, 131\penalty0 (5):\penalty0
  1294--1355, 2023.

\bibitem[Attar et~al.(2023)Attar, Campioni, and Piaser]{attar2023equilibrium}
Andrea Attar, Eloisa Campioni, and Gwena{\"e}l Piaser.
\newblock Equilibrium (non-)existence in games with competing principals.
\newblock \emph{Economics Letters}, 222:\penalty0 110932, 2023.
\newblock \doi{10.1016/j.econlet.2022.110932}.
\newblock URL \url{https://doi.org/10.1016/j.econlet.2022.110932}.

\bibitem[Babichenko and Rubinstein(2017)]{babichenko2017communication}
Yakov Babichenko and Aviad Rubinstein.
\newblock Communication complexity of approximate nash equilibria.
\newblock In \emph{Proceedings of the 49th Annual ACM SIGACT Symposium on
  Theory of Computing}, pages 878--889. Association for Computing Machinery,
  2017.
\newblock \doi{10.1145/3055399.3055407}.

\bibitem[Balart et~al.(2016)Balart, Flamand, and
  Troumpounis]{balart2016strategic}
Pau Balart, Sabine Flamand, and Orestis Troumpounis.
\newblock Strategic choice of sharing rules in collective contests.
\newblock \emph{Social Choice and Welfare}, 46:\penalty0 239--262, 2016.

\bibitem[Balder(1988)]{Balder1988}
Erik~J. Balder.
\newblock Generalized equilibrium results for games with incomplete
  information.
\newblock \emph{Mathematics of Operations Research}, 13\penalty0 (2):\penalty0
  265--276, 1988.

\bibitem[Balder(2000)]{balder2021new}
Erik~J. Balder.
\newblock New fundamentals of {Y}oung measure convergence.
\newblock In Alexander Ioffe, Simeon Reich, and Itai Shafrir, editors,
  \emph{Calculus of Variations and Optimal Control}, volume 411 of
  \emph{Chapman \& Hall/CRC Research Notes in Mathematics}, pages 24--48.
  Chapman and Hall/CRC, Boca Raton, 2000.

\bibitem[Barbieri and Malueg(2014)]{barbieri2014group}
Stefano Barbieri and David~A. Malueg.
\newblock Group efforts when performance is determined by the “best shot”.
\newblock \emph{Economic Theory}, 56:\penalty0 333--373, 2014.

\bibitem[Barbieri and Malueg(2016)]{barbieri2016private}
Stefano Barbieri and David~A. Malueg.
\newblock Private-information group contests: Best-shot competition.
\newblock \emph{Games and Economic Behavior}, 98:\penalty0 219--234, 2016.

\bibitem[Barbieri and Topolyan(2021)]{barbieri2021private}
Stefano Barbieri and Iryna Topolyan.
\newblock Private-information group contests with complementarities.
\newblock \emph{Journal of Public Economic Theory}, 23\penalty0 (5):\penalty0
  772--800, 2021.

\bibitem[Barbieri et~al.(2019)Barbieri, Kovenock, Malueg, and
  Topolyan]{barbieri2019group}
Stefano Barbieri, Dan Kovenock, David~A. Malueg, and Iryna Topolyan.
\newblock Group contests with private information and the “weakest link”.
\newblock \emph{Games and Economic Behavior}, 118:\penalty0 382--411, 2019.

\bibitem[Barelli and Meneghel(2013)]{barelli2013note}
Paulo Barelli and Idione Meneghel.
\newblock A note on the equilibrium existence problem in discontinuous games.
\newblock \emph{Econometrica}, 81\penalty0 (2):\penalty0 813--824, 2013.

\bibitem[Bich and Laraki(2017)]{bich2017existence}
Philippe Bich and Rida Laraki.
\newblock On the existence of approximate equilibria and sharing rule solutions
  in discontinuous games.
\newblock \emph{Theoretical Economics}, 12\penalty0 (1):\penalty0 79--108,
  2017.

\bibitem[Bonatti and H{\"o}rner(2011)]{bonatti2011collaborating}
Alessandro Bonatti and Johannes H{\"o}rner.
\newblock Collaborating.
\newblock \emph{American Economic Review}, 101\penalty0 (2):\penalty0 632--663,
  2011.

\bibitem[Border(1985)]{border1985fixed}
Kim~C. Border.
\newblock \emph{Fixed Point Theorems with Applications to Economics and Game
  Theory}.
\newblock Cambridge university press, 1985.

\bibitem[Bowen et~al.(2019)Bowen, Georgiadis, and Lambert]{bowen2019collective}
T.~Renee Bowen, George Georgiadis, and Nicolas~S. Lambert.
\newblock Collective choice in dynamic public good provision.
\newblock \emph{American Economic Journal: Microeconomics}, 11\penalty0
  (1):\penalty0 243--298, 2019.

\bibitem[Brookins and Ryvkin(2016)]{brookins2016equilibrium}
Philip Brookins and Dmitry Ryvkin.
\newblock Equilibrium existence in group contests.
\newblock \emph{Economic Theory Bulletin}, 4\penalty0 (2):\penalty0 265--276,
  2016.

\bibitem[Cando{\u{g}}an et~al.(2025)Cando{\u{g}}an, Korpeoglu, and
  Tang]{candougan2024implications}
S{\i}d{\i}ka~Tun{\c{c}} Cando{\u{g}}an, C.~Gizem Korpeoglu, and Christopher~S.
  Tang.
\newblock Implications of team submissions in open innovation contests.
\newblock \emph{IMA Journal of Management Mathematics}, 36\penalty0
  (1):\penalty0 1--19, 2025.

\bibitem[Carbonell-Nicolau and McLean(2018)]{carbonell2018existence}
Oriol Carbonell-Nicolau and Richard~P. McLean.
\newblock On the existence of {N}ash equilibrium in {B}ayesian games.
\newblock \emph{Mathematics of Operations Research}, 43\penalty0 (1):\penalty0
  100--129, 2018.

\bibitem[Carmona(2009)]{carmona2009existence}
Guilherme Carmona.
\newblock An existence result for discontinuous games.
\newblock \emph{Journal of Economic Theory}, 144\penalty0 (3):\penalty0
  1333--1340, 2009.

\bibitem[Cetemen et~al.(2020)Cetemen, Hwang, and Kaya]{cetemen2020uncertainty}
Doruk Cetemen, Ilwoo Hwang, and Ay{\c{c}}a Kaya.
\newblock Uncertainty-driven cooperation.
\newblock \emph{Theoretical Economics}, 15\penalty0 (3):\penalty0 1023--1058,
  2020.

\bibitem[Che and Yoo(2001)]{che2001optimal}
Yeon-Koo Che and Seung-Weon Yoo.
\newblock Optimal incentives for teams.
\newblock \emph{American Economic Review}, 91\penalty0 (3):\penalty0 525--541,
  2001.

\bibitem[Crutzen et~al.(2020)Crutzen, Flamand, and Sahuguet]{crutzen2020}
Beno{\^\i}t~S.Y. Crutzen, Sabine Flamand, and Nicolas Sahuguet.
\newblock A model of a team contest, with an application to incentives under
  list proportional representation.
\newblock \emph{Journal of Public Economics}, 182:\penalty0 104109, 2020.

\bibitem[Dasgupta and Maskin(1986)]{dasgupta1986existence}
Partha Dasgupta and Eric Maskin.
\newblock The existence of equilibrium in discontinuous economic games, i:
  Theory.
\newblock \emph{The Review of Economic Studies}, 53\penalty0 (1):\penalty0
  1--26, 1986.

\bibitem[Dasgupta and Maskin(2015)]{dasgupta2015debreu}
Partha Dasgupta and Eric Maskin.
\newblock Debreu’s social equilibrium existence theorem.
\newblock \emph{Proceedings of the National Academy of Sciences}, 112\penalty0
  (52):\penalty0 15769--15770, 2015.

\bibitem[D'Aspremont and Jacquemin(1988)]{aspremont88}
Claude D'Aspremont and Alexis Jacquemin.
\newblock Cooperative and noncooperative r \& d in duopoly with spillovers.
\newblock \emph{The American Economic Review}, 78:\penalty0 1133--1137, 1988.

\bibitem[Debreu(1952)]{debreu1952social}
Gerard Debreu.
\newblock A social equilibrium existence theorem.
\newblock \emph{Proceedings of the National Academy of Sciences}, 38\penalty0
  (10):\penalty0 886--893, 1952.

\bibitem[Debreu(1982)]{debreu1982existence}
Gerard Debreu.
\newblock Existence of competitive equilibrium.
\newblock In K.~Arrow and M.~Intrilligator, editors, \emph{Handbook of
  Mathematical Economics}, volume~2. North-Holland, New York, 1982.

\bibitem[Dugundji(1951)]{dugundji1951}
James Dugundji.
\newblock An extension of {T}ietze's theorem.
\newblock \emph{Pacific Journal of Mathematics}, 1\penalty0 (3):\penalty0
  353--367, 1951.

\bibitem[Eliaz and Wu(2018)]{eliaz2018simple}
Kfir Eliaz and Qinggong Wu.
\newblock A simple model of competition between teams.
\newblock \emph{Journal of Economic Theory}, 176:\penalty0 372--392, 2018.

\bibitem[Esteban and Ray(2001)]{esteban2001collective}
Joan Esteban and Debraj Ray.
\newblock Collective action and the group size paradox.
\newblock \emph{American Political Science Review}, 95\penalty0 (3), 2001.

\bibitem[Facchinei and Kanzow(2010)]{facchinei2010generalized}
Francisco Facchinei and Christian Kanzow.
\newblock Generalized nash equilibrium problems.
\newblock \emph{Annals of Operations Research}, 175\penalty0 (1):\penalty0
  177--211, 2010.

\bibitem[Fan(1952)]{fan1952fixed}
Ky~Fan.
\newblock Fixed-point and minimax theorems in locally convex topological linear
  spaces.
\newblock \emph{Proceedings of the National Academy of Sciences}, 38\penalty0
  (2):\penalty0 121--126, 1952.

\bibitem[Forges et~al.(2024)Forges, Koessler, and
  Salamanca]{forges2024interacting}
Fran{\c{c}}oise Forges, Fr{\'e}d{\'e}ric Koessler, and Andr{\'e}s Salamanca.
\newblock Interacting mechanisms: A perspective on generalized principal--agent
  problems.
\newblock \emph{Journal of Mathematical Economics}, 114:\penalty0 103023, 2024.
\newblock \doi{10.1016/j.jmateco.2024.103023}.
\newblock URL \url{https://doi.org/10.1016/j.jmateco.2024.103023}.

\bibitem[Fudenberg and Levine(1986)]{fudenberg1986limit}
Drew Fudenberg and David~K. Levine.
\newblock Limit games and limit equilibria.
\newblock \emph{Journal of Economic Theory}, 38\penalty0 (2):\penalty0
  261--279, 1986.
\newblock \doi{10.1016/0022-0531(86)90118-3}.

\bibitem[Georgiadis(2015)]{georgiadis2015projects}
George Georgiadis.
\newblock Projects and team dynamics.
\newblock \emph{The Review of Economic Studies}, 82\penalty0 (1):\penalty0
  187--218, 2015.

\bibitem[Glicksberg(1952)]{glicksberg1952further}
Irving~L. Glicksberg.
\newblock A further generalization of the {K}akutani fixed point theorem, with
  application to nash equilibrium points.
\newblock \emph{Proceedings of the American Mathematical Society}, 3\penalty0
  (1):\penalty0 170--174, 1952.

\bibitem[Goldberg and Roth(2014)]{goldberg2014query}
Paul~W. Goldberg and Aaron Roth.
\newblock Bounds for the query complexity of approximate equilibria.
\newblock In \emph{Proceedings of the Fifteenth ACM Conference on Economics and
  Computation}. Association for Computing Machinery, 2014.
\newblock \doi{10.1145/2600057.2602845}.

\bibitem[Halac et~al.(2021)Halac, Lipnowski, and Rappoport]{halac2021rank}
Marina Halac, Elliot Lipnowski, and Daniel Rappoport.
\newblock Rank uncertainty in organizations.
\newblock \emph{American Economic Review}, 111\penalty0 (3):\penalty0 757--786,
  2021.

\bibitem[He and Yannelis(2016)]{he2016existence}
Wei He and Nicholas~C. Yannelis.
\newblock Existence of equilibria in discontinuous {B}ayesian games.
\newblock \emph{Journal of Economic Theory}, 162:\penalty0 181--194, 2016.

\bibitem[Holmstrom(1982)]{holmstrom1982moral}
Bengt Holmstrom.
\newblock Moral hazard in teams.
\newblock \emph{The Bell Journal of Economics}, pages 324--340, 1982.

\bibitem[Itoh(1991)]{itoh1991incentives}
Hideshi Itoh.
\newblock Incentives to help in multi-agent situations.
\newblock \emph{Econometrica}, pages 611--636, 1991.

\bibitem[Jackson et~al.(2012)Jackson, Rodriguez-Barraquer, and
  Tan]{jackson2012epsilon}
Matthew~O. Jackson, Tomas Rodriguez-Barraquer, and Xu~Tan.
\newblock Epsilon-equilibria of perturbed games.
\newblock \emph{Games and Economic Behavior}, 75\penalty0 (1):\penalty0
  198--216, 2012.
\newblock \doi{10.1016/j.geb.2011.09.007}.

\bibitem[Kadan et~al.(2017)Kadan, Reny, and Swinkels]{kadan2017}
Ohad Kadan, Philip~J. Reny, and Jeroen~M. Swinkels.
\newblock Existence of optimal mechanisms in principal-agent problems.
\newblock \emph{Econometrica}, 85\penalty0 (3):\penalty0 769--823, 2017.

\bibitem[Kajii and Morris(1997)]{kajii1998robustness}
Atsushi Kajii and Stephen Morris.
\newblock The robustness of equilibria to incomplete information.
\newblock \emph{Econometrica}, 65\penalty0 (6):\penalty0 1283--1309, 1997.
\newblock \doi{10.2307/2171737}.

\bibitem[Kamien et~al.(1992)Kamien, Muller, and Zang]{kamien92}
Morton~I. Kamien, Eitan Muller, and Israel Zang.
\newblock Research joint ventures and r\&d cartels.
\newblock \emph{The American Economic Review}, 82:\penalty0 1293--1306, 1992.

\bibitem[Kobayashi and Konishi(2021)]{kobayashi2021effort}
Katsuya Kobayashi and Hideo Konishi.
\newblock Effort complementarity and sharing rules in group contests.
\newblock \emph{Social Choice and Welfare}, 56\penalty0 (2):\penalty0 205--221,
  2021.

\bibitem[Konishi et~al.(2024)Konishi, Sahuguet, and
  Crutzen]{konishi2024allocation}
Hideo Konishi, Nicolas Sahuguet, and Beno{\^\i}t~S.Y. Crutzen.
\newblock Allocation rules of indivisible prizes in team contests.
\newblock \emph{Economic Theory}, 78\penalty0 (1):\penalty0 69--100, 2024.

\bibitem[McAfee and McMillan(1991)]{mcafee1991optimal}
R.~Preston McAfee and John McMillan.
\newblock Optimal contracts for teams.
\newblock \emph{International Economic Review}, 32\penalty0 (3):\penalty0
  561--577, 1991.

\bibitem[McLennan et~al.(2011)McLennan, Monteiro, and
  Tourky]{mclennan2011games}
Andrew McLennan, Paulo~K. Monteiro, and Rabee Tourky.
\newblock Games with discontinuous payoffs: a strengthening of {R}eny's
  existence theorem.
\newblock \emph{Econometrica}, 79\penalty0 (5):\penalty0 1643--1664, 2011.

\bibitem[Milgrom and Segal(2002)]{milgromsegal2002}
Paul Milgrom and Ilya Segal.
\newblock Envelope theorems for arbitrary choice sets.
\newblock \emph{Econometrica}, 70\penalty0 (2):\penalty0 583--601, 2002.
\newblock \doi{10.1111/1468-0262.00296}.

\bibitem[Milgrom and Weber(1985)]{milgrom1985distributional}
Paul~R. Milgrom and Robert~J. Weber.
\newblock Distributional strategies for games with incomplete information.
\newblock \emph{Mathematics of Operations Research}, 10\penalty0 (4):\penalty0
  619--632, 1985.

\bibitem[Monderer and Samet(1989)]{monderer1989approximating}
Dov Monderer and Dov Samet.
\newblock Approximating common knowledge with common beliefs.
\newblock \emph{Games and Economic Behavior}, 1\penalty0 (2):\penalty0
  170--190, 1989.
\newblock \doi{10.1016/0899-8256(89)90017-1}.

\bibitem[Myerson(1982)]{myerson1982optimal}
Roger~B. Myerson.
\newblock Optimal coordination mechanisms in generalized principal--agent
  problems.
\newblock \emph{Journal of Mathematical Economics}, 10\penalty0 (1):\penalty0
  67--81, 1982.

\bibitem[Nitzan(1991)]{nitzan1991collective}
Shmuel Nitzan.
\newblock Collective rent dissipation.
\newblock \emph{The Economic Journal}, 101\penalty0 (409):\penalty0 1522--1534,
  1991.

\bibitem[Nitzan and Ueda(2011)]{nitzan2011prize}
Shmuel Nitzan and Kaoru Ueda.
\newblock Prize sharing in collective contests.
\newblock \emph{European Economic Review}, 55\penalty0 (5):\penalty0 678--687,
  2011.

\bibitem[Nitzan and Ueda(2014)]{nitzan2018}
Shmuel Nitzan and Kaoru Ueda.
\newblock Intra-group heterogeneity in collective contests.
\newblock \emph{Social Choice and Welfare}, 43:\penalty0 219--238, 2014.

\bibitem[Nitzan and Ueda(2018)]{nitzan2014intra}
Shmuel Nitzan and Kaoru Ueda.
\newblock Selective incentives and intragroup heterogeneity in collective
  contests.
\newblock \emph{Journal of Public Economic Theory}, 20:\penalty0 477--498,
  2018.

\bibitem[Olszewski and Siegel(2023)]{olszewski2023equilibrium}
Wojciech Olszewski and Ron Siegel.
\newblock Equilibrium existence in games with ties.
\newblock \emph{Theoretical Economics}, 18\penalty0 (2):\penalty0 481--502,
  2023.

\bibitem[Ozerturk and Yildirim(2021)]{ozerturk2021credit}
Saltuk Ozerturk and Huseyin Yildirim.
\newblock Credit attribution and collaborative work.
\newblock \emph{Journal of Economic Theory}, 195:\penalty0 105264, 2021.

\bibitem[Prokopovych and Yannelis(2023)]{prokopovych2023monotone}
Pavlo Prokopovych and Nicholas~C. Yannelis.
\newblock On monotone pure-strategy {B}ayesian-{N}ash equilibria of a
  generalized contest.
\newblock \emph{Games and Economic Behavior}, 140:\penalty0 348--362, 2023.

\bibitem[Radner(1980)]{radner1980collusive}
Roy Radner.
\newblock Collusive behavior in noncooperative epsilon-equilibria of
  oligopolies with long but finite lives.
\newblock \emph{Journal of Economic Theory}, 22\penalty0 (2):\penalty0
  136--154, 1980.
\newblock \doi{10.1016/0022-0531(80)90037-X}.

\bibitem[Raith(2003)]{raith03}
Michael Raith.
\newblock Competition, risk, and managerial incentives.
\newblock \emph{The American Economic Review}, 93:\penalty0 1425--1436, 2003.

\bibitem[Reny(1999)]{reny1999existence}
Philip~J. Reny.
\newblock On the existence of pure and mixed strategy {N}ash equilibria in
  discontinuous games.
\newblock \emph{Econometrica}, 67\penalty0 (5):\penalty0 1029--1056, 1999.

\bibitem[Reny(2020)]{reny2020nash}
Philip~J. Reny.
\newblock Nash equilibrium in discontinuous games.
\newblock \emph{Annual Review of Economics}, 12\penalty0 (1):\penalty0
  439--470, 2020.

\bibitem[Saldi and Y{\"u}ksel(2026)]{saldi2026kernel}
Naci Saldi and Serdar Y{\"u}ksel.
\newblock Kernel mean embedding topology: Weak and strong forms for stochastic
  kernels and implications for model learning.
\newblock \emph{The Annals of Applied Probability}, 36\penalty0 (2):\penalty0
  1653--1689, 2026.
\newblock \doi{10.1214/25-AAP2266}.

\bibitem[Scheinberg(2026)]{Schein26}
Katya Scheinberg.
\newblock Editor's comments on the $50^{th}$ anniversary of \textit{Mathematics
  of Operations Research}.
\newblock \emph{Mathematics of Operations Research}, 51:\penalty0 iv--vii,
  2026.

\bibitem[Simeonov(2020)]{simeonov2020individual}
Dimitar Simeonov.
\newblock \emph{Individual and team incentives in contests with heterogeneous
  agents}.
\newblock PhD thesis, Boston College, 2020.

\bibitem[Sinander(2022)]{sinander2022}
Ludvig Sinander.
\newblock The converse envelope theorem.
\newblock \emph{Econometrica}, 90\penalty0 (6):\penalty0 2795--2819, 2022.
\newblock \doi{10.3982/ECTA18119}.

\bibitem[Tao and Xu(2025)]{tao2024generalized}
Yuan Tao and Huifu Xu.
\newblock Generalized {B}ayesian {N}ash equilibrium with continuous type and
  action spaces.
\newblock \emph{SIAM Journal on Optimization}, 35\penalty0 (2):\penalty0
  789--817, 2025.

\bibitem[T{\'o}bi{\'a}s(2022)]{tobias2022equilibrium}
{\'A}ron T{\'o}bi{\'a}s.
\newblock Equilibrium non-existence in generalized games.
\newblock \emph{Games and Economic Behavior}, 135:\penalty0 327--337, 2022.

\bibitem[Trevisan(2020)]{trevisan2020optimal}
Francesco Trevisan.
\newblock Optimal prize allocations in group contests.
\newblock \emph{Social Choice and Welfare}, 55:\penalty0 431--451, 2020.

\bibitem[Winter(2004)]{winter2004incentives}
Eyal Winter.
\newblock Incentives and discrimination.
\newblock \emph{American Economic Review}, 94\penalty0 (3):\penalty0 764--773,
  2004.

\bibitem[Yildirim(2006)]{yildirim2006getting}
Huseyin Yildirim.
\newblock Getting the ball rolling: Voluntary contributions to a large-scale
  public project.
\newblock \emph{Journal of Public Economic Theory}, 8\penalty0 (4):\penalty0
  503--528, 2006.

\bibitem[Yildirim(2023)]{yildirim2023fares}
Huseyin Yildirim.
\newblock Who fares better in teamwork?
\newblock \emph{The RAND Journal of Economics}, 54\penalty0 (2):\penalty0
  299--324, 2023.

\end{thebibliography}

\end{document}